%% file: paper.tex
\documentclass[11pt]{article}
\pdfoutput=1
\usepackage[letterpaper,margin=1in]{geometry}
\usepackage[USenglish]{babel}
\usepackage[T1]{fontenc}
\usepackage{lmodern}
\usepackage{microtype}
\usepackage{amsmath,amssymb,amsthm,mathtools}
\usepackage{graphicx,booktabs,tabularx,array,placeins,needspace}
\usepackage{xcolor}
\usepackage[hidelinks,hypertexnames=false]{hyperref}
\newtheorem{theoremA}{Theorem}
\newtheorem{theoremB}{Theorem}
\newtheorem{theoremBintro}{Theorem}
\newtheorem{theoremAintro}{Theorem}
\newtheorem{theoremC}{Theorem}
\newtheorem{theoremD}{Theorem}
\newtheorem{theoremE}{Theorem}
\newtheorem{theorem}{Theorem}[section]
\newtheorem{conjecture}[theorem]{Conjecture}
\usepackage{mathrsfs}
\newtheorem{lemma}[theorem]{Lemma}
\newtheorem{proposition}[theorem]{Proposition}
\newtheorem{corollary}{Corollary}
\newtheorem{corollaryB}{Corollary}
\title{\fontsize{16}{19}\selectfont The Exact Welfare Guarantee of Fixed-Price Bilateral Trade}
\author{%
\begin{tabular}[t]{@{}c@{\hspace{4em}}c@{}}
Tingyi Lin\setcounter{footnote}{1}\thanks{Correspondence to \texttt{tingyi3@illinois.edu}.} & Yichen Shi\\
\small Adrasteia Labs and UIUC & \small Zhongnan University of Economics and Law\\[1.6ex]
Ke Dong & Jiazhuo Li\\
\small Tsinghua University & \small University of Michigan, Ann Arbor\\[1.6ex]
Shawn Yu & Huanxi Zhang\\
\small Boston College & \small University of Wisconsin--Madison
\end{tabular}}
\date{}
\begin{document}
\maketitle
\begin{abstract}
A seller and a buyer with independent private values can trade only at a posted price. We determine the worst case of this mechanism exactly: the best posted price always guarantees a \(\beta_*=0.738024\ldots\) fraction of first-best welfare, where \(\beta_*\) is given in closed form by the root of an explicit equation, the worst-case buyer is unique up to scaling, and no pair of distributions attains the worst case. This closes the gap \([0.7292,0.73805]\) left by a line of work from SODA 2016 through two STOC 2023 papers and an AAAI 2026 follow-up. The proof is an explicit certificate of optimality: after one change of variables, the gap between the optimal value and the value of any buyer distribution is a sum of manifestly nonnegative integrals, as in a linear-programming dual, and equality identifies the worst-case shape.

The certificate also gives the complete tradeoff between gains from trade and the seller's initial welfare: when first-best gains from trade are a fraction \(\kappa\) of initial seller welfare, the exact worst-case fraction \(\rho(\kappa)\) of gains obtained by the best price satisfies \(\rho(\kappa)\sim2/\log(1/\kappa)\) as \(\kappa\to0\). The worst-case buyer's survival function has two constant segments joined by an explicit nonexponential curve, and the worst case is approached through a vanishing atom whose value tends to infinity. The same constant is the exact guarantee of dominant-strategy mechanisms with individual rationality and strong budget balance in every realization.

For two units with increasing submodular valuations, an explicit finite instance has ratio below \(0.7290804\), so multi-unit trade is strictly harder than single-unit trade. Fixing the buyer, we characterize the least probability that a random common price needs to guarantee a given ratio against every seller; bounding this quantity over all buyers would determine the exact two-unit constant. Within an explicit family the worst ratio is \(0.729080\ldots\), which we conjecture to be the two-unit constant.
\end{abstract}

\clearpage
\setcounter{tocdepth}{1}
{\small\tableofcontents}
\clearpage

\section{Introduction}\label{sec:intro}

A seller owns an item that a buyer may want. Myerson and Satterthwaite \cite{ms83} showed that no mechanism can be at once efficient, incentive compatible, voluntary, and budget balanced. The simplest response is to post a price and let the two parties take it or leave it. How much welfare does this simplicity cost, in the worst case over all pairs of independent value distributions? This paper gives the exact answer, \(\beta_*=0.738024\ldots\), and describes the distributions that come arbitrarily close to forcing it. Two things are needed: a price that works for every pair, and a description of the pairs that defeat every price.

No pair of distributions attains \(\beta_*\). It is approached by pairs in which the buyer's value has an explicit density on an interval, no mass below that interval, and a vanishing atom far above it, while the seller's value has an atom at zero, an explicit distribution on the same interval, and an atom above it (Figure~\ref{fig:extremal}b,c). The shape carries more than the number: the same family of pairs, indexed by one parameter, gives every sharp guarantee on gains from trade and the exact rate at which such guarantees fail as gains become small relative to the seller's initial welfare.

The upper bound came first, and it came close. Colini-Baldeschi, de Keijzer, Leonardi, and Turchetta \cite{cbklt16} gave a pair of distributions on which no fixed price exceeds \(0.7485\), and Kang and Vondr\'{a}k \cite{kv19} lowered this to \(0.7385\). Cai and Wu \cite{cw23} and Liu, Ren, and Wang \cite{lrw23} developed complementary optimization approaches at STOC 2023, each with hard instances at \(0.7381\). Giambartolomei and de Keijzer \cite{gdk26} subsequently narrowed the welfare ratio to \([0.7292,0.73805]\).

We close the gap by proving a matching lower bound for every pair of distributions and identifying the worst-case pair exactly. Numerical programs bound the constant but cannot decide whether it is attained, what the worst case looks like, or how the guarantee changes with the instance; those are the questions answered here. The exact value \(0.738024\ldots\) raises the previous lower guarantee by about \(0.009\); the reported upper bound was already within \(3\times10^{-5}\) of this value. Our results were obtained concurrently with and independently of the work of Tao Jiang, Minbo Gao, and Shaowei Cai \cite{jgc26}, who report the same exact single-unit welfare ratio; they do not study gains-from-trade guarantees, the dominant-strategy welfare ceiling, or multiple units. They reach the same variational problem, in other coordinates, by saturating a constraint on the seller, and prove optimality by showing that a maximizer exists and solving its first-order conditions; our certificate compares every competitor directly and assumes no maximizer. Their saturated seller is the worst-case seller of \cite[Theorem 3.5]{lrw23}, and their threshold equivalence \cite[Proposition 1.2]{jgc26} corresponds to the criterion of \cite[Theorem 3.3]{lrw23}; they cite \cite{lrw23} only for its numerical bounds. Table~\ref{tab:contributions} summarizes what was known and what is new.

\begin{table}[!htbp]
\centering
\fontsize{9}{10.4}\selectfont
\setlength{\tabcolsep}{4pt}
\renewcommand{\arraystretch}{1.1}
\begin{tabularx}{\linewidth}{@{}>{\raggedright\arraybackslash}p{.17\linewidth}>{\raggedright\arraybackslash}X>{\raggedright\arraybackslash}X@{}}
\toprule
Question & Previously known & This paper \\
\midrule
Worst-case welfare ratio & The interval \([0.7292,0.73805]\): hard instances at \(0.7485\), \(0.7385\), \(0.7381\), and \(0.73805\) \cite{cbklt16,kv19,cw23,lrw23,gdk26}, and lower bounds from finite programs \cite{cw23,gdk26}. & The exact value \(\beta_*=0.738024\ldots\), the root of an explicit equation (Theorem~\ref{thm:exact}). \\
Worst-case pair & Necessary conditions on an unknown interval; an exponential trial family for upper bounds \cite{lrw23}. & A unique explicit shape with a nonexponential middle arc, never attained and approached through an escaping atom (Theorem~\ref{thm:calibration}, Proposition~\ref{prop:nonexponential}). \\
Method & Finite programs \cite{cw23,gdk26}; first-order conditions with an explicit trial family \cite{lrw23}. & A global optimality certificate for a nonconcave control problem: a concavifying change of variables (the reciprocal clock) and a calibration (a sufficiency certificate from the calculus of variations) with obstacle multipliers (Lemmas~\ref{lem:energy} and~\ref{lem:gaps}, Theorem~\ref{thm:calibration}). \\
Gains-from-trade guarantees & No constant fraction \cite{bd21}; logarithmic bounds in the efficient-trade probability \cite{cbg17}. & The exact affine frontier and the exact conditional guarantee \(\rho(\kappa)\) at each ratio \(\kappa=G/M\), with \(\rho(\kappa)\sim2/\log(1/\kappa)\) (Theorem~\ref{thm:frontier}, Corollary~\ref{cor:asymptotics}). \\
Dominant-strategy mechanisms & The randomized-price representation \cite{hr87,cp16}. & The exact welfare ceiling \(\beta_*\) under realizationwise individual rationality and strong budget balance (Corollary~\ref{cor:dsic}). \\
Two units & A numerically computed upper bound of \(0.7291\) and a separation conjecture \cite{gdk26}. & A finite instance below \(0.7290804\) and the separation theorem; a one-sided reduction of the exact constant and a conjectured value \(0.729080\ldots\) (Theorems~\ref{thm:two-separation} and~\ref{thm:two-pricing}, Conjecture~\ref{conj:two-exact}). \\
\bottomrule
\end{tabularx}
\caption{What was known and what is new.}
\label{tab:contributions}
\end{table}

Our two main results are the exact welfare guarantee in Theorem~\ref{thm:exact} and the complete gains-from-trade frontier in Theorem~\ref{thm:frontier}. The global comparison in Theorem~\ref{thm:calibration} underlies both results.

\Needspace{8\baselineskip}
\begin{theoremAintro}[Exact welfare guarantee, informal]\label{thm:exact-intro}
Let the seller's and the buyer's values be independent and nonnegative, with positive and finite first-best welfare. The best posted price always achieves at least a \(\beta_*=0.738024\ldots\) fraction of first-best welfare, and no larger constant is guaranteed. Here \(\beta_*\) is given in closed form by the unique root of an explicit equation in one variable. Every pair of distributions achieves strictly more than \(\beta_*\), and an explicit family of bounded pairs approaches it.
\end{theoremAintro}

The equation is derived in Section~\ref{sec:overview} and stated in full in Theorem~\ref{thm:exact}; the family is in Lemma~\ref{lem:hard}.

In words: a posted price never loses more than \(26.2\) percent of the available welfare. The adversary who comes closest to this loss pairs a buyer who is ordinary with probability close to one, and a lottery ticket with the remaining sliver of probability, with a seller whose costs spread over the same range. A low price serves the ordinary buyer but shuts out every seller whose cost exceeds it; a high price lets every seller trade with the lottery ticket but loses the ordinary buyer. No single price serves both, and \(\beta_*\) is the exact price of that dilemma.

\Needspace{9\baselineskip}
\begin{theoremBintro}[Sharp gains-from-trade guarantees, informal]\label{thm:frontier-intro}
For independent nonnegative values with finite first-best welfare, the sharpest lower bounds on gains from trade that are linear in first-best gains from trade and in the seller's initial welfare form an explicit strictly convex frontier. At each prescribed ratio \(\kappa>0\) of first-best gains from trade to initial seller welfare, this frontier determines the exact fraction \(\rho(\kappa)\) of first-best gains from trade guaranteed by the best fixed price, with
\[
\rho(\kappa)\sim\frac{2}{\log(1/\kappa)}\qquad(\kappa\downarrow0).
\]
Bounded independent hard pairs approach the guarantee at every prescribed \(\kappa\).
\end{theoremBintro}

In words: when the potential gains from trade equal what the seller already has, the best price secures \(54\) percent of them; when they are a tenth of it, \(36\) percent; a hundredth, \(26\) percent. The welfare worst case of Theorem~\ref{thm:exact-intro} sits at a ratio near \(2.6\), where the best price secures \(64\) percent of the gains. The logarithmic rate is reached only slowly; at a ratio of one hundredth, \(2/\log(1/\kappa)\) is still \(0.43\).

\paragraph{The complete gains-from-trade frontier.}
Welfare counts the seller's initial value; gains from trade do not. A posted price guarantees a constant fraction of welfare but no constant fraction of gains from trade \cite{bd21}, and the second main result says exactly how the guarantee on gains from trade decays. Write \(M\) for initial seller welfare and \(G\) for first-best gains from trade (GFT). Theorem~\ref{thm:frontier} gives the smallest penalty \(\delta(\beta)\) for which some price always obtains GFT at least \(\beta G-\delta(\beta)M\). This strictly convex curve meets \(\delta=1-\beta\) exactly at the welfare constant, because a price with welfare at least \(\beta(M+G)\) is exactly a price with GFT at least \(\beta G-(1-\beta)M\). For a fixed ratio \(\kappa=G/M>0\), let \(\rho(\kappa)\) denote the exact worst-case fraction of GFT obtained by the best price. Theorem~\ref{thm:frontier} gives
\[
\kappa\rho(\kappa)=\sup_{0<\beta<1}\{\kappa\beta-\delta(\beta)\}.
\]
Thus the known failure of a constant GFT approximation has an exact logarithmic scale when parameterized by \(G/M\); the asymptotic in Theorem~\ref{thm:frontier-intro} is proved in Corollary~\ref{cor:asymptotics}. The coefficient \(2\) is sharp: for each prescribed \(G/M\), bounded hard pairs approach the corresponding conditional guarantee. Both results are read off one explicit curve, indexed by a single parameter, which Section~\ref{sec:model} computes in closed form.

\paragraph{Why is this hard?}
The worst case is the value of an optimization problem over buyer distributions that is not concave, so first-order conditions cannot certify a candidate, and its supremum is not attained. The limiting buyer atom has positive contribution to expected value even as its probability vanishes. This loss of tail mass under weak convergence prevents a genuine distribution pair from attaining the worst-case welfare ratio. Identifying the value exactly requires controlling this limiting tail and proving a matching bound for every distribution. The upper bound of \cite[Theorem 3.8]{lrw23} comes from a buyer whose probability of valuing the item above \(s\) is of the form \(\lambda_0+\lambda_1e^{\lambda_2s}\) between two constant segments. The exact worst-case buyer also has two constant segments, but its middle segment is not of this form on any subinterval (Proposition~\ref{prop:nonexponential}), so no member of that family reaches the exact constant. Theorem~\ref{thm:calibration} supplies the global comparison and identifies the unique maximizing control.

Cai and Wu \cite[Lemmas 3.1 and 3.4]{cw23} characterize the exact ratio as the value of an infinite-dimensional min-max program and bound it through two discretizations that converge to it; Theorem~\ref{thm:exact} determines that value analytically. Liu, Ren, and Wang supply three of our ingredients: a price density that depends on the buyer, the criterion that reduces the worst case to maximizing a functional of the buyer, and a seller distribution that makes every price equally good \cite[Theorems 3.1, 3.3, and 3.5]{lrw23}. Their Theorem 3.4 gives the first-order equation that a maximizing buyer satisfies between the point where its survival function leaves one and the point where it reaches zero. We locate these two points, integrate the equation between them, and prove the global comparison with a unique maximizing control.

\paragraph{Here is the key idea.}
In a linear program, a dual solution certifies that a candidate is optimal by writing the gap between the candidate and any feasible point as a sum of nonnegative terms. Our problem has no finite-dimensional program and is not concave (Section~\ref{sec:overview} gives a two-point example), so the first-order conditions of \cite{lrw23} are necessary but not sufficient, and no trial family can be certified from them. Our proof imports the sufficiency tool of the calculus of variations, calibration \cite{abd03,bf16}, and obtains a certificate of this kind. One change of variables, which we call the reciprocal clock, measures the buyer by the reciprocal of its tail integral and measures time by how much of the second-moment penalty is still to come; in these coordinates the problem is concave once a single boundary value, the endpoint, is fixed. For every endpoint we write down a reference curve and a multiplier that is positive only where the constraints bind. Any competitor with the same endpoint is worse than the reference curve by a sum of three nonnegative integrals, and the value of the reference curve is an explicit function of the endpoint with a unique maximum. Equality forces the competitor to coincide with the reference curve at the best endpoint, which is the worst-case shape. The certificate assumes neither existence nor regularity of a maximizer. We need this generality: no genuine distribution pair attains the worst-case ratio, and the relaxed control problem has a maximizer only as a limit of distributions.

Two consequences place the result in the mechanism-design literature. By Cai and Wu's Theorem 4.1 \cite{cw23}, \(\beta_*\) is also the exact universal ratio when only the buyer's complete prior or only the seller's complete prior is known and randomized prices are allowed. Among all mechanisms for risk-neutral traders that are truthful in dominant strategies, balance the budget in every realization, and never make a trader worse off in any realization,\footnote{That is, ex post: for every realized pair of values, not only in expectation.} none guarantees more than \(\beta_*\) (Corollary~\ref{cor:dsic}). Up to null sets every such mechanism is a random posted price \cite{hr87,cp16}, so on smooth priors it does no better than the best fixed price, and smoothing carries this bound to the hard pairs. Simplicity costs nothing here: among mechanisms that are robust in this sense, a single posted price is worst-case optimal. It is not a ceiling for Bayesian incentive-compatible trade: Dobzinski and Shaulker \cite{ds26} obtain \(0.746\) with a mechanism that is DSIC for the seller and Bayesian IC for the buyer.

\paragraph{Extension to two units.}
When the seller owns two units and both traders have decreasing marginal values, one posted price must serve both marginal trades. We give an explicit finite instance on which no posted price reaches ratio \(0.7290804\), so two units are strictly harder than one, and an instance below \(0.83693\) in which the buyer's and the seller's valuations are identically distributed, improving the bound \(0.8372\) of \cite{gdk26} for that case (Theorem~\ref{thm:two-separation}). With one unit and identically distributed values the exact ratio is \((2+\sqrt2)/4\approx0.854\) \cite{kpv22}, so the second unit is costly in that case too. The first instance, evaluated in exact rational arithmetic, settles the separation conjecture of \cite{gdk26}. We then reduce the exact two-unit constant to a question about buyers alone. As in the single-unit argument of Section~\ref{sec:overview}, we fix the buyer and ask how much probability a random price needs to guarantee a target ratio against every seller. The seller's two costs are ordered, and this order couples the two units. A monotone compensation between the units removes the coupling; for each compensation the least price mass is the fixed point of a contraction, and we then minimize over compensations (Theorem~\ref{thm:two-pricing}). Within an explicit family of instances the worst ratio is \(0.729080\ldots\), the unique solution of an explicit system of equations, with rigorous bounds (Proposition~\ref{prop:2fam-optimum}). We conjecture that it is the two-unit constant (Conjecture~\ref{conj:two-exact}); what remains is the bound \eqref{eq:two-open-budget} on the least price mass over all buyers. A second unit therefore costs the designer something: the worst-case guarantee drops by at least \(0.0089\), and if the conjecture holds, by exactly \(\beta_*-0.729080\ldots\).

\paragraph{Formal verification.}
Theorems~\ref{thm:exact}--\ref{thm:two-pricing}, Corollaries~\ref{cor:dsic} and~\ref{cor:asymptotics}, and Proposition~\ref{prop:2fam-optimum} have machine-checked proofs in Lean~4 with Mathlib, with the decimal value of \(\beta_*\) checked by interval arithmetic; the formal proofs of Theorem~\ref{thm:two-pricing} and Proposition~\ref{prop:2fam-optimum} assume scalar facts certified by exact, interval, and SMT arithmetic, and the \href{https://drive.google.com/file/d/1bG8ySCTCiVD0MK4EAhtyTkSXpvBY_4NC/view?usp=drivesdk}{\textcolor{blue}{supplementary material}} lists what is formalized and what is not.

\section{Technical overview}\label{sec:overview}

\begin{figure}[t]
\centering
\includegraphics[width=\linewidth]{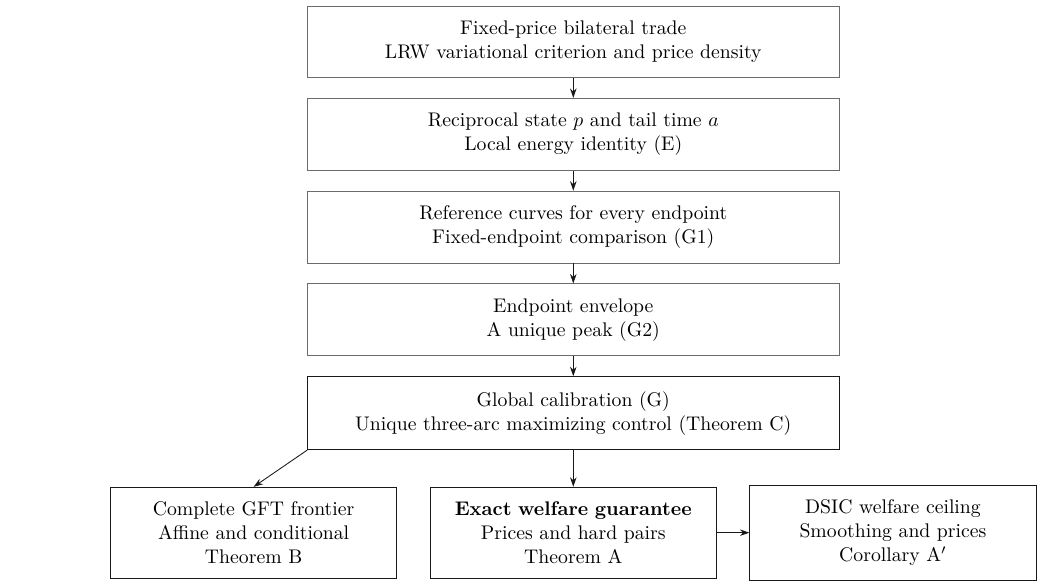}
\caption{The single-unit proof of the guarantee runs down the central chain. The reference curve at the maximizing endpoint and bounded hard pairs give sharpness; varying the parameter gives the GFT frontier. The DSIC consequence additionally uses the price representation and smoothing.}
\label{fig:proofmap}
\end{figure}

\paragraph{From a price distribution to a functional of the buyer.}
Fix the buyer's distribution and a target ratio \(\beta\). A random price certifies the guarantee \(\beta\) against every seller distribution if, for each seller value \(s\), the expected gains from trade conditional on \(s\) are at least \(\beta\) times the conditional first-best gains minus \(\beta d\) times \(s\), for a constant \(d>0\). Averaging over the seller gives GFT at least \(\beta G-\beta d\cdot M\), which is a welfare guarantee of \(\beta\) when \(\beta d\le1-\beta\). Liu, Ren, and Wang \cite[Theorem 3.1]{lrw23} give a price density that meets this target with equality, and its total mass is \(\beta\) times a functional of the buyer alone. The density is a random price when this mass is at most one, with the remaining probability placed at price zero; because the target is met in expectation over this random price, some fixed price meets it. We rescale values so that the target vanishes at the cutoff \(s=1\); above the cutoff the target is negative and holds trivially, and at the cutoff the buyer's tail integral equals \(d\). Writing \(h\) for the buyer's survival function in reversed and normalized time, \(g\) for its running integral started at \(d\), and \(k\) for the running integral of \(h^2\), that functional is
\[
\mathcal J_d(h)=\int_0^1\left[\frac{h}{g}+\frac{d^2-k}{g^2}\right]\mathrm dt .
\]
The guarantee holds whenever \(\beta\mathcal J_d(h)\le1\) for every buyer, and the equalizing seller of \cite{lrw23} shows that nothing weaker suffices. The worst-case buyer therefore maximizes \(\mathcal J_d\); we allow all measurable \(h\) with values in \([0,1]\), a relaxation, since survival functions give only nondecreasing \(h\). The seller has been eliminated, because the price density is the certificate against every seller (Lemma~\ref{lem:pricing}). Section~\ref{sec:model} records the normalization.

\paragraph{The value curve.}
Everything then rests on the value \(V(d)=\sup_h\mathcal J_d(h)\): the guarantee of GFT at least \(\beta G-\beta d\cdot M\) holds for all pairs exactly when \(\beta V(d)\le1\). Theorem~\ref{thm:calibration} computes this value along an explicit curve. For each \(d>0\) there is a unique \(C_d\in(1/4,1/2)\) with \(\tau(C_d)=d\), and \(V(d)=S(C_d)\), where \(\tau\) and \(S\) are the closed forms of Section~\ref{sec:model}. The frontier of Theorem~\ref{thm:frontier} is therefore the curve \((1/S(C),\tau(C)/S(C))\), and the welfare constant is its point on the line \(\delta=1-\beta\), where \(S(C_*)=1+\tau(C_*)\); numerically \(C_*\approx0.394\) and \(\tau(C_*)\approx0.355\).

\paragraph{The maximization is not concave.}
At \(d=1/2\), the controls equal to \((4/5,3/5)\) and to \((1/5,2/5)\) on the two halves of \([0,1]\) have values \(203/216\) and \(103/96\). Their average is the constant control \(1/2\), whose value is \(1\), yet the average of the two values exceeds \(1\) by \(11/1728\); Appendix~\ref{app:nonconcavity} gives the general two-step formula. No first-order or Jensen argument can therefore certify a maximizer. This is the obstacle; everything below is built to get around it.

\paragraph{The reciprocal clock.}
Each contribution \(h(u)^2\) to the accumulated second moment \(k\) is paid at every later time through the denominator \(g^2\). Collecting these later payments into a single weight defines a new clock, the tail time \(a=\int_t^1g^{-2}\), under which the objective becomes local. We take the reciprocal \(p=1/g\) of the tail integral as the new state. In the new clock the slope of \(p\) is exactly \(h\), so \(h\le1\) gives the upper obstacle \(p(a)\le x+a\), and \(h\ge0\) gives \(g\ge d\), hence the cap \(p\le1/d\); both obstacles are linear in \(p\). The objective becomes \(-\log(d\cdot x)\) plus an energy that is concave in \(\log p\) once the endpoint \(x=p(0)\), the reciprocal of the final tail integral, is fixed (Lemma~\ref{lem:energy}). This is the only step that changes the structure of the problem; everything after it is a comparison.

\paragraph{A reference curve at every endpoint, and a certificate.}
Fixing \(x\) leaves a concave problem in \(\log p\) with two upper obstacles. For every endpoint \(x\) we construct a reference curve \(P_x\) that follows the line \(x+a\), then solves the Euler equation of the energy on a middle arc, then follows the cap \(1/d\) (Lemma~\ref{lem:reference}). Comparing any admissible \(p\) with the same endpoint to \(P_x\), and writing \(v=\log(p/P_x)\), the energy of \(P_x\) minus the energy of \(p\) is
\[
\int \left[
a(v')^2+P_x^{-2}(e^{-2v}-1+2v)-2Q_xv
\right]\,\mathrm da.
\]
The first two terms are nonnegative. The function \(Q_x\) is nonnegative and supported where the reference curve \(P_x\) meets an upper obstacle, so \(v\le0\) there and the last term is nonnegative as well. In the fixed-endpoint concave relaxation, \(2Q_x\) is the obstacle multiplier. Its sign gives the weak-duality comparison familiar from LPs, and its contact-set support expresses complementary slackness. This is the fixed-endpoint gap, identity (G1) of Lemma~\ref{lem:gaps}. A second identity, (G2) of the same lemma, shows that the reference value as a function of the endpoint \(x\) has a unique maximum. Its derivative along the family of reference curves has the sign of a single explicit factor that changes sign exactly once; the nonconcavity of the original problem survives only in this one-dimensional comparison. Any competitor is therefore dominated in two steps, first by the reference curve at its own endpoint and then by the reference curve at the best endpoint. Equality in both steps forces it to coincide with the reference curve at the best endpoint, where the inverse coordinate change yields a control satisfying the original time constraint; this control is the unique maximizer (Theorem~\ref{thm:calibration}). In the language of the calculus of variations, \(P_x\) is a calibration: an exact comparison identity that certifies global optimality, with \(2Q_x\) as the obstacle multiplier.

\paragraph{What the relaxations cost.}
Three constraints were dropped on the way: monotonicity of \(h\); the slope bound \(0\le p'\le1\), kept only through its two integrated obstacles; and the normalization \(\int p^{-2}\,\mathrm da=1\), which says that original time has length one. The maximizing control satisfies all three, so the relaxed value is the true value. It is still not a buyer: it vanishes near \(t=0\), so the survival function is zero just below \(s=1\), while the tail integral at \(s=1\) equals \(d>0\). Only a limit of buyers has both properties, through an atom of mass \(\eta\to0\) placed at \(1+d/\eta\), which is why the worst case is approached and never attained.

\begin{figure}[t]
\centering
\includegraphics[width=\linewidth]{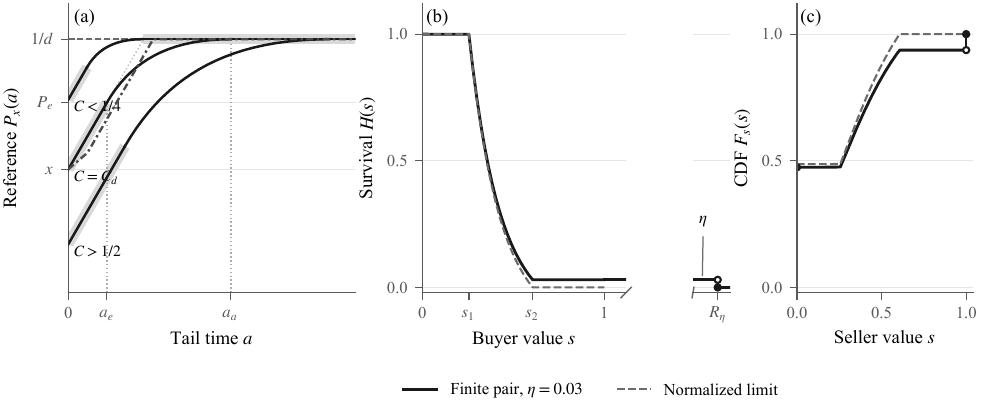}
\caption{Reference curves, and the bounded hard pairs at \(C=C_*\). (a) Reference curves at three endpoints, for \(C<1/4\), \(C=C_d\), and \(C>1/2\), between the line obstacles \(x+a\) and the cap \(1/d\); contact segments are shaded, and \(x\), \(P_e\), \(a_e\), \(a_a\) refer to the middle curve. The dash-dotted curve is a competitor in the obstacle class with the endpoint of the middle curve. (b,c) Solid curves give an actual independent pair with \(\eta=0.03\); dashed curves give the normalized limit. The buyer has an atom of mass \(\eta\) at \(R_\eta=1+d/\eta\), shown after the axis break. The seller of the actual pair has atoms at zero and one; the atom at one vanishes in the limit. The buyer's free boundaries are \(s_1=C^2/(1-C)\) and \(s_2=1-C\).}
\label{fig:extremal}
\end{figure}

\paragraph{The worst-case pair.}
At the optimal endpoint the reference curve follows the line \(x+a\), then a curved arc solving \(P''=-CP/a^2\), where \(C\) is the first integral of the Euler equation, then the cap \(1/d\) (Figure~\ref{fig:extremal}a). Transforming back, the worst-case buyer's survival function equals \(1\) up to \(s_1=C^2/(1-C)\), decreases along an explicit curve to \(0\) at \(s_2=1-C\), and the bounded hard pairs add an atom of mass \(\eta\) at value \(1+d/\eta\), whose contribution \(\eta\cdot d/\eta=d\) to the tail integral at \(s=1\) survives as its probability vanishes (Figure~\ref{fig:extremal}b). The equalizing seller of \cite{lrw23} has atoms at zero and at one, the second vanishing as \(\eta\to0\), and makes every price in \([0,1)\) equally good (Figure~\ref{fig:extremal}c). No genuine pair attains the ratio, and every pair is strictly above \(\beta_*\). The middle arc is not an exponential on any subinterval (Proposition~\ref{prop:nonexponential}). Moving the parameter \(C\) moves the pair along the frontier of Theorem~\ref{thm:frontier}; adding the finite buyer tail gives bounded hard pairs at every point of that frontier, and the price representation and smoothing then give the dominant-strategy welfare ceiling (Corollary~\ref{cor:dsic}).

\paragraph{New and adapted ingredients.}
The change of variables (Lemma~\ref{lem:energy}) and the two gap identities (Lemma~\ref{lem:gaps}) carry the single-unit argument. Lemmas~\ref{lem:pricing} and~\ref{lem:hard} adapt the price density and the equalizing seller of \cite{lrw23} to arbitrary buyer atoms and to bounded supports, and we give their short proofs to fix the conventions. Appendix~\ref{app:representation} proves the price representation of \cite{hr87,cp16} in an almost-everywhere form for measurable mechanisms, which is the form that smoothed hard pairs require. The family optimization for two units (Appendix~\ref{app:2fam}) is computer-assisted; its interval and positivity certificates, with a replay script, are in the supplementary material.

\paragraph{Extension to two units.}
The same architecture applies to two units with one new difficulty. After fixing the buyer's two marginal values, a common price must satisfy one inequality for every ordered pair of seller costs. The order between the two costs couples the two units. A nondecreasing function, the compensation, moves slack between them: if the first unit's slack is at least minus the compensation and the second unit's slack is at least the compensation, their sum is nonnegative for every ordered pair of costs. For each compensation the least feasible price mass is the fixed point of a contraction (Theorem~\ref{thm:two-pricing}). This least mass is convex in the compensation; its minimum is a functional of the buyer, and the supremum of that functional over buyers is the one quantity that remains open (Conjecture~\ref{conj:two-exact}). Within an explicit family the worst ratio is attained at the unique admissible solution of the family's first-order conditions and lies in \((0.72908,\,0.729081)\) (Proposition~\ref{prop:2fam-optimum}); the tail normalization that reduces general instances to bounded ones is Proposition~\ref{prop:two-tail}. Figure~\ref{fig:proofmap} maps the single-unit chain.

\FloatBarrier
\section{Model and the scalar curve}\label{sec:model}

Welfare is the seller's initial value plus the gains created by trade; keeping these two terms separate also exposes the GFT frontier.

The independent values \(V_s,V_b\) are nonnegative and satisfy
\[
0<\mathbb E\max\{V_s,V_b\}<\infty.
\]
A deterministic price \(z\ge0\) trades when \(V_s\le z\le V_b\). We use this inclusive convention for fixed prices. Define
\begin{equation}
M=\mathbb EV_s,\qquad G=\mathbb E(V_b-V_s)_+,\qquad
\Gamma(z)=\mathbb E[(V_b-V_s)\mathbf1_{\{V_s\le z\le V_b\}}].
\label{eq:1}
\end{equation}
The price welfare is \(W(z)=M+\Gamma(z)\), and first-best welfare is \(W_{\mathrm{FB}}=M+G\). We write \(\Gamma_{\max}=\max_{z\ge0}\Gamma(z)\) and
\begin{equation}
r_{\mathrm{FP}}
=\inf_{F_s,F_b}\frac{M+\Gamma_{\max}}{M+G}.
\label{eq:2}
\end{equation}
The maximum exists: the inclusive rule makes \(\Gamma\) upper semicontinuous by reverse Fatou, and \(\Gamma(z)\le\mathbb E[V_b\mathbf1_{\{V_b\ge z\}}]\to0\) as \(z\to\infty\). If \(\Gamma\) is identically zero, every price is optimal. Randomization over prices cannot improve the optimum for a fixed pair of priors.

One parameter \(C\) indexes the whole construction. It is the first integral \eqref{eq:firstintegral} of the Euler equation on the middle arc of the reference curve, and it sets the arc's curvature through \(P''=-CP/a^2\). As the overview explains, \(\tau(C)\) is the initial state \(d\) at which the arc with parameter \(C\) is optimal, and \(S(C)=\max\mathcal J_{\tau(C)}\) is the optimal value there (Theorem~\ref{thm:calibration}). The quadratic \(D_C\) and its integral \(\mathcal I_C\) arise when the Euler equation is integrated along the arc, whose running parameter is \(y\); \(y_e(C)\) is the value of \(y\) where the optimal arc meets the line. We record the closed forms now and use them in Theorems~\ref{thm:exact} and~\ref{thm:frontier}. The scalar functions governing the optimum use \(1/4<C<1/2\); the comparison at arbitrary endpoints will require the larger range specified in Section~\ref{sec:reference}. For \(1/4<C<1/2\), define
\begin{equation}
\begin{aligned}
D_C(y)&=1-y+Cy^2,&
y_e(C)&=\frac{1-2C}{C(1-C)},\\
\mathcal I_C(y)&=\int_0^y\frac{\mathrm dr}{D_C(r)},& \mathcal I(C)&=\mathcal I_C(y_e(C)),
\end{aligned}
\label{eq:3}
\end{equation}
and the scalar curve
\begin{equation}
\tau(C)=\frac{C^2}{1-2C}e^{-\mathcal I(C)/2},
\qquad S(C)=C(2+\mathcal I(C)).
\label{eq:4}
\end{equation}
With \(\omega=\sqrt{C-1/4}\) and principal arctangents, integration gives
\begin{equation}
\mathcal I(C)=\frac1\omega
\left[
\arctan\frac{1-3C}{2(1-C)\omega}
+\arctan\frac1{2\omega}
\right].
\label{eq:5}
\end{equation}
The normalized functional of \cite{lrw23} uses a buyer survival function \(H(s)=\Pr(V_b>s)\), its tail integral \(L(s)=\mathbb E(V_b-s)_+\), and \(K(s)=\int_s^1H(u)^2\,\mathrm du\). At a cutoff normalized by \(L(1)=d\), reversing time gives \((h,g,k)(t)=(H,L,K)(1-t)\). For general \(d>0\) and measurable \(h:[0,1]\to[0,1]\), define
\begin{equation}
\begin{aligned}
g(t)&=d+\int_0^t h(u)\,\mathrm du,\qquad
k(t)=\int_0^t h(u)^2\,\mathrm du,\\
\mathcal J_d(h)&=\int_0^1
\left[\frac{h(t)}{g(t)}+\frac{d^2-k(t)}{g(t)^2}\right]\,\mathrm dt.
\end{aligned}
\label{eq:6}
\end{equation}
Here \(d>0\) is the fixed initial state \(g(0)\); the function \(\tau(C)\) in \eqref{eq:4} parameterizes the initial states of the maximizing family. Buyer survival functions correspond to nondecreasing controls after reversing time. We maximize \(\mathcal J_d\) over the entire measurable box.

Differentiation of \eqref{eq:3}--\eqref{eq:5} gives
\begin{equation}
\mathcal I'=-\frac{2(C\mathcal I+2)}{C(4C-1)},\qquad
S'=\frac{(2C-1)\mathcal I+8C-6}{4C-1}<0,\qquad
\frac{\tau'}\tau=\frac{S'}{2C-1}>0.
\label{eq:7}
\end{equation}
The signs follow from \(1/4<C<1/2\) and \(\mathcal I>0\). At the endpoints,
\begin{equation}
(S(C),\tau(C))\longrightarrow(\infty,0)\quad(C\downarrow1/4),
\qquad
(S(C),\tau(C))\longrightarrow(1,\infty)\quad(C\uparrow1/2).
\label{eq:8}
\end{equation}
In particular, for every \(d>0\) there is a unique \(C_d\in(1/4,1/2)\) with \(\tau(C_d)=d\). Table~\ref{tab:coordinates} in Appendix~\ref{app:notation} lists every symbol by role with its definition site.

\begin{figure}[t]
\centering
\includegraphics[width=.68\linewidth]{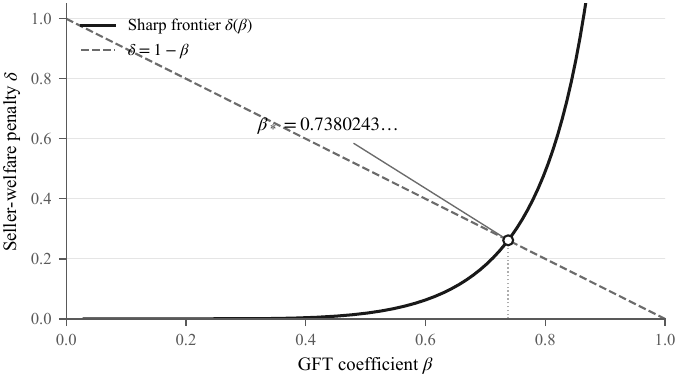}
\caption{The sharp affine frontier. Its unique intersection with \(\delta=1-\beta\) gives the welfare constant. Tangents with slope \(\kappa\) determine the exact conditional GFT guarantee through \(\kappa\rho(\kappa)=\delta^*(\kappa)\).}
\label{fig:frontier}
\end{figure}

\FloatBarrier

\section{The complete gains-from-trade frontier}\label{sec:frontier}

Define the sharp affine coefficient by
\begin{equation}
\delta(\beta)=\inf\left\{\delta\in\mathbb R:
\Gamma_{\max}\ge\beta G-\delta M
\text{ for every pair in Section~\ref{sec:model}}\right\},
\qquad 0<\beta<1.
\label{eq:46}
\end{equation}

For \(\kappa>0\), define
\begin{equation}
\rho(\kappa)=
\inf_{\substack{F_s,F_b\ \text{as in Section~\ref{sec:model}}\\M>0,\ G/M=\kappa}}
\frac{\Gamma_{\max}}G.
\label{eq:50}
\end{equation}

\begin{theoremB}[Sharp affine and conditional GFT guarantees]\label{thm:frontier}
The scalar curve \eqref{eq:4} gives the entire frontier:
\begin{equation}
\beta(C)=\frac1{S(C)},\qquad
\delta(\beta(C))=\frac{\tau(C)}{S(C)},
\qquad \frac14<C<\frac12.
\label{eq:47}
\end{equation}
Every coefficient smaller than \(\delta(\beta)\) is refuted by a bounded hard pair of independent distributions.

The function \(\delta\) is strictly increasing and strictly convex on \((0,1)\). Its derivative ranges from zero to infinity.

With \(u=1-2C\), the map
\begin{equation}
\kappa(C)=\frac{\tau(C)(S(C)+u)}u
\label{eq:51}
\end{equation}
is a strictly increasing bijection from \((1/4,1/2)\) onto \((0,\infty)\), and
\begin{equation}
\rho(\kappa(C))
=\frac1{S(C)+u}
=\frac1{1+C\mathcal I(C)}.
\label{eq:52}
\end{equation}
Equivalently, extending \(\delta\) by \(+\infty\) outside \((0,1)\) for the purpose of conjugation,
\begin{equation}
\kappa\rho(\kappa)=\delta^*(\kappa)
:=\sup_{0<\beta<1}\{\kappa\beta-\delta(\beta)\}.
\label{eq:53}
\end{equation}
For each fixed \(\kappa\), bounded hard pairs with exactly \(G/M=\kappa\) approach \eqref{eq:52}.

\end{theoremB}

\begin{corollaryB}[Endpoint asymptotics]\label{cor:asymptotics}
The sharp curves satisfy
\begin{equation}
\begin{aligned}
\delta(\beta)&\sim\frac e8\,\beta e^{-2/\beta}
&&(\beta\downarrow0),&
\delta(\beta)&\sim\frac1{4(1-\beta)}
&&(\beta\uparrow1),\\
\rho(\kappa)&\sim\frac2{\log(1/\kappa)}
&&(\kappa\downarrow0),&
1-\rho(\kappa)&\sim\frac1{\sqrt\kappa}
&&(\kappa\to\infty).
\end{aligned}
\label{eq:55}
\end{equation}
\end{corollaryB}

The frontier and its hard pairs are proved in Section~\ref{sec:frontierproof}; the endpoint expansions are proved in Appendix~\ref{app:asymptotics}.

In \eqref{eq:51}, \(\kappa(C)=\delta'(\beta(C))\) is the slope of the frontier at \(\beta(C)\), so the conditional guarantee at a ratio \(\kappa\) is read off where the frontier has slope \(\kappa\) (Figure~\ref{fig:frontier}). The welfare constant of Theorem~\ref{thm:exact-intro} is the point of this frontier on the line \(\delta=1-\beta\). There \(\tau(C_*)/S(C_*)=1-1/S(C_*)\), that is, \(S(C_*)=1+\tau(C_*)\) and \(\beta_*=1/S(C_*)\).

\section{From controls to a local energy}\label{sec:coordinates}

Each contribution \(h(u)^2\) to the accumulated second moment affects every later time. The reciprocal clock collects those later contributions into a single weight, making the objective local in the reciprocal state \(p\). The natural state for \eqref{eq:6} is \(g\); in that variable the objective is not concave and \(0\le h\le1\) constrains a slope. They are chosen so that the second-moment term becomes a Dirichlet energy with a positive weight and the two slope constraints integrate to two upper obstacles.

For a control in \eqref{eq:6}, introduce
\begin{equation}
A(t)=\int_t^1g(u)^{-2}\,\mathrm du,\qquad a_0=A(0),\qquad
p(A(t))=\frac1{g(t)},\qquad x=p(0)=\frac1{g(1)}.
\label{eq:11}
\end{equation}
Extend \(p\) by the constant \(1/d\) for \(a\ge a_0\).

\begin{lemma}[energy representation]\label{lem:energy}
For every \(d>0\), a measurable \(h:[0,1]\to[0,1]\) induces through \eqref{eq:11} an absolutely continuous \(p\) with
\begin{equation}
p'(A(t))=h(t)\quad\text{a.e.},\qquad
p(a)\le \min\{x+a,1/d\},\qquad
\frac1{d+1}\le x\le\frac1d,
\label{eq:12}
\end{equation}
and satisfies
\begin{equation}
\mathcal J_d(h)-1=\mathcal F_d[p]
:=-\log(d\cdot x)+\int_0^\infty
\left[d^2-p^{-2}-a((\log p)')^2\right]\,\mathrm da.
\tag{E}\label{eq:13}
\end{equation}
\end{lemma}
\begin{proof}
Since \(d\le g\le d+1\), the map \(A\) is strictly decreasing and bi-Lipschitz. The chain rule gives \(p'(A(t))=h(t)\) and \(\mathrm dt=-p^{-2}\,\mathrm da\); the slope bound and constant extension give the two obstacles in \eqref{eq:12}.

The term \(h(u)^2\) enters \(k(t)\) for every \(t\ge u\), so the objective has memory, and in the original variables no pointwise comparison of two controls is available. We reverse the order of integration to collect the total weight of \(h(u)^2\):
\[
\int_0^1\frac{k(t)}{g(t)^2}\,\mathrm dt
=\int_0^1 h(u)^2\left(\int_u^1g(t)^{-2}\,\mathrm dt\right)\mathrm du
=\int_0^1 A(u)h(u)^2\,\mathrm du.
\]
This is the source of the weight \(a=A(u)\). Using \(g'=h\) for the remaining first-order term and then changing variables gives
\[
\mathcal J_d(h)
=\log\frac{g(1)}d+d^2a_0-\int_0^1A(t)h(t)^2\,\mathrm dt
=-\log(d\cdot x)+\int_0^{a_0}\left[d^2-a((\log p)')^2\right]\,\mathrm da.
\]
The original time interval has length one, so \(1=\int_0^{a_0}p(a)^{-2}\,\mathrm da\). Subtracting this normalization yields the energy; its integrand vanishes on the constant extension:
\[
\mathcal J_d(h)-1
=-\log(d\cdot x)+\int_0^{a_0}
\left[d^2-p^{-2}-a((\log p)')^2\right]\,\mathrm da.
\]
\end{proof}

At fixed \(x\), the integral in \eqref{eq:13} is concave in \(\ell=\log p\), and the obstacles become
\[
\ell(a)\le \min\{\log(x+a),-\log d\}.
\]
The obstacle class relaxes the original problem, since we keep the two integrated consequences of \(0\le h\le1\) and drop the slope constraint itself. This is what makes the fixed-endpoint problem concave. The price is that a maximizer of the relaxed problem need not come from an admissible control, and attaining the bound also requires the inverse coordinate change to satisfy the original slope and time constraints; this is settled only in Theorem~\ref{thm:calibration}.

For this concave problem, a reference curve \(P_x\) must solve the Euler equation between its contacts with the two obstacles. Writing \(\ell_x=\log P_x\), the equation on that interior interval is
\[
(a\ell_x')'+e^{-2\ell_x}=0.
\]
Along a solution, differentiation shows that
\begin{equation}
C=a e^{-2\ell_x}+a\ell_x'-a^2(\ell_x')^2
\quad\text{is constant},\qquad P_x''=-\frac{C}{a^2}P_x.
\label{eq:firstintegral}
\end{equation}
We integrate this equation between tangency to the line \(x+a\) and tangency to the cap \(1/d\). The construction below gives a reference curve for every endpoint, including endpoints where that reference curve fails the original time constraint.

\section{A reference curve at every endpoint}\label{sec:reference}

At a fixed initial state \(d\), which sits at the far end of the tail-time axis where \(p=1/d\), each reference curve follows the line, solves the Euler equation on its middle arc, and then meets the constant cap. Covering every endpoint is necessary for a global comparison: a bound at the optimal endpoint alone says nothing about a competitor with a different endpoint. The branch therefore has to be continued below \(C=1/4\), where \(D_C\) has real roots and the arctangent form \eqref{eq:5} no longer applies, and above \(1/2\) up to a terminal value \(\bar C(d)\).

Fix \(d>0\). To compare all endpoints, we let the curvature parameter \(C\) vary over positive values for which the contact equations are defined. Lemma~\ref{lem:reference} identifies this range as \((0,\bar C(d))\). The endpoint envelope in the next section still uses this full range; its maximum selects \(C_d\in(1/4,1/2)\). Throughout this construction, \(D_C\) and \(\mathcal I_C\) retain their formulas in \eqref{eq:3} on intervals where \(D_C>0\). Define
\begin{equation}
\chi(C,y)=\frac{\sqrt{D_C(y)}}y e^{-\mathcal I_C(y)/2},
\qquad Cy<1,\qquad D_C(r)>0\ (0\le r\le y).
\label{eq:14}
\end{equation}
The running parameter is \(y\), and \(Y_d(C)\) denotes the contact parameter solving \(\chi(C,Y_d(C))=d\). This equation is a tangency condition: traced from the cap, where its slope is zero, the middle arc of the reference curve reaches slope one at the parameter value \(Y_d(C)\), and there it touches the line \(x+a\) (Lemma~\ref{lem:reference}). Below we often write \(Y_d\) for \(Y_d(C)\), and similarly \(Z_d\), \(U_d\), and \(\mathcal B_d\) for the functions of \(C\) defined below. The corresponding reciprocal-state endpoint and contact data are
\begin{equation}
\begin{gathered}
x(C)=\frac{1-D_C(Y_d)}{D_C(Y_d)},\qquad
a_e=\frac{C Y_d^2}{D_C(Y_d)},\\
a_a=\frac C{d^2},\qquad P_e=\frac{Y_d}{D_C(Y_d)}.
\end{gathered}
\label{eq:15}
\end{equation}
For the endpoint derivative, define
\begin{equation}
Z_d(C)=\int_0^{Y_d}\frac{r^2}{D_C(r)^2}\,\mathrm dr,\qquad
U_d(C)=\int_0^{Y_d}\frac{Y_d-r}{D_C(r)^2}\,\mathrm dr.
\label{eq:branch_integrals}
\end{equation}

\begin{lemma}[reference curves for every endpoint]\label{lem:reference}
For \(d>0\), there is \(\bar C(d)>1/4\) such that \(Y_d\) from \eqref{eq:14} exists uniquely on \(0<C<\bar C(d)\). The map \(C\mapsto x(C)\) is a continuous strictly decreasing bijection from this interval onto \((0,1/d)\). For each such endpoint, using the contact data \eqref{eq:15}, the function
\begin{equation}
\begin{gathered}
P_x(a)=
\begin{cases}
x+a,&0\le a\le a_e,\\
\text{the middle arc below},&a_e\le a\le a_a,\\
1/d,&a\ge a_a,
\end{cases}\\[4pt]
a(r)=\frac C{d^2}e^{-\mathcal I_C(r)},\quad
P_x(a(r))=\frac{e^{-\mathcal I_C(r)/2}}{d\sqrt{D_C(r)}}
\quad(0\le r\le Y_d(C))
\end{gathered}
\label{eq:16}
\end{equation}
is continuously differentiable and lies below both obstacles in \eqref{eq:12}. The function
\begin{equation}
Q_x=P_x^{-2}+\bigl(a(\log P_x)'\bigr)'
\label{eq:17}
\end{equation}
has no atoms and is given almost everywhere by
\begin{equation}
Q_x(a)=
\begin{cases}
(1+x)/(x+a)^2,&0<a<a_e,\\
0,&a_e<a<a_a,\\
d^2,&a>a_a.
\end{cases}
\label{eq:18}
\end{equation}
\end{lemma}
\begin{proof}
The contact equation determines an endpoint branch. Its monotone endpoint map gives coverage; the three pieces then determine the reference curve and its multiplier.

\textit{Existence of the endpoint branch.} The partial derivatives of \eqref{eq:14} are
\begin{equation}
\partial_y\log\chi=-\frac1{yD_C(y)}<0,\qquad
\partial_C\log\chi
=\frac{y^2}{2D_C(y)}+\frac12\int_0^y\frac{r^2}{D_C(r)^2}\,\mathrm dr>0.
\label{eq:19}
\end{equation}
For \(0<C\le1/4\), the permitted interval ends at the first positive root of \(D_C\). As \(y\) ranges from zero to that root, \(\chi\) decreases from infinity to zero. For \(C>1/4\), the permitted interval is \((0,1/C)\), and its terminal value is
\[
\chi_{\mathrm{end}}(C)
=C e^{-\mathcal I_C(1/C)/2}.
\]
Along this endpoint, \eqref{eq:19} gives
\[
(\log\chi_{\mathrm{end}})'
=\partial_C\log\chi(C,1/C)+1/C>0.
\]
This endpoint value tends to zero as \(C\downarrow1/4\) and to infinity as \(C\to\infty\); for the latter limit, use
\(\mathcal I_C(1/C)\le [C(1-1/(4C))]^{-1}\).
There is therefore a unique \(\bar C(d)>1/4\) with \(\chi_{\mathrm{end}}(\bar C(d))=d\), and the branch exists exactly on \((0,\bar C(d))\).

\textit{Coverage of every reciprocal-state endpoint.} Along \(\chi(C,Y_d)=d\), implicit differentiation using \eqref{eq:19} gives
\[
Y_d'
=-\left.\frac{\partial_C\log\chi}{\partial_y\log\chi}\right|_{y=Y_d}
=\frac{Y_d^3+Y_dD_C(Y_d)Z_d}2>0.
\]
To express \(U_d\) through \(Z_d\), we integrate two exact derivatives over \([0,Y_d]\):
\[
\frac{\mathrm d}{\mathrm dr}\frac{r}{D_C(r)}=\frac{1-Cr^2}{D_C(r)^2},\qquad
\frac{\mathrm d}{\mathrm dr}\frac{r^2}{D_C(r)}=\frac{2r-r^2}{D_C(r)^2}.
\]
They give
\[
\int_0^{Y_d}\frac{\mathrm dr}{D_C(r)^2}=\frac{Y_d}{D_C(Y_d)}+CZ_d,\qquad
2\int_0^{Y_d}\frac{r\,\mathrm dr}{D_C(r)^2}=\frac{Y_d^2}{D_C(Y_d)}+Z_d,
\]
and together with \(Y_d'\) they give
\begin{equation}
\begin{aligned}
2U_d&=\frac{Y_d^2}{D_C(Y_d)}+(2CY_d-1)Z_d,\\
\frac{\mathrm d}{\mathrm dC}D_C(Y_d)
&=Y_d^2+(2CY_d-1)Y_d'\\
&=Y_dD_C(Y_d)(Y_d+U_d).
\end{aligned}
\label{eq:20}
\end{equation}
The first line holds because \(U_d\) is \(Y_d\) times the first integral above minus the second. The second line substitutes \(Y_d'\) into the total derivative of \(D_C(Y_d)\); after inserting the first line and \(D_C(Y_d)=1-Y_d+CY_d^2\), both sides of its last equality expand to \(Y_d^2-\tfrac12Y_d^3+CY_d^4+\tfrac12(2CY_d-1)Y_dD_C(Y_d)Z_d\). Since \(U_d>0\) by its integral definition,
\begin{equation}
x'(C)=-\frac{Y_d(Y_d+U_d)}{D_C(Y_d)}<0.
\label{eq:21}
\end{equation}
At \(C=0\), the continuous extension of \eqref{eq:14} is \(\chi(0,y)=(1-y)/y\). The implicit root tends to \(1/(1+d)\), so \(x(C)\to1/d\) as \(C\downarrow0\). At the other end, \(Y_d\) is increasing and bounded by \(1/C\), so it has a limit \(\bar Y\le1/\bar C(d)\) as \(C\uparrow\bar C(d)\), and continuity of \(\chi\) gives \(\chi(\bar C(d),\bar Y)=d\). The function \(\chi(\bar C(d),\cdot)\) is strictly decreasing and takes the value \(d\) at its terminal point \(1/\bar C(d)\), so \(\bar Y=1/\bar C(d)\). Consequently \(D_C(Y_d)\to1-1/\bar C(d)+1/\bar C(d)=1\) and \(x(C)\to0\).

\textit{The reference curve and its multiplier.} We differentiate the parametric middle arc:
\begin{equation}
\frac{\mathrm da}{\mathrm dr}=-\frac a{D_C(r)},\qquad
P_x'=\frac{r}{P_xD_C(r)},\qquad
P_x''=-\frac C{a^2}P_x<0.
\label{eq:22}
\end{equation}
The identity \(\chi(C,Y_d)=d\) gives \(a(Y_d)=a_e\) and \(P_x(a_e)=P_e=x+a_e\), with slope one. At \(r=0\), the curve reaches \((a_a,1/d)\) with slope zero, so the pieces fit in \(C^1\). Figure~\ref{fig:extremal}(a) shows the curves at three endpoints. Concavity puts the middle arc below both obstacles and keeps its slope between zero and one. On that arc \(a(\log P_x)'=Cr\), so \eqref{eq:17} vanishes. Since the joins are \(C^1\), there are no jump terms, and on the contact sets
\[
Q_x(a)=\frac{1+x}{(x+a)^2}\quad(0<a<a_e),\qquad Q_x(a)=d^2\quad(a>a_a).
\]
\end{proof}

At the limiting endpoint, set \(Y_d(0)=1/(1+d)\), \(x(0)=1/d\), and \(P_{1/d}\equiv1/d\). A reference curve need not meet the original time constraint; it bounds the energy of every reciprocal state arising from a control with the same endpoint.

\section{Fixed-endpoint comparison and the endpoint envelope}\label{sec:global}

At a fixed endpoint, the energy gap splits into a square, an exponential remainder, and a contact term. Once these terms are nonnegative, only the endpoint envelope remains to be maximized.

For the reciprocal state \(p\) induced by a control \(h\), let
\begin{equation}
v(a)=\log\frac{p(a)}{P_x(a)},\qquad
\Phi_d(x)=\mathcal F_d[P_x],\qquad
\mathcal B_d(C)=2C-1+C(1-C)Y_d.
\label{eq:23}
\end{equation}

\begin{lemma}[nonnegative gaps]\label{lem:gaps}
For \(p\) from \eqref{eq:11} and the notation \eqref{eq:23}, the fixed-endpoint gap is
\begin{equation}
\Phi_d(x)-\mathcal F_d[p]
=\int_0^\infty
\left[
a(v')^2+\frac{e^{-2v}-1+2v}{P_x^2}-2Q_xv
\right]\,\mathrm da\ \ge0.
\tag{G1}\label{eq:24}
\end{equation}
It vanishes exactly when \(p=P_x\). The endpoint envelope has a unique maximum at \(C=C_d\), and
\begin{equation}
\begin{aligned}
\Phi_d(x(C))&=C\mathcal I_C(Y_d(C))-\frac{C^2Y_d(C)}{1-CY_d(C)},\\
\frac{\mathrm d}{\mathrm dC}\Phi_d(x(C))
&=-\frac{(Y_d(C)+U_d(C))\mathcal B_d(C)}{(1-CY_d(C))^2},\\
\Phi_d(x(C_d))-\Phi_d(x(C))
&=\int_{\min\{C,C_d\}}^{\max\{C,C_d\}}
\frac{(Y_d(u)+U_d(u))|\mathcal B_d(u)|}
{(1-uY_d(u))^2}\,\mathrm du.
\end{aligned}
\tag{G2}\label{eq:25}
\end{equation}
\end{lemma}
\begin{proof}
The fixed-endpoint gap bounds the energy at a prescribed endpoint; the endpoint envelope then selects the maximizing endpoint.

\textit{The fixed-endpoint gap (G1).} The reciprocal state \(p\) and reference curve \(P_x\) agree at zero and eventually equal \(1/d\), so \(v(0)=0\) and \(v\) has compact support. Expanding the energy gives
\[
\begin{aligned}
\Phi_d(x)-\mathcal F_d[p]
&=\int_0^\infty\!\left[\frac{e^{-2v}-1}{P_x^2}+a(2\ell_x'v'+(v')^2)\right]\,\mathrm da\\
&=\int_0^\infty\!\left[a(v')^2+\frac{e^{-2v}-1+2v}{P_x^2}
                         +2a\ell_x'v'-\frac{2v}{P_x^2}\right]\,\mathrm da\\
&=\int_0^\infty\!\left[a(v')^2+\frac{e^{-2v}-1+2v}{P_x^2}-2Q_xv\right]\,\mathrm da.
\end{aligned}
\]
Integration by parts contributes \(2a\ell_x'v\), which vanishes at both ends; the \(C^1\) joins produce no internal boundary terms.

The exponential remainder has the positive-kernel representation
\begin{equation}
e^{-2v}-1+2v
=4v^2\int_0^1(1-\theta)e^{-2\theta v}\,\mathrm d\theta.
\label{eq:26}
\end{equation}
Where \(Q_x>0\), the reference curve \(P_x\) equals an upper obstacle. By \eqref{eq:12}, \(v\le0\) there; this sign is the only place where the constraint \(0\le h\le1\) enters the comparison. Figure~\ref{fig:extremal}(a) shows a competitor with the endpoint of the middle curve: it lies below that curve on both contact segments and above it in between, where \(Q_x=0\). Every summand in \eqref{eq:24} is nonnegative, and \eqref{eq:26} forces \(v=0\) almost everywhere in the equality case. Continuity gives \(p=P_x\).

\textit{Value of the envelope.} We first evaluate \(\Phi_d\) by integrating the three pieces of the reference curve \(P_x\). The contact data \eqref{eq:15} and \(\chi(C,Y_d)=d\) give three identities that are used throughout this step:
\[
\begin{gathered}
1-D_C(Y_d)=Y_d\bigl(1-CY_d\bigr),\qquad
\frac x{P_e}=1-CY_d,\\
\log(d\cdot P_e)=-\tfrac12\bigl[\log D_C(Y_d)+\mathcal I_C(Y_d)\bigr].
\end{gathered}
\]
On the line,
\[
\int_0^{a_e}\frac{\mathrm da}{(x+a)^2}=\frac1x-\frac1{P_e}=\frac{CD_C(Y_d)}{1-CY_d},\qquad
\int_0^{a_e}\frac{a\,\mathrm da}{(x+a)^2}=\log\frac{P_e}x+\frac x{P_e}-1.
\]
On the middle arc we integrate in the parameter \(r\) of \eqref{eq:16}. By \eqref{eq:22}, \(\mathrm da=-(a/D_C(r))\,\mathrm dr\), \(P_x^{-2}\,\mathrm da=-C\,\mathrm dr\), and \(a(\log P_x)'=Cr\), so
\[
\begin{aligned}
\int_{a_e}^{a_a}P_x^{-2}\,\mathrm da&=CY_d,\\
\int_{a_e}^{a_a}a\bigl((\log P_x)'\bigr)^2\,\mathrm da
&=C^2\int_0^{Y_d}\frac{r^2\,\mathrm dr}{D_C(r)}\\
&=CY_d-C\mathcal I_C(Y_d)+\tfrac12\bigl[\log D_C(Y_d)+\mathcal I_C(Y_d)\bigr].
\end{aligned}
\]
The last equality uses \(Cr^2=D_C(r)-1+r\) and \(2C(1-r)=(2C-1)-D_C'(r)\). The cap contributes nothing, and \(d^2a_a=C\). The two integrals of \(P_x^{-2}\) give
\begin{equation}
\int_0^{a_a}P_x^{-2}\,\mathrm da
=C Y_d+\frac{C D_C(Y_d)}{1-CY_d}.
\label{eq:27}
\end{equation}
In the sum of all pieces the logarithms cancel against \(-\log(d\cdot x)\) by the third contact identity, and \(x/P_e=1-CY_d\) removes the remaining constants:
\[
\Phi_d(x(C))
=C\bigl(1+\mathcal I_C(Y_d)\bigr)-CY_d-\frac{CD_C(Y_d)}{1-CY_d}
=C\mathcal I_C(Y_d)-\frac{C^2Y_d}{1-CY_d},
\]
which is the first identity in \eqref{eq:25}.

\textit{Derivative along the branch.} The derivative of \(\Phi_d\) along the branch is the one computation in the single-unit proof that does not reduce to a single identity: it needs the branch derivative \(Y_d'\), the contact identities, and the following relation together:
\begin{equation}
(4C-1)Z_d=2\mathcal I_C(Y_d)+
\frac{Y_d(Y_d-2)}{D_C(Y_d)}.
\label{eq:28}
\end{equation}
To prove \eqref{eq:28}, we replace \(Y_d\) by a running value \(y\). Both sides vanish at \(y=0\), and their \(y\)-derivatives agree because \((4C-1)y^2=2yD_C(y)-y(y-2)(2Cy-1)\). No division by \(4C-1\) occurs, so the identity also holds at \(C=1/4\). Since \(\partial_C\mathcal I_C(y)=-\int_0^yr^2D_C(r)^{-2}\,\mathrm dr\), the first identity in \eqref{eq:25} gives
\[
\frac{\mathrm d}{\mathrm dC}\Phi_d(x(C))
=\mathcal I_C(Y_d)-CZ_d+\frac{CY_d'}{D_C(Y_d)}
-\frac{2CY_d-C^2Y_d^2+C^2Y_d'}{(1-CY_d)^2}.
\]
We eliminate \(\mathcal I_C(Y_d)\) by \eqref{eq:28} and substitute \(2Y_d'=Y_d^3+Y_dD_C(Y_d)Z_d\). The result is affine in \(Z_d\). Its \(Z_d\)-coefficient is
\[
\frac{(2C-1)+CY_d(3-5C)-2C^2(1-C)Y_d^2}{2(1-CY_d)^2}
=-\frac{(2CY_d-1)\,\mathcal B_d}{2(1-CY_d)^2},
\]
and the part free of \(Z_d\) is
\[
-\frac{Y_d\bigl(2D_C(Y_d)+Y_d\bigr)\mathcal B_d}{2D_C(Y_d)(1-CY_d)^2};
\]
both are polynomial identities in \(C\) and \(Y_d\) after clearing denominators. By the first line of \eqref{eq:20}, their sum is \(-(Y_d+U_d)\mathcal B_d/(1-CY_d)^2\), the derivative in \eqref{eq:25}.

\textit{Location of the maximum.} Concavity at a fixed endpoint does not give global optimality, because the nonconcavity of the original problem survives in the endpoint. A second local maximum of the envelope would break the comparison, and excluding it requires the sign of \(\mathcal B_d\) on the whole branch \(0<C<\bar C(d)\), including the ranges \(C\le1/4\) and \(C\ge1/2\) where the closed forms of Section~\ref{sec:model} are not available.

To locate the maximum, use the rational formula for \(y_e(C)\) in \eqref{eq:3} on all of \(0<C<1/2\). Within the permitted branch, \(\mathcal B_d=0\) is equivalent to \(Y_d=y_e(C)\). The contact parameter \(Y_d\) increases strictly with \(C\), whereas
\[
y_e'(C)=-\frac1{C^2}-\frac1{(1-C)^2}<0.
\]
For \(C\le1/4\), the permitted branch lies before the first root of \(D_C\), so \(2CY_d<1\) and \(\mathcal B_d<0\). For \(1/2\le C\le1\), both terms in \(2C-1+C(1-C)Y_d\) are nonnegative and their sum is positive. For \(C>1\), the bound \(CY_d<1\) gives \(\mathcal B_d>C>0\). If \(\bar C(d)<1/2\), the right-end limit instead gives \(\mathcal B_d\to\bar C(d)>0\). Thus there is exactly one zero in \((1/4,1/2)\). At that zero,
\[
D_C(y_e(C))=\frac{C^2}{(1-C)^2},\qquad
\chi(C,y_e(C))=\tau(C),
\]
so this zero is \(C_d\). The endpoint envelope increases before \(C_d\) and decreases after it. Integrating its derivative from \(C_d\) gives (G2).
\end{proof}

Dividing the derivative in \eqref{eq:25} by \eqref{eq:21} gives
\[
\frac{\mathrm d\Phi_d}{\mathrm dx}
=\frac{D_C(Y_d)\,\mathcal B_d}{Y_d(1-CY_d)^2}
=\frac1{x^2}-\frac{1+x}{P_e^2},
\]
where the second form uses \(x=(1-D_C(Y_d))/D_C(Y_d)\), \(P_e=Y_d/D_C(Y_d)\), and \(D_C(Y_d)-(1-CY_d)^2=Y_d\mathcal B_d\). This is the value that the multiplier predicts. Moving the endpoint changes \(-\log(d\cdot x)\) at rate \(-1/x\) and relaxes the line obstacle \(\log(x+a)\) at rate \(1/(x+a)\), and by \eqref{eq:18}
\[
-\frac1x+\int_0^{a_e}\frac{2Q_x(a)}{x+a}\,\mathrm da=\frac1{x^2}-\frac{1+x}{P_e^2}.
\]

\section{Global calibration and the extremal shape}\label{sec:equality}

\begin{theoremC}[Global bound and maximizing control]\label{thm:calibration}
Let \(d>0\), let \(h:[0,1]\to[0,1]\) be measurable, and let \(C_d\in(1/4,1/2)\) satisfy \(\tau(C_d)=d\). Set \(x=(d+\int_0^1h(t)\,\mathrm dt)^{-1}\). For \(x<1/d\), let \(C\in(0,\bar C(d))\) be the unique solution of \(D_C(Y_d(C))=(1+x)^{-1}\); for \(x=1/d\), set \(C=0\) and \(P_x\equiv1/d\). With the reciprocal state \(p\), reference curve \(P_x\), and logarithmic difference \(v\) from \eqref{eq:11}, \eqref{eq:16}, and \eqref{eq:23}, respectively,
\begin{equation}
\begin{aligned}
S(C_d)-\mathcal J_d(h)
={}&\int_{\min\{C,C_d\}}^{\max\{C,C_d\}}
\frac{(Y_d(u)+U_d(u))|\mathcal B_d(u)|}
{(1-uY_d(u))^2}\,\mathrm du\\
&+\int_0^\infty
\left[
a(v')^2+\frac{e^{-2v}-1+2v}{P_x^2}-2Q_xv
\right]\,\mathrm da
\ \ge0.
\end{aligned}
\tag{G}\label{eq:29}
\end{equation}
Consequently,
\begin{equation}
\max_{0\le h\le1}\mathcal J_d(h)=S(C_d).
\label{eq:30}
\end{equation}
The maximizing control is unique up to equality almost everywhere. It has a zero arc of length \(C_d\), a strictly increasing arc of length \((1-2C_d)/(1-C_d)\), and a one arc of length \(C_d^2/(1-C_d)\). On the middle arc,
\begin{equation}
t=C_d(1+y),\qquad
g(t)=d\sqrt{D_{C_d}(y)}\,e^{\mathcal I_{C_d}(y)/2},\qquad
h(t)=\frac{y\,g(t)}{D_{C_d}(y)},\qquad 0\le y\le y_e(C_d).
\label{eq:31}
\end{equation}
\end{theoremC}

\begin{proof}
The two gaps give the global bound. Attainment requires reconstructing original time; the time constraint then gives uniqueness.

\textit{The global bound.} At \(C=C_d\), the first formula in \eqref{eq:25} reduces to
\[
1+\Phi_d(x(C_d))=C_d(2+\mathcal I(C_d))=S(C_d).
\]
Lemma~\ref{lem:energy} and Lemma~\ref{lem:gaps} now give \eqref{eq:29}, with four nonnegative contributions. Equality requires \(C=C_d\) and \(p=P_{x(C_d)}\).

\textit{Attainment.} So far a reference curve is only a bound: at a general endpoint it may need more than the available unit of original time. We now take \(C=C_d\) and \(x=x(C_d)\), and reconstruct original time from the reference curve by \(\mathrm dt=-P_x^{-2}\,\mathrm da\). On the middle arc, \eqref{eq:16} gives \(\mathrm dt=C\,\mathrm dy\), and \eqref{eq:27} at \(y=y_e(C)\) becomes
\begin{equation}
Cy_e(C)+\frac{CD_C(y_e(C))}{1-Cy_e(C)}=1-C.
\label{eq:32}
\end{equation}
Thus the middle and one arcs occupy time \(1-C\). We append a zero arc of length \(C\) at the initial state \(g=d\). The resulting control is feasible on \([0,1]\), and substituting \eqref{eq:16} gives \eqref{eq:31}. Its one-arc length is \(CD_C(y_e(C))/(1-Cy_e(C))=C^2/(1-C)\). The derivative \(P_x''<0\) in reversed time gives strict increase of the control on the middle arc.

\textit{Uniqueness.} The reference curve \(P_x\) at \(C=C_d\) first reaches \(1/d\) at \(a_a=C/d^2\). Any maximizing control must induce this reciprocal state up to some tail-time horizon \(a_0\ge a_a\); its time constraint gives
\begin{equation}
1=\int_0^{a_0}P_x^{-2}\,\mathrm da
=1-C+d^2(a_0-a_a),
\qquad a_0=\frac{2C}{d^2}.
\label{eq:33}
\end{equation}
The maximizing control is therefore fixed by
\[
t(a)=\int_a^{a_0}P_x(u)^{-2}\,\mathrm du,\qquad h(t(a))=P_x'(a).
\]
\end{proof}

\begin{proposition}[No exponential subarc]\label{prop:nonexponential}
For \(d>0\), let \(h\) be the maximizing control of Theorem~\ref{thm:calibration}, represented by \eqref{eq:31} on its middle arc. There is no nonempty open subinterval of that arc on which \(h(t)=\lambda_0+\lambda_1e^{\lambda_2t}\) for real constants \(\lambda_0,\lambda_1,\lambda_2\) with \(h'>0\).
\end{proposition}

\begin{proof}
We take \(C=C_d\), so that \(d=\tau(C)\). An exponential subarc with \(h'>0\) would make \(h''/h'\) constant. From \eqref{eq:31},
\begin{equation}
\frac{h''}{h'}=\frac{2-3Cy}{CD_C(y)},\qquad
\frac{\mathrm d}{\mathrm dt}\left(\frac{h''}{h'}\right)
=\frac{2-3C-4Cy+3C^2y^2}{C^2D_C(y)^2}.
\label{eq:34}
\end{equation}
The second numerator is a nonzero quadratic. Since \(y\) varies over every interior subinterval, \(h''/h'\) cannot be constant there.
\end{proof}

\section{Prices and bounded hard pairs}\label{sec:prices}

Use the buyer survival function \(H\) and tail integral \(L\) from Section~\ref{sec:model}. The tail integral is absolutely continuous with derivative \(-H\) almost everywhere. The following density is \cite[Theorem 3.1]{lrw23} in cutoff normalization. We give the direct argument to fix the treatment of arbitrary buyer atoms.

\begin{lemma}[pricing from the variational bound]\label{lem:pricing}
Let \(0<\beta<1\) and \(d>0\). If
\begin{equation}
\beta\sup_{0\le h\le1}\mathcal J_d(h)\le1,
\label{eq:35}
\end{equation}
then every independent pair of finite-mean distributions admits a price with
\begin{equation}
\Gamma(z)\ge \beta G-\beta d\cdot M.
\label{eq:36}
\end{equation}
\end{lemma}

\begin{proof}
The cutoff fixes the target conditional GFT; the price density realizes that target after averaging over the seller.

\textit{Cutoff normalization.} We choose the cutoff where the target conditional GFT becomes zero. If \(\mathbb EV_b=0\), the right side of \eqref{eq:36} is nonpositive. Otherwise the continuous function \(d\cdot s-L(s)\) is strictly increasing from a negative value to infinity, so it has a unique positive zero. We rescale values by this zero. On \([0,1]\), the normalized tail functions satisfy
\[
L(s)=d+\int_s^1H(u)\,\mathrm du,\qquad
K(s)=\int_s^1H(u)^2\,\mathrm du.
\]
\textit{The price density.} We choose a random price with density
\begin{equation}
q(s)=\beta\left[\frac{H(s)}{L(s)}
+\frac{d^2-K(s)}{L(s)^2}\right],\qquad 0\le s\le1,
\label{eq:37}
\end{equation}
and put any remaining probability at zero. Since \(H\) is nonincreasing,
\(K(s)\le H(s)(L(s)-d)\); the density is nonnegative. Its integral is
\(\beta\mathcal J_d(H(1-\cdot))\le1\).

The absolutely continuous function \((d^2-K)/L\) has derivative \(Hq/\beta\), giving
\begin{equation}
L(s)q(s)+\int_s^1H(z)q(z)\,\mathrm dz=\beta H(s)+\beta d
\quad\text{a.e.}
\label{eq:38}
\end{equation}
\textit{The conditional guarantee.} For a seller of value \(s\le1\), the conditional GFT is at least
\[
\gamma(s)=\int_s^1[L(z)+H(z)(z-s)]q(z)\,\mathrm dz.
\]
The lower bound includes every buyer with value strictly above the price; buyers at an atom equal to the price can only add GFT. By \eqref{eq:38},
\(\gamma'(s)=-\beta H(s)-\beta d\) almost everywhere, and \(\gamma(1)=0\). Hence
\(\gamma(s)=\beta L(s)-\beta d\cdot s\) for all \(s\in[0,1]\). For \(s>1\), this target is nonpositive because \(L(s)\le d\). Independence allows us to average the conditional bound over the seller value:
\[
\int_0^1 q(z)\Gamma(z)\,\mathrm dz\ge
\beta\mathbb E L(V_s)-\beta d\,\mathbb EV_s
=\beta G-\beta d\cdot M.
\]
\end{proof}

The seller CDF below is the equalizing construction of \cite[Theorem 3.5]{lrw23}. We use it with an explicit finite buyer tail and compute the endpoint atoms and residual gap needed for the sharp frontier. Figure~\ref{fig:extremal} shows the reference curve and the associated hard pair.

\begin{lemma}[bounded hard pairs]\label{lem:hard}
Let \(H:[0,1]\to[0,1]\) be continuous and nonincreasing, with \(H(1)=\eta>0\), and let \(d>0\). Extend \(H\) to a buyer survival function by keeping it equal to \(\eta\) on \([1,R_\eta)\), where \(R_\eta=1+d/\eta\), and setting it to zero at and above \(R_\eta\). Include an atom of size \(1-H(0)\) at zero if needed. Using the tail integral \(L\), define the auxiliary integral \(T\) and seller CDF \(F_s\) by
\begin{equation}
L(s)=d+\int_s^1H(u)\,\mathrm du,\qquad
T(s)=\int_0^sL(u)^{-2}\,\mathrm du,
\qquad
F_s(s)=d\bigl(L(s)^{-1}-H(s)T(s)\bigr)\quad(0\le s<1),
\label{eq:39}
\end{equation}
and complete \(F_s\) with an atom at one. This is a valid seller CDF. With
\(T_\eta=T(1)\) and \(J_\eta=\mathcal J_d(H(1-\cdot))\), the pair satisfies
\begin{equation}
M_\eta=1-d^2T_\eta,\qquad
\Gamma_{\max,\eta}=d+\eta d^2T_\eta,\qquad
G_\eta=d+d\cdot J_\eta-d^3T_\eta.
\label{eq:40}
\end{equation}
In particular, for every \(\beta\),
\begin{equation}
\Gamma_{\max,\eta}-\beta G_\eta+\beta d\cdot M_\eta
=d(1-\beta J_\eta)+\eta d^2T_\eta.
\label{eq:41}
\end{equation}
\end{lemma}
\begin{proof}
\textit{The seller distribution.} In the Stieltjes differential of \(F_s\), the two terms containing \(\mathrm ds\) cancel, leaving
\[
\mathrm dF_s=-d\cdot T\,\mathrm dH\ge0\qquad(0<s<1).
\]
At zero, the proposed seller law has mass \(d/L(0)\). At the right endpoint,
\(F_s(1-)=1-d\cdot\eta T_\eta\), so the completion adds mass \(d\cdot\eta T_\eta\). These are valid probabilities: \(L(s)\ge d+\eta(1-s)\) gives
\(d\cdot\eta T_\eta\le \eta/(d+\eta)<1\).

\textit{GFT at a price.} The identity \((d\cdot LT)'=F_s\) implies
\begin{equation}
\int_0^zF_s(s)\,\mathrm ds=d\cdot L(z)T(z),\qquad
F_s(z)L(z)+H(z)\int_0^zF_s(s)\,\mathrm ds=d
\quad(0\le z<1).
\label{eq:42}
\end{equation}
The left side is the GFT at price \(z\), by splitting \(V_b-V_s=(V_b-z)+(z-V_s)\) on the trade event and using independence. For prices \(1\le z\le R_\eta\), all sellers trade with the remote buyer atom, and GFT equals
\(d+\eta(1-M_\eta)\). The inclusive convention includes \(z=R_\eta\).
Prices above \(R_\eta\) give no trade. Integrating \eqref{eq:42} gives the first two identities in \eqref{eq:40}.

\textit{First-best GFT.} We integrate by parts:
\[
G_\eta
=\mathbb E L(V_s)
=d+\int_0^1H(s)F_s(s)\,\mathrm ds.
\]
Fubini identifies
\[
\int_0^1H(s)^2T(s)\,\mathrm ds
=\int_0^1\frac{K(s)}{L(s)^2}\,\mathrm ds.
\]
Since the two iterated integrals agree, substitution in \eqref{eq:39} gives the third identity in \eqref{eq:40}. Subtracting yields
\[
\Gamma_{\max,\eta}-\beta G_\eta+\beta d\cdot M_\eta
=d(1-\beta J_\eta)+\eta d^2T_\eta.
\]
\end{proof}

For the rest of the paper, let \(h_C\) be the maximizing control of Theorem~\ref{thm:calibration} at \(d=\tau(C)\), and put
\begin{equation}
H_C(s)=h_C(1-s),\qquad
H_{\eta,C}(s)=\eta+(1-\eta)H_C(s)\quad(0\le s\le1).
\label{eq:43}
\end{equation}
We use Lemma~\ref{lem:hard} with this survival function. Every member has bounded supports; only the family has an escaping buyer atom.

\section{The exact welfare guarantee}\label{sec:exactproof}

\begin{theoremA}[Exact welfare guarantee]\label{thm:exact}
Let \(V_s,V_b\) be independent nonnegative values with \(0<\mathbb E\max\{V_s,V_b\}<\infty\). A fixed price trades when \(V_s\le z\le V_b\). Its optimal universal welfare guarantee is
\[
\inf_{F_s,F_b}\max_{z\ge0}
\frac{\mathbb EV_s+\mathbb E[(V_b-V_s)\mathbf1_{\{V_s\le z\le V_b\}}]}
{\mathbb E\max\{V_s,V_b\}}
=\beta_*\approx0.73802.
\]
The constant is \(\beta_*=[C_*(2+\mathcal I(C_*))]^{-1}\), where \(C_*\) is the unique root in \((1/4,1/2)\) of
\begin{equation}
\begin{aligned}
C(2+\mathcal I(C))&=1+\frac{C^2}{1-2C}e^{-\mathcal I(C)/2},\\
\mathcal I(C)&=\frac{1}{\sqrt{C-1/4}}
\left[\arctan\frac{1-3C}{2(1-C)\sqrt{C-1/4}}
+\arctan\frac{1}{2\sqrt{C-1/4}}\right].
\end{aligned}
\label{eq:rootintro}
\end{equation}
Every pair has ratio strictly greater than \(\beta_*\). For every \(\varepsilon>0\), the explicit bounded family in Lemma~\ref{lem:hard} contains a pair with ratio below \(\beta_*+\varepsilon\).
\end{theoremA}

The special functions in \eqref{eq:rootintro} come from integrating the reciprocal quadratic \(D_C(y)=1-y+Cy^2\) along the middle arc and then reconstructing the original state: the integral gives arctangents, and reconstruction gives its exponential. The quadratic acquires a double root at \(C=1/4\); the middle arc shrinks to zero length as \(C\uparrow1/2\).

In the notation of Section~\ref{sec:model}, as noted after Theorem~\ref{thm:frontier}, the root equation \eqref{eq:rootintro} reads
\begin{equation}
S(C_*)=1+\tau(C_*),
\label{eq:44}
\end{equation}
and the constant is
\begin{equation}
\beta_*=\frac1{S(C_*)}=\frac1{1+\tau(C_*)}.
\label{eq:45}
\end{equation}

\paragraph{Numerical values.}
Interval arithmetic gives \(0.73802433573<\beta_*<0.73802433574\)
from \eqref{eq:3}--\eqref{eq:5} and the root equation \eqref{eq:44}.
The corresponding parameter is \(C_*\approx0.3937571625\), with
\(7/20<\tau(C_*)<9/25\).
At \(\eta=10^{-12}\), the independent pair \eqref{eq:43} has optimal fixed-price welfare ratio
\(r_\eta<0.738024335797\) and support contained in \([0,R_\eta]\), where
\(R_\eta<360000000001\), by Proposition~\ref{prop:finite}.

\begin{proof}[Proof of Theorem~\ref{thm:exact}]
The variational maximum gives the guarantee and the limiting hard pairs give sharpness. The zero arc then rules out attainment by a finite-mean pair.

\textit{The guarantee.} By \eqref{eq:7}--\eqref{eq:8}, \(S\) strictly decreases from infinity to one while \(\tau\) strictly increases from zero to infinity, so \eqref{eq:44} determines a unique \(C_*\). Theorem~\ref{thm:calibration} gives
\(\beta_*\sup_h\mathcal J_{\tau(C_*)}(h)=1\). Lemma~\ref{lem:pricing} and \(\beta_*\tau(C_*)=1-\beta_*\) imply
\[
M+\Gamma_{\max}\ge \beta_*(M+G).
\]

\textit{Sharpness.} We fix \(C=C_*\) and let \(\eta\downarrow0\) in \eqref{eq:43}. Uniform convergence of the controls and the bound \(g\ge d\) give
\(J_\eta\to S(C_*)=1/\beta_*\) and convergence of \(T_\eta\). Equation \eqref{eq:41} becomes
\[
(M_\eta+\Gamma_{\max,\eta})-\beta_*(M_\eta+G_\eta)
=d(1-\beta_*J_\eta)+\eta d^2T_\eta\longrightarrow0.
\]
The denominator \(M_\eta+G_\eta\) is at least the buyer mean, which is at least \(d>0\). Consequently
\[
r_{\mathrm{FP}}
\le \lim_{\eta\downarrow0}
\frac{M_\eta+\Gamma_{\max,\eta}}{M_\eta+G_\eta}
=\beta_*
\le r_{\mathrm{FP}}.
\]

\textit{Nonattainment.} First suppose \(G>0\). In the cutoff normalization of Lemma~\ref{lem:pricing}, \(L(1)=d>0\), so \(H(1)>0\). The reversed control is bounded below by \(H(1)\); it cannot equal the control in Theorem~\ref{thm:calibration}, which has an initial zero arc. Thus the density in \eqref{eq:37} has mass
\begin{equation}
\alpha=\beta_*\mathcal J_d(H(1-\cdot))<1,
\qquad \alpha\Gamma_{\max}\ge\beta_*G-(1-\beta_*)M.
\label{eq:strictmass}
\end{equation}
Since \(G>0\), the event \(\{V_s<V_b\}\) has positive probability; it is the union over rational \(q\) of \(\{V_s<q<V_b\}\), so some rational \(q\) has \(\Gamma(q)>0\), and \(\Gamma_{\max}>0\). We obtain
\[
M+\Gamma_{\max}-\beta_*(M+G)
\ge(1-\alpha)\Gamma_{\max}>0.
\]
If \(G=0\), the welfare ratio is one. In all cases,
\[
\frac{M+\Gamma_{\max}}{M+G}>\beta_*.
\]

\end{proof}

\section{Proof of the gains-from-trade frontier}\label{sec:frontierproof}

For the derivative calculations, set \(\sigma=S'(C)\), with \(C\in(1/4,1/2)\) as in Theorem~\ref{thm:frontier}.

\begin{proof}[Proof of Theorem~\ref{thm:frontier}]
The affine guarantee and its sharpness determine the frontier. Convexity then identifies the tangent giving each conditional guarantee.

\textit{Affine guarantee and sharpness.} We fix \(\beta\) and choose the unique \(C\) with \(S(C)=1/\beta\). Put \(d=\tau(C)\). Since Theorem~\ref{thm:calibration} makes \eqref{eq:35} an equality, Lemma~\ref{lem:pricing} gives the guarantee with \(\delta=\beta d\).

For sharpness, apply \eqref{eq:41} to \eqref{eq:43}. Its right side tends to zero. The limiting seller mean is strictly positive; more precisely, \eqref{eq:33} gives
\begin{equation}
T_*=2C/d^2,\qquad M_*=1-2C>0.
\label{eq:48}
\end{equation}
If \(\tilde\delta<\beta d\), the expression with coefficient \(\tilde\delta\) has limit
\[
\Gamma_{\max,\eta}-\beta G_\eta+\tilde\delta M_\eta
\longrightarrow
-(\beta d-\tilde\delta)(1-2C)<0.
\]

\textit{Convexity.} With \(u\) as in \eqref{eq:51}, formula \eqref{eq:7} gives \(\sigma<0\) and \(\tau'(C)=-d\cdot\sigma/u\). Differentiating the parametric curve gives
\begin{equation}
\frac{\mathrm d\beta}{\mathrm dC}=-\frac{\sigma}{S^2}>0,\qquad
\delta'(\beta)=d\left(1+\frac S u\right)>0,\qquad
\frac{\mathrm d}{\mathrm dC}\delta'(\beta(C))
=\frac{dS(2-\sigma)}{u^2}>0.
\label{eq:49}
\end{equation}
As \(C\downarrow1/4\), \(d\) decays exponentially in \(\mathcal I\) while \(S\) grows linearly, so the middle expression tends to zero. As \(C\uparrow1/2\), both \(d\) and \(S/u\) tend to infinity. Thus \(\delta'\) increases from zero to infinity.

\textit{Conditional guarantee.} Optimizing the affine bound gives
\[
\frac{\Gamma_{\max}}G
\ge \sup_{0<\beta<1}
\left\{\beta-\frac{\delta(\beta)}\kappa\right\}.
\]
By \eqref{eq:49}, the maximizing point is unique and satisfies
\(\delta'(\beta)=\kappa\). Equations \eqref{eq:49} and \eqref{eq:47} identify this point with \eqref{eq:51}; substitution gives \(1/(S+u)\). These derivatives also show that \(\kappa(C)\) is a bijection.

For the matching sequence, \eqref{eq:40} and \eqref{eq:48} give the limits
\begin{equation}
M_*=u,\qquad \Gamma_{\max,*}=d,\qquad G_*=d(S+u).
\label{eq:54}
\end{equation}
For the pairs in \eqref{eq:43}, the functions \(\kappa_\eta(C)=G_{\eta,C}/M_{\eta,C}\) are continuous and converge locally uniformly to \(\kappa(C)\) as \(\eta\downarrow0\): the controls in \eqref{eq:31} vary continuously across their moving joins, and \(\tau(C)\) and the limiting seller mean stay bounded away from zero on compact parameter intervals, so \eqref{eq:40} gives uniform convergence of the ratios.
Since \(\kappa\) is strictly increasing, we apply the intermediate value theorem on shrinking brackets around \(C\) to obtain \(C_\eta\to C\) with
\[
\frac{G_{\eta,C_\eta}}{M_{\eta,C_\eta}}=\kappa(C),
\qquad
\frac{\Gamma_{\max,\eta,C_\eta}}{G_{\eta,C_\eta}}
\longrightarrow\frac1{S(C)+u}.
\]
\end{proof}

\section{The dominant-strategy welfare ceiling}\label{sec:dsic}

A mechanism specifies a measurable lottery over trade and transfers for every pair of reports. Truthfulness is dominant in expected utility. We require individual rationality relative to retaining the initial endowment and equal buyer-to-seller transfers in every realization. These requirements imply that there is no transfer without trade and that every realized trade price lies between the reported values at a truthful profile.

The structural fact is the randomized-price representation of Hagerty and Rogerson \cite{hr87} and \v{C}opi\v{c} and Ponsat\'{i} \cite{cp16}. We use it in an almost-everywhere form, which accommodates weak-indifference conventions at price atoms: Proposition~\ref{lem:representation} shows that, for measurable mechanisms under an absolutely continuous joint prior on a common bounded report interval, mechanism welfare is at most best-price welfare. Smoothing (Lemma~\ref{lem:smoothing}) makes this sufficient for the hard pairs.

\begin{lemma}[smoothing bounded hard pairs]\label{lem:smoothing}
An independent pair supported on \([0,\Lambda]\) has independent absolutely continuous approximations with full support on \((0,\Lambda)\) for which first-best welfare converges and the limit superior of best fixed-price welfare is at most the original optimum.
\end{lemma}
\begin{proof}
For an original value \(X\), we take independent uniforms \(\zeta_0,\zeta_1\in(0,1)\) and \(\zeta_2\in[-1/2,1/2]\) and set
\begin{equation}
X_\varepsilon=
\begin{cases}
\Lambda\zeta_1,&\zeta_0<\varepsilon,\\
(1-2\varepsilon)X+\varepsilon\Lambda+\varepsilon\Lambda\zeta_2,
&\zeta_0\ge\varepsilon.
\end{cases}
\label{eq:57}
\end{equation}
For \(0<\varepsilon<1/2\), this variable has a full-support density on \((0,\Lambda)\) and converges almost surely to \(X\). Using independent couplings for the two traders preserves independence, and bounded convergence gives convergence of first-best welfare.

The realized price welfare
\[
w(s,b,z)=s+(b-s)_+\mathbf1_{\{s\le z\le b\}}
\]
is bounded and upper semicontinuous on the compact cube. We choose approximately optimal prices along a sequence realizing the limit superior and pass to a convergent subsequence \(z_\varepsilon\to z\). Reverse Fatou gives
\[
\limsup_{\varepsilon\downarrow0}
\max_{z'}\mathbb E w(V_{s,\varepsilon},V_{b,\varepsilon},z')
\le\mathbb E w(V_s,V_b,z)
\le\max_{z'}\mathbb E w(V_s,V_b,z').
\]
\end{proof}

\begin{proposition}[Price representation under absolutely continuous priors]\label{lem:representation}
Fix a common report interval \((0,\Lambda)\), where \(\Lambda>0\). For every measurable, risk-neutral DSIC mechanism with realizationwise individual rationality and strong budget balance, there is a positive Borel measure \(\nu\) on \((0,\Lambda)\) of mass \(m\le1\) such that its allocation probability \(\varphi\) and expected transfer \(\Pi\) satisfy
\begin{equation}
\varphi(s,b)=\nu((s,b]),\qquad
\Pi(s,b)=\int_{(s,b]}z\,\nu(\mathrm dz)
\quad\text{for Lebesgue-a.e. }0<s<b<\Lambda.
\label{eq:56}
\end{equation}
Under any absolutely continuous joint prior on this report square, its expected welfare is at most that of the best deterministic price.
\end{proposition}

The proof is given in Appendix~\ref{app:representation}.

\begin{corollary}[exact dominant-strategy guarantee]\label{cor:dsic}
For the independent-prior model of Section~\ref{sec:model} and measurable mechanisms on the common nonnegative real report domain, suppose traders are risk neutral, truthfulness is dominant in expected lottery payoffs, and individual rationality and strong budget balance hold in every realization. The optimal universal welfare guarantee, allowing full prior information, is \(\beta_*\).
\end{corollary}

\begin{proof}
Fixed prices belong to this class, so Theorem~\ref{thm:exact} gives a guarantee of \(\beta_*\). For the reverse bound, we choose a bounded hard pair from Theorem~\ref{thm:exact} and apply Lemma~\ref{lem:smoothing}. Every mechanism restricted to the open support interval satisfies Proposition~\ref{lem:representation}, so on each smooth pair its welfare is at most the best price welfare. We first take the smoothing parameter to zero and then the hard-pair parameter to zero. These inequalities hold even for a mechanism selected separately for each prior pair, giving
\[
\beta_*\le r_{\mathrm{DSIC}}\le r_{\mathrm{FP}}=\beta_*.
\]
\end{proof}

\input{two_units}

\section{Related work}\label{sec:related}

\paragraph{Welfare and fixed prices.}
Myerson and Satterthwaite \cite{ms83} establish the general incompatibility of full efficiency, Bayesian incentive compatibility, voluntary participation, and no outside subsidy. Fixed-price approximation keeps strong budget balance and dominant strategies while relaxing efficiency. Blumrosen and Dobzinski \cite{bd21} obtain a \(1-1/e\) welfare guarantee. Kang, Pernice, and Vondr\'{a}k \cite{kpv22} improve that guarantee and determine the exact ratio \((2+\sqrt2)/4\) for identically distributed values. On the impossibility side, Colini-Baldeschi, de Keijzer, Leonardi, and Turchetta \cite{cbklt16} constructed a pair on which no fixed price exceeds \(0.7485\); Kang and Vondr\'{a}k \cite{kv19} lowered this to \(0.7385\), the two STOC 2023 papers to \(0.7381\), and Giambartolomei and de Keijzer to \(0.73805\). The asymmetric problem leads to the programs of Cai and Wu \cite{cw23}, the variational approach of Liu, Ren, and Wang \cite{lrw23}, and the tighter numerical bounds of Giambartolomei and de Keijzer \cite{gdk26}.

The variational ingredients from \cite{lrw23} and their role in our proof are described in Section~\ref{sec:intro}. Cai and Wu's Lemma 3.3 already produces genuine hard instances from finite upper-program solutions. Our hard pairs approach the exact constant and every point of the frontier, including at an exactly prescribed moment ratio.

\paragraph{Gains from trade.}
Welfare includes the seller's initial value; GFT measures only the improvement over keeping the item. This distinction is essential for approximation. Fixed prices admit no uniform positive GFT ratio \cite{bd21}. Colini-Baldeschi, Goldberg, de Keijzer, Leonardi, and Turchetta \cite{cbg17} give asymptotically tight logarithmic guarantees in terms of efficient-trade probability. Theorem~\ref{thm:frontier} instead determines an exact curve in the moment ratio \(G/M\), including its leading logarithmic coefficient.

Allowing Bayesian incentives yields a different GFT theory. Blumrosen and Mizrahi \cite{bm16} analyze offering mechanisms under distributional assumptions. Brustle, Cai, Wu, and Zhao \cite{bcwz17} obtain a constant approximation to the second-best GFT benchmark in two-sided markets. Deng, Mao, Sivan, and Wang \cite{dmsw22} establish a constant approximation to first-best GFT for arbitrary independent bilateral priors, improved by Fei \cite{fei22}; Jo \cite{jo26} further sharpens the bounds for the random-offerer mechanism. Liu, Qin, Ren, and Wang \cite{lqrw26} determine the exact second-best-to-first-best GFT ratio as \(1/2\), using Bayesian incentive compatibility and interim individual rationality. Their benchmark optimizes over feasible Bayesian mechanisms. Our frontier optimizes fixed prices and conditions on initial welfare, so these are different exact-value questions.

\paragraph{Information and larger markets.}
Babaioff, Goldner, and Gonczarowski \cite{bgg20} study additional competition and sample pricing; D\"utting, Fusco, Lazos, Leonardi, and Reiffenh\"auser \cite{dfflr21} establish welfare guarantees from a single seller sample in broader two-sided markets. Cai and Wu \cite{cw23} compare full and one-sided priors and analyze fixed numbers of samples. Deng, Mao, Sivan, Wang, and Wu \cite{dmsww25} study sample-based offering mechanisms for GFT. The one-sided-prior consequence of Theorem~\ref{thm:exact} uses the existing minimax equivalence; it does not supply a new sample bound.

McAfee's double auction \cite{mcafee92} obtains efficiency as markets grow while preserving dominant strategies and avoiding a deficit. Babaioff, Cai, Gonczarowski, and Zhao \cite{bcgz18} combine asymptotic efficiency with a finite-market GFT guarantee. Gerstgrasser, Goldberg, de Keijzer, Lazos, and Skopalik \cite{ggkls19} characterize multi-unit bilateral mechanisms and analyze their welfare; \cite{gdk26} further improves the fixed-price bounds in that setting. With unrestricted correlated values, Dobzinski and Shaulker \cite{ds24} show that DSIC mechanisms have no positive universal welfare ratio, while a mechanism with one side DSIC and the other Bayesian IC does. Their later reserve mechanism \cite{ds26} exceeds the fixed-price welfare ceiling even for independent values.

Giambartolomei and de Keijzer \cite{gdk26} report a two-unit upper bound of \(0.7291\), just below their single-unit lower bound of \(0.7292\), and conjecture strict separation from single-unit trade. The margin between these two numerically computed bounds is \(10^{-4}\); Theorem~\ref{thm:two-separation} widens it to \(0.0089\), with an exact rational instance below \(0.7290804\) on one side and the exact constant \(\beta_*\) on the other, and thereby proves the separation. Their marginal formulation retains the buyer and seller order constraints; Theorem~\ref{thm:two-pricing} uses the seller order through its prefix compensation. Their symmetric two-unit bound is \(0.8372\).

\paragraph{Representation and variational comparison.}
The interpretation through dominant strategies uses the randomized-price representations of Hagerty and Rogerson \cite{hr87} and \v{C}opi\v{c} and Ponsat\'{i} \cite{cp16}. The almost-everywhere representation and smoothing argument in Section~\ref{sec:dsic} handle the measure-theoretic conventions needed by the hard pairs. The representation relies on strong budget balance. When only no deficit is required, dominant-strategy mechanisms need not be randomized prices. The generalized double auctions of Zhang \cite{zhang26} trade when the bid--ask difference exceeds a random spread, and they maximize the worst-case gains from trade when values in \([0,1]\) may be correlated and only the first-best GFT is known.

Calibration is a classical way to prove global variational optimality by an exact comparison, including in nonconvex problems \cite{abd03,bf16}. Contact multipliers and complementary slackness belong to the theory of obstacle problems and variational inequalities \cite{ks00}. Here the specific construction is the reciprocal clock, a reference curve for every endpoint, and the endpoint envelope controlling the remaining nonconcavity. The conditional frontier uses ordinary convex conjugacy \cite{rock70}; related minimax methods \cite{sion58} enter the information equivalence of \cite{cw23} and the second-best analysis of \cite{lqrw26}. We evaluate the frontier and construct hard pairs at each slope.

\section{Open questions}\label{sec:open}

For randomized DSIC mechanisms, what is the exact welfare guarantee when realizationwise strong budget balance is weakened to no deficit? Beyond dominant strategies, the optimal Bayesian welfare guarantee under ex-post individual rationality remains to be determined \cite{ds26}. Unrestricted correlation destroys a positive DSIC ratio \cite{ds24}; a quantitative frontier under restricted dependence could describe how the independent guarantee deteriorates. For two units, Conjecture~\ref{conj:two-exact} asks for the candidate constant, with \eqref{eq:two-open-budget} giving a sufficient buyer-functional comparison. For \(k\ge3\), exact common-price constants and extremal structures remain open \cite{ggkls19,gdk26}. Direct compression of the second buyer to its mean and Jensen comparisons for the optimized first-unit value fail on explicit instances, which are given in the supplementary material. Which other single-parameter mechanism families admit a state transformation that makes the adversary's objective concave after fixing finitely many boundary values?

\section*{AI disclosure}
\addcontentsline{toc}{section}{AI disclosure}
Generative AI tools, including Codex and Claude, assisted with mathematical exploration, literature searches, and prose polishing. In particular, they were used to explore the reciprocal-state representation and calibration in Lemmas~\ref{lem:energy} and~\ref{lem:gaps}, the common-tail normalization in Proposition~\ref{prop:two-tail}, the structured-family optimization in Proposition~\ref{prop:2fam-optimum}, and the rational comparison construction in Lemma~\ref{lem:2fam-witness}. They were also used to search for counterexamples to intermediate claims and to implement linear programming, symbolic algebra, exact rational arithmetic, interval arithmetic, and SMT checks. With AI assistance, the authors conceived the paper's core ideas and developed its methods. The authors improved the paper's writing, readability, and economic interpretation. One of the authors verified the analytic derivations by hand. The computer-assisted parts are too large to be checked by hand: the interval and positivity certificates for the two-unit family optimization (Proposition~\ref{prop:2fam-optimum} and Appendix~\ref{app:2fam}), the exact rational evaluation of the two-unit instances in Theorem~\ref{thm:two-separation}, including the 131-type instance (Appendix~\ref{app:two-finite}), and the numerical enclosures in Section~\ref{sec:exactproof}. These checks are carried out by deterministic software (exact rational arithmetic, interval arithmetic, and SMT solvers), not by language models. Although parts of the code were written with AI assistance, the checks themselves do not depend on any language model and can be rerun independently; the code, certificates, and a replay script are available at \url{https://github.com/lintingyee-del/fixed-price-bilateral-trade}.

Another author used \href{https://github.com/frenzymath/Archon}{\textcolor{blue}{Archon}}, an agent for Lean~4 proof development, to assist in formalizing the core results of the paper in Lean~4 with Mathlib. The Lean code is available at \url{https://github.com/lintingyee-del/fixed-price-bilateral-trade}; the \href{https://drive.google.com/file/d/1bG8ySCTCiVD0MK4EAhtyTkSXpvBY_4NC/view?usp=drivesdk}{\textcolor{blue}{supplementary material}} describes the formalization. The authors checked the correctness and originality of all content, including the references, and take full responsibility for the paper.

\bibliographystyle{alpha}
\addcontentsline{toc}{section}{References}
\bibliography{references}

\clearpage
\appendix
\addtocontents{toc}{\protect\setcounter{tocdepth}{2}}

\section{Notation}\label{app:notation}

Table~\ref{tab:coordinates} lists the symbols of the single-unit analysis by role; its continuation lists the two-unit symbols.

\begin{table}[!htbp]
\centering
\fontsize{9}{10.4}\selectfont
\setlength{\tabcolsep}{4pt}
\renewcommand{\arraystretch}{1.06}
\begin{tabularx}{\linewidth}{@{}>{\raggedright\arraybackslash}p{.24\linewidth}>{\raggedright\arraybackslash}X>{\raggedright\arraybackslash}p{.15\linewidth}@{}}
\toprule
Symbol & Meaning & Definition \\
\midrule
\multicolumn{3}{@{}l}{\textit{Model quantities}}\\
\(V_s,V_b\) & Seller and buyer values. & Section~\ref{sec:model}\\
\(M,G\) & Initial seller welfare and first-best GFT. & \eqref{eq:1}\\
\(\Gamma(z),\Gamma_{\max}\) & GFT at price \(z\) and its maximum over prices. & \eqref{eq:1}--\eqref{eq:2}\\
\(r_{\mathrm{FP}}\) & Optimal universal fixed-price welfare ratio. & \eqref{eq:2}\\
\(W(z),W_{\mathrm{FB}}\) & Price welfare \(M+\Gamma(z)\) and first-best welfare \(M+G\). & Section~\ref{sec:model}\\
\multicolumn{3}{@{}l}{\textit{Scalar curve and parameter roles}}\\
\(C\) & Curvature parameter: \((1/4,1/2)\) for the scalar curve, \((0,\bar C(d))\) for reference curves; the auxiliary endpoint-value test uses \(C>1/4\). & Sections~\ref{sec:model}, \ref{sec:reference}--\ref{sec:global}\\
\(D_C(y),y_e(C)\) & Quadratic denominator and middle-arc endpoint; \(y_e\) also extends to \(0<C<1/2\) for comparison. & \eqref{eq:3}; Section~\ref{sec:global}\\
\(\mathcal I_C(y),\mathcal I(C)\) & Integral to a running endpoint \(y\), and its value at \(y_e(C)\). & \eqref{eq:3}\\
\(\tau(C),S(C)\) & Initial-state and objective-value coordinates of the scalar curve. & \eqref{eq:4}\\
\(d,C_d\) & Fixed input to \(\mathcal J_d\), and the selected parameter satisfying \(\tau(C_d)=d\). & \eqref{eq:6}--\eqref{eq:8}\\
\(C_*,\beta_*\) & Welfare parameter solving \(S(C_*)=1+\tau(C_*)\), and \(1/S(C_*)\). & Theorem~\ref{thm:exact}\\
\multicolumn{3}{@{}l}{\textit{Control and original states}}\\
\(h,g,k,\mathcal J_d\) & Control, accumulated first state, accumulated square, and objective. & \eqref{eq:6}\\
\(H,L,K\) & Buyer survival function, tail integral, and truncated squared-survival integral. & Before \eqref{eq:6}\\
\(s,t\) & Normalized value and reversed time, with \(t=1-s\). & Before \eqref{eq:6}\\
\multicolumn{3}{@{}l}{\textit{Transformed coordinates}}\\
\(A(t),a,a_0\) & Tail-time map, a point in tail time, and the horizon \(A(0)\). & \eqref{eq:11}\\
\(p,x\) & Reciprocal state \(p(A(t))=1/g(t)\) and its endpoint \(x=p(0)\). & \eqref{eq:11}\\
\(\ell,\ell_x\) & Logarithms of \(p\) and the reference curve \(P_x\). & Section~\ref{sec:coordinates}\\
\(\mathcal F_d[p]\) & Energy equal to \(\mathcal J_d(h)-1\). & \eqref{eq:13}\\
\multicolumn{3}{@{}l}{\textit{Reference curves at a fixed initial state}}\\
\(\chi(C,y),Y_d(C),\bar C(d)\) & Contact function, contact parameter solving \(\chi(C,Y_d(C))=d\), and upper limit of the curvature range. & \eqref{eq:14}; Lemma~\ref{lem:reference}\\
\(y,r\) & Running interior parameters; the contact parameter is \(Y_d(C)\), often written \(Y_d\). & \eqref{eq:3}, \eqref{eq:16}\\
\(Z_d(C),U_d(C)\) & Auxiliary integrals for the endpoint derivative. & \eqref{eq:branch_integrals}\\
\(x(C)\) & Endpoint of the reference curve with curvature \(C\). & \eqref{eq:15}\\
\(a_e,a_a,P_e\) & Line/middle-arc contact, middle-arc/cap contact, and first contact height. & \eqref{eq:15}\\
\(P_x,Q_x\) & Reference curve at endpoint \(x\), and half its obstacle multiplier. & \eqref{eq:16}--\eqref{eq:18}\\
\(\Phi_d(x),\mathcal B_d(C)\) & Endpoint envelope and the factor in its endpoint derivative. & \eqref{eq:23}\\
\(v\) & Logarithmic difference \(\log(p/P_x)\). & \eqref{eq:23}\\
\multicolumn{3}{@{}l}{\textit{Bounded hard pairs}}\\
\(\eta,R_\eta\) & Remote buyer-atom mass and its value \(1+d/\eta\). & Lemma~\ref{lem:hard}\\
\(T,T_\eta,J_\eta\) & Integral of \(L^{-2}\), its value at one, and \(\mathcal J_d(H(1-\cdot))\). & \eqref{eq:39}--\eqref{eq:40}\\
\(h_C,H_C,H_{\eta,C}\) & Maximizing control, its reversed survival function, and the survival function with finite tail mass. & \eqref{eq:43}\\
\multicolumn{3}{@{}l}{\textit{GFT frontier}}\\
\(\beta,\delta(\beta)\) & Coefficient of first-best GFT and the sharp seller-welfare penalty. & \eqref{eq:46}\\
\(\kappa,\rho(\kappa),\kappa(C)\) & Ratio \(G/M\), its conditional GFT guarantee, and its parameterization. & \eqref{eq:50}--\eqref{eq:51}\\
\(u,\sigma\) & Limiting seller mean \(u=1-2C\) and scalar derivative \(\sigma=S'(C)\). & \eqref{eq:51}; Section~\ref{sec:frontierproof}\\
\bottomrule
\end{tabularx}
\caption{Notation by role. Definition references locate the formulas; \(h\) is the control, \(p\) the reciprocal state, and \(P_x\) the reference curve.}
\label{tab:coordinates}
\end{table}

\begin{table}[!htbp]
\addtocounter{table}{-1}
\centering
\fontsize{9}{10.4}\selectfont
\setlength{\tabcolsep}{4pt}
\renewcommand{\arraystretch}{1.06}
\begin{tabularx}{\linewidth}{@{}>{\raggedright\arraybackslash}p{.24\linewidth}>{\raggedright\arraybackslash}X>{\raggedright\arraybackslash}p{.15\linewidth}@{}}
\toprule
Symbol & Meaning & Definition\\
\midrule
\multicolumn{3}{@{}l}{\textit{Two units}}\\
\(B_i,S_i;Z_i,Y_i\) & Ordered marginal values and costs; their normalized bodies. & Section~\ref{sec:two-model}--\ref{sec:two-tail}\\
\(M_2,O_2,\Gamma_2,r_2^*\) & Initial and efficient welfare, common-price gain, and worst-case ratio. & \eqref{eq:two-model}\\
\(H_i,F_i,\bar L_i,\bar b\) & Buyer survival, seller CDF, normalized tail integral, and body bound. & Section~\ref{sec:two-pricing}; \eqref{eq:two-kernel}\\
\(\pi,\varpi,\psi_i^\varpi\) & Price measure, normalized cumulative mass, and seller potential. & \eqref{eq:two-potential}\\
\(\vartheta,\Theta_\varpi,\mathcal T_{d,\vartheta}\) & Compensation, canonical prefix compensation, and obstacle operator. & \eqref{eq:two-operator}\\
\(\mathcal J_d^{(2)}\) & Least normalized common-price mass for a fixed buyer pair. & \eqref{eq:two-functional}\\
\(\pi_{\mathrm f},f_{\pi,\mathrm f}\) & Reference price probability and its density. & \eqref{eq:2fam-price-density}\\
\(\psi_{i,\mathrm f},\Upsilon_{i,\mathrm f}\) & Reference seller and buyer measure potentials. & \eqref{eq:2fam-potentials}\\
\(\mathfrak F,r_{\mathrm{fam}},d_{\mathrm{fam}}\) & Structured family, its minimum ratio, and associated penalty. & Section~\ref{sec:two-candidate}; \eqref{eq:two-open-budget}\\
\bottomrule
\end{tabularx}
\caption[]{Notation by role (continued): two units.}
\label{tab:two-notation}
\end{table}
\FloatBarrier

\section{Omitted proofs}\label{app:omitted}

This appendix proves the statements whose proofs were deferred, in the order in which they appear: Corollary~\ref{cor:asymptotics}, Proposition~\ref{lem:representation}, Proposition~\ref{prop:two-tail}, and Theorem~\ref{thm:two-pricing}.

\subsection{Endpoint asymptotics}\label{app:asymptotics}

\begin{proof}[Proof of Corollary~\ref{cor:asymptotics}]
Along the scalar curve, we write \(d=\tau(C)\).

\textit{Small gains.} As \(C\downarrow1/4\), formula \eqref{eq:5} gives \(\mathcal I=O((C-1/4)^{-1/2})\). Since \(1/\beta=S=C(2+\mathcal I)\), we have \(\mathcal I/2=1/(2\beta C)-1\) and \(1/(2\beta C)=2/\beta-(4C-1)/(2C\beta)\), so that \(e^{-\mathcal I/2}=e\,e^{-2/\beta}\exp((4C-1)/(2C\beta))\). Because \(C^2/(1-2C)\to1/8\),
\[
\frac{\tau(C)}{(e/8)e^{-2/\beta}}
=\frac{8C^2}{1-2C}
\exp\left(\frac{4C-1}{2C\beta}\right)\longrightarrow1.
\]
The exponent tends to zero because \((C-1/4)/\beta=(C-1/4)S(C)=O((C-1/4)^{1/2})\). Multiplication by \(\beta\) gives the first asymptotic. In the same limit, \(u\to1/2\),
\(\kappa=d(S+u)/u\sim(e/4)S e^{-2S}\), and \(\rho=1/(S+u)\sim1/S\). Thus \(\log(1/\kappa)\sim2S\).

\textit{Large gains.} We write \(C=1/2-\varepsilon\). Then
\(y_e\sim8\varepsilon\) and \(D_C(y)=1+O(\varepsilon)\) uniformly on the shrinking integration interval. Hence
\[
\mathcal I\sim8\varepsilon,\qquad
S=1+2\varepsilon+o(\varepsilon),\qquad
d\sim\frac1{8\varepsilon},\qquad u=2\varepsilon.
\]
Equations \eqref{eq:47}, \eqref{eq:51}, and \eqref{eq:52} give
\[
\delta\sim\frac1{8\varepsilon}\sim\frac1{4(1-\beta)},
\qquad
\kappa\sim\frac1{16\varepsilon^2},
\qquad
1-\rho\sim4\varepsilon\sim\frac1{\sqrt\kappa}.
\]
\end{proof}

\subsection{Proof of the price representation}\label{app:representation}

The representation is due to \cite{hr87,cp16}. We prove the almost-everywhere form stated in Proposition~\ref{lem:representation} in full, because the welfare comparison on smoothed hard pairs uses exactly this form and pointwise statements depend on tie conventions at price atoms.

On the report triangle \(\mathcal T=\{(s,b):0<s<b<\Lambda\}\), partial derivatives below are understood in the sense of distributions.

\begin{proof}[Proof of Proposition~\ref{lem:representation}]
The envelope identity forces the mixed derivative to vanish, which gives local separation of the allocation probability. Gluing these representations yields a price measure and the transfer formula; averaging prices gives the welfare bound.

\textit{The envelope identity.} Realizationwise IR and equal transfers force zero transfer without trade and a price in \([s,b]\) with trade, so \(\varphi=\Pi=0\) for \(s>b\) and both truthful utilities vanish on the diagonal. For fixed \(s\), put \(u_s(t)=t\varphi(s,t)-\Pi(s,t)\). The buyer's DSIC inequalities make \(\varphi(s,\cdot)\) nondecreasing and give, for \(t_1<t_2\),
\begin{equation}
(t_2-t_1)\varphi(s,t_1)
\le u_s(t_2)-u_s(t_1)
\le(t_2-t_1)\varphi(s,t_2).
\label{eq:A1}
\end{equation}
Summing over partitions of \([s,b]\) gives the first identity below, since the bounded monotone allocation has matching lower and upper integral sums and diagonal utility is zero. The seller's inequalities give the second and make \(\varphi(\cdot,b)\) nonincreasing:
\begin{align}
b\varphi(s,b)-\Pi(s,b)&=\int_s^b\varphi(s,t)\,\mathrm dt,
\label{eq:A2}\\
\Pi(s,b)-s\varphi(s,b)&=\int_s^b\varphi(t,b)\,\mathrm dt.
\label{eq:A3}
\end{align}
Adding gives the Hagerty--Rogerson envelope identity
\begin{equation}
(b-s)\varphi(s,b)=\int_s^b[\varphi(s,t)+\varphi(t,b)]\,\mathrm dt.
\label{eq:A4}
\end{equation}

\textit{Vanishing mixed derivative.} For \(E_b(s,b)=\int_s^b\varphi(s,t)\,\mathrm dt\) and \(E_s(s,b)=\int_s^b\varphi(t,b)\,\mathrm dt\), Fubini and one-dimensional integration by parts give \(\partial_bE_b=\varphi\) and \(\partial_sE_s=-\varphi\). Applying \(\partial_s\partial_b\) to \eqref{eq:A4}, with tests compactly supported away from the diagonal, cancels \(\partial_s\varphi-\partial_b\varphi\) on both sides and leaves \((b-s)\partial_s\partial_b\varphi=0\). Since \(1/(b-s)\) is smooth on \(\mathcal T\),
\begin{equation}
\partial_s\partial_b\varphi=0\quad\text{in }\mathcal D'(\mathcal T).
\label{eq:A5}
\end{equation}

\textit{Local separation.} On an open rectangle \(\mathcal U\times\mathcal V\subset\mathcal T\), we choose smooth compactly supported functions \(\theta_s,\theta_b\) of integral one and define
\begin{equation}
\begin{aligned}
F_{\mathcal U,\mathcal V}(b)&=\int_{\mathcal U}\theta_s(s)\varphi(s,b)\,\mathrm ds,\\
H_{\mathcal U,\mathcal V}(s)&=\int_{\mathcal V}\theta_b(b)[\varphi(s,b)-F_{\mathcal U,\mathcal V}(b)]\,\mathrm db.
\end{aligned}
\label{eq:A6}
\end{equation}
Every zero-integral test has a compactly supported smooth primitive, so \eqref{eq:A5} annihilates products of two such tests. For arbitrary product tests, we subtract each test's integral times \(\theta_s\) or \(\theta_b\) and expand. By \eqref{eq:A6}, \(\varphi-F_{\mathcal U,\mathcal V}-H_{\mathcal U,\mathcal V}\) annihilates all product tests. Fubini and the one-dimensional distributional fundamental lemma give
\begin{equation}
\varphi(s,b)=F_{\mathcal U,\mathcal V}(b)+H_{\mathcal U,\mathcal V}(s)\quad\text{a.e. on }\mathcal U\times\mathcal V.
\label{eq:A7}
\end{equation}

\begin{figure}[t]
\centering
\includegraphics[width=.82\linewidth]{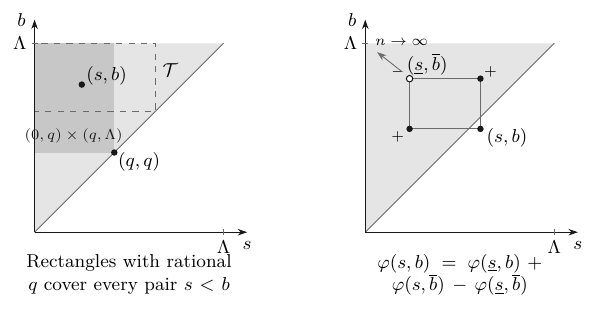}
\caption{Gluing in the proof of Proposition~\ref{lem:representation}. Left: the rectangles \((0,q)\times(q,\Lambda)\), one for each rational \(q\), cover the report triangle \(\mathcal T\). Right: identity \eqref{eq:A8} at a point \((s,b)\) with anchor \((\underline s,\overline b)\); the anchors move to \((0,\Lambda)\) along the sequence.}
\label{fig:gluing}
\end{figure}

\textit{Gluing.} Apply \eqref{eq:A7} on the countable rectangles \((0,q)\times(q,\Lambda)\), with rational \(q\in(0,\Lambda)\) (Figure~\ref{fig:gluing}, left). On each rectangle, Fubini lifts the four exceptional evaluation sets to four-dimensional null sets. These rectangles cover all ordered quadruples, giving
\begin{equation}
\varphi(s,b)=\varphi(\underline{s},b)+\varphi(s,\overline{b})-\varphi(\underline{s},\overline{b})
\label{eq:A8}
\end{equation}
for almost every \(0<\underline{s}<s<b<\overline{b}<\Lambda\). Fubini supplies a full-measure set of anchors \((\underline{s},\overline{b})\) for which \eqref{eq:A8} holds for almost every pair between them. We choose nested anchors with \(\underline{s}_n\downarrow0\) and \(\overline{b}_n\uparrow\Lambda\) (Figure~\ref{fig:gluing}, right), each from a positive-area rectangle near the endpoints. Outside one two-dimensional null set, \eqref{eq:A8} then holds for every sufficiently large \(n\).

Coordinate monotonicity gives the limits
\[
F_0(b)=\lim_{\underline{s}\downarrow0}\varphi(\underline{s},b),\qquad
H_0(s)=\lim_{\overline{b}\uparrow\Lambda}\varphi(s,\overline{b}),\qquad
m=\lim_n\varphi(\underline{s}_n,\overline{b}_n)\le1.
\]
The anchors eventually enclose every point, so \(m=\sup_{\mathcal T}\varphi=F_0(\Lambda-)=H_0(0+)\). We set \(F_1=m-H_0\) and pass to the limit in \eqref{eq:A8}:
\begin{equation}
\varphi(s,b)=F_0(b)-F_1(s)\quad\text{a.e. on }\mathcal T,
\qquad F_1(0+)=0.
\label{eq:A9}
\end{equation}
Both functions are nondecreasing and lie in \([0,m]\).

\textit{The price measure and transfer.} Substitution of \eqref{eq:A9} into both integrals of \eqref{eq:A4} is valid for almost every pair by Fubini and gives
\begin{equation}
\int_s^b(F_0(t)-F_1(t))\,\mathrm dt=0
\quad\text{for a.e. }(s,b)\in\mathcal T.
\label{eq:A10}
\end{equation}
An absolutely continuous primitive of \(F_0-F_1\) has equal values at almost every ordered pair, hence at every pair by continuity. The primitive is constant, so \(F_0=F_1\) almost everywhere. Their common right-continuous representative \(\widehat F\) satisfies \(\widehat F(0+)=0\) and \(\widehat F(\Lambda-)=m\). Its Stieltjes measure \(\nu\) is positive, has mass \(m\le1\), and gives \(\varphi(s,b)=\nu((s,b])\) almost everywhere. Substituting in \eqref{eq:A2} and applying Tonelli gives
\begin{equation}
\Pi(s,b)
=b\nu((s,b])-\int_{(s,b]}(b-z)\,\nu(\mathrm dz)
=\int_{(s,b]}z\,\nu(\mathrm dz).
\label{eq:A11}
\end{equation}

\textit{The welfare bound.} Changing \((s,b]\) to \([s,b]\) affects only seller-type lines through the countably many atoms of \(\nu\). These lines and the diagonal are Lebesgue-null, and the lower triangle has no trade. The exceptional set therefore has zero probability under the absolutely continuous prior. With \(W_\varnothing=M\), Tonelli gives
\begin{equation}
W_{\mathrm{mech}}
=(1-m)W_\varnothing+\int_{(0,\Lambda)}W(z)\,\nu(\mathrm dz).
\label{eq:A12}
\end{equation}
Bounded convergence and the absence of marginal atoms make \(W(z)\) continuous on \([0,\Lambda]\), where it attains a maximum. Since every price has welfare at least \(W_\varnothing\),
\[
W_{\mathrm{mech}}\le\max_{0\le z\le\Lambda}W(z).
\]
\end{proof}

\input{two_unit_pricing_proofs}

\section{Auxiliary results}\label{app:auxiliary}

\subsection{Quantitative finite approximation}\label{app:finite}

The hard-pair approximation error is linear in the buyer tail mass.

\begin{proposition}[finite-tail error]\label{prop:finite}
Let \(d=\tau(C_*)\), and let \(r_\eta\) be the optimal fixed-price welfare ratio of \eqref{eq:43} at \(C_*\). With
\[
L_{\mathrm{err}}(d)=\frac3d+\frac3{d^2}+\frac2{d^3},
\]
one has
\begin{equation}
0\le r_\eta-\beta_*
\le\eta\left[\beta_*L_{\mathrm{err}}(d)+\frac1d\right].
\label{eq:B1}
\end{equation}
\end{proposition}

\begin{proof}
For \(0\le h\le1\), the perturbation \(h_\eta=\eta+(1-\eta)h\) and its states satisfy
\[
0\le h_\eta-h\le\eta,\quad
0\le g_\eta-g\le\eta,\quad
0\le k_\eta-k\le2\eta,\quad
g_\eta,g\ge d.
\]
The inverse-function derivative bounds give
\[
|g_\eta^{-1}-g^{-1}|\le\eta d^{-2},\qquad
|g_\eta^{-2}-g^{-2}|\le2\eta d^{-3}.
\]
Since \(|d^2-k|\le d^2+1\), integration over the unit interval gives
\begin{equation}
|\mathcal J_d(h_\eta)-\mathcal J_d(h)|
\le\eta\left[
\frac1d+\frac1{d^2}+\frac2{d^2}
+\frac{2(d^2+1)}{d^3}\right]
=\eta L_{\mathrm{err}}(d).
\label{eq:B2}
\end{equation}
For \(h=h_{C_*}\), equations \eqref{eq:41} and \eqref{eq:45} imply
\[
r_\eta-\beta_*
=\frac{\beta_*d(\mathcal J_d(h)-\mathcal J_d(h_\eta))
+\eta d^2T_\eta}{M_\eta+G_\eta}.
\]
The first numerator term is nonnegative by Theorem~\ref{thm:calibration}. We bound the other term by \(d^2T_\eta\le1\) and the denominator by \(M_\eta+G_\eta\ge d\), giving
\[
0\le r_\eta-\beta_*
\le\eta\left[\beta_*L_{\mathrm{err}}(d)+d^{-1}\right].
\]
\end{proof}

\subsection{Nonconcavity in the original control}\label{app:nonconcavity}

The objective in \eqref{eq:6} is not concave in \(h\). If \(h=u_0\) on \([0,\theta]\) and \(h=u_1\) on \((\theta,1]\), direct integration gives
\begin{equation}
\mathcal J_d(h)
=1+\frac{\theta(1-\theta)(d+u_0)(u_1-u_0)}
{(d+\theta u_0)(d+\theta u_0+(1-\theta)u_1)}.
\label{eq:9}
\end{equation}
At \(d=\theta=1/2\), the controls \((4/5,3/5)\) and \((1/5,2/5)\) have values \(203/216\) and \(103/96\). Their midpoint is the constant \(1/2\) control, whose value is \(1\). The Jensen gap is
\begin{equation}
\frac12\left(\frac{203}{216}+\frac{103}{96}\right)-1
=\frac{11}{1728}>0.
\label{eq:10}
\end{equation}

\addtocontents{toc}{\protect\setcounter{tocdepth}{1}}
\input{two_unit_family}

\input{two_unit_computation}

\end{document}

%% file: two_units.tex
\section{Extension to two units}\label{sec:two}

\subsection{Model and separation}\label{sec:two-model}

We use the increasing submodular model of \cite{ggkls19,gdk26}.\footnote{Increasing and submodular over two units means \(B_1\ge B_2\) and \(S_1\le S_2\): marginal values decrease and marginal costs increase.}
The seller owns two units. The buyer's marginal values and the seller's
marginal costs of relinquishing successive units satisfy
\(B_1\ge B_2>0\) and \(0<S_1\le S_2\).
The two vectors are independent, with finite first moments;
their coordinates may be dependent. A common price \(z\) trades successive
units while both agents accept, using the inclusive convention of
Section~\ref{sec:model}. Since the accepting units form an initial segment,
\begin{equation}\label{eq:two-model}
\begin{aligned}
M_2&=\mathbb E(S_1+S_2),&
O_2&=\sum_{i=1}^2\mathbb E\max\{B_i,S_i\},\\
\Gamma_2(z)&=\sum_{i=1}^2\mathbb E[(B_i-S_i)
                       \mathbf1_{\{S_i\le z\le B_i\}}],&
r_2^*&=\inf\frac{M_2+\sup_{z\ge0}\Gamma_2(z)}{O_2}.
\end{aligned}
\end{equation}
The symmetric subclass requires \((B_1,B_2)\) and \((S_2,S_1)\)
to have the same law, so the valuation functions are identically distributed.
We denote its infimum by \(r_{2,\mathrm{sym}}^*\).

\begin{theoremD}[Two units are strictly harder]\label{thm:two-separation}
There are independent, finitely supported laws of strictly positive
ordered buyer and seller vectors for which every common price has
welfare ratio less than \(0.7290804\).
There is also a symmetric such instance with ratio less than \(0.83693\).
Consequently
\[
r_2^*<0.7290804<0.738024<\beta_*,
\qquad r_{2,\mathrm{sym}}^*<0.83693.
\]
\end{theoremD}

The general instance has \(131\) buyer types and \(131\) seller types.
It has a common high buyer atom, a single second-buyer body atom,
and two second-seller atoms. Its exact integer evaluation at all
\(264\) valuation events proves the bound for every real price.
Appendix~\ref{app:two-finite} gives the construction parameters and
identifies the rational data in the supplement.
Together with Theorem~\ref{thm:exact}, this resolves the separation
conjecture of \cite[Section~1]{gdk26} already for two units.

\subsection{Tail normalization}\label{sec:two-tail}

Proposition~\ref{prop:two-tail} below shows that every instance can be
replaced, without raising its ratio, by a limit of instances in which the
buyer has a common atom of vanishing probability far above all other values.
We normalize the atom's limiting first moment to one per unit and describe
the rest of each side by a bounded vector, its body: \(Z\) for the buyer and
\(Y\) for the seller.
Let \(Z=(Z_1,Z_2)\) and \(Y=(Y_1,Y_2)\) be independent bounded
nonnegative vectors with \(Z_1\ge Z_2\) and \(Y_1\le Y_2\). Define
\begin{equation}\label{eq:two-body}
\begin{aligned}
\Gamma_0(z)&=\sum_{i=1}^2\mathbb E\!\left[
 \mathbf1_{\{Y_i<z\}}
 \left(1+(Z_i-Y_i)\mathbf1_{\{z\le Z_i\}}\right)\right],\\
\mathcal R_2(Z,Y)&=
 \frac{2+\mathbb E(Y_1+Y_2)}
 {2+\sum_{i=1}^2\mathbb E\max\{Z_i,Y_i\}}.
\end{aligned}
\end{equation}
The constant \(1\) in each gain term is the contribution of that atom. The normalized model uses strict
seller acceptance and inclusive buyer acceptance, as induced by a
positive seller displacement before taking the limit.
We call \((Z,Y)\) \emph{admissible bodies} if \(\Gamma_0(z)\le2\)
for every price. Above the bodies, \(\Gamma_0(z)=2\).

\begin{proposition}[A common escaping tail]\label{prop:two-tail}
For the admissible bodies and normalized ratio defined in
\eqref{eq:two-body},
\[
r_2^*=\inf_{(Z,Y)\ \mathrm{admissible}}\mathcal R_2(Z,Y).
\]
Every admissible body pair is approached by independent, strictly
positive, finite-mean ordered instances whose best common-price welfare
ratios tend to \(\mathcal R_2(Z,Y)\).
\end{proposition}

Appendix~\ref{app:two-tail} proves both directions, including the
price constraint at the clipping atom.

\subsection{The pricing functional}\label{sec:two-pricing}

As in the single-unit argument, we first fix the buyers and look for the
least total price mass that meets a target against every seller. With two
units the target must hold for every ordered pair of seller costs
\(s_1\le s_2\). We describe a price by its cumulative mass \(\varpi(s)\),
the mass of prices above \(s\), and measure the slack of unit \(i\) at
seller cost \(s\) by a potential \(\psi_i^\varpi(s)\). Feasibility is the
joint condition \(\psi_1^\varpi(s_1)+\psi_2^\varpi(s_2)\ge0\) on
\(s_1\le s_2\). A nondecreasing compensation \(\vartheta\) splits it into
one inequality per unit, and for fixed \(\vartheta\) the least feasible mass
is the fixed point of a contraction.

We fix buyer bodies \(Z_1\ge_{\mathrm{st}}Z_2\) supported on
\([0,\bar b]\) and write \(H_i(s)=\Pr(Z_i>s)\), so that the buyer
Stieltjes measures are \(-\mathrm dH_i\); for seller bodies we write
\(F_i(s)=\Pr(Y_i\le s)\).
For \(s,z\ge0\), define
\begin{equation}\label{eq:two-kernel}
\begin{aligned}
\bar L_i(s)&=1+\mathbb E(Z_i-s)_+,\\
g_i(s,z)&=\mathbf1_{\{s<z\}}
 \left[1+\mathbb E\bigl((Z_i-s)\mathbf1_{\{z\le Z_i\}}\bigr)\right].
\end{aligned}
\end{equation}
Thus \(s+\bar L_i(s)=1+\mathbb E\max\{Z_i,s\}\) is normalized
efficient welfare. A nonnegative price measure \(\pi\) gives conditional gain
\(\int g_i(s,z)\,\pi(\mathrm dz)\), and we minimize its total mass;
any unused probability is placed at zero.

Fix \(d>0\); as with one unit, the target charges \(\beta d\) per unit of
seller cost, and \(d=(1-\beta)/\beta\) is the welfare case
(Theorem~\ref{thm:two-pricing}(iii)). A normalized price consists of a body measure
\(\widehat\pi_0\) on \((0,\bar b]\) and an ideal tail mass \(\varpi(\bar b)\),
which contributes one unit of gain per unit of mass to every seller
in \([0,\bar b]\). Its cumulative mass
\[
\varpi(s)=\varpi(\bar b)+\widehat\pi_0((s,\bar b])
\]
is nonnegative, nonincreasing, and right-continuous.
The gain and the seller potential of unit \(i\) are
\begin{equation}\label{eq:two-potential}
\begin{aligned}
G_i^\varpi(s)&=\bar L_i(s)\varpi(s)
 -\int_{(s,\bar b]}(t-s)\varpi(t)\,(-\mathrm dH_i(t)),\\
\psi_i^\varpi(s)&=d\cdot s-\bar L_i(s)+G_i^\varpi(s).
\end{aligned}
\end{equation}
The atom at \(s\) is excluded from \(\varpi(s)\); with this trace
convention a buyer atom at the posted price still trades.

The order between the two seller costs enters through a compensation.
Let \(\mathcal V\) consist of bounded, nondecreasing, right-continuous
compensations \(\vartheta:[0,\bar b]\to\mathbb R\).
If \(\psi_1^\varpi\ge-\vartheta\) and \(\psi_2^\varpi\ge\vartheta\), then
\(\psi_1^\varpi(s_1)+\psi_2^\varpi(s_2)\ge\vartheta(s_2)-\vartheta(s_1)\ge0\)
whenever \(s_1\le s_2\), because \(\vartheta\) is nondecreasing.
Conversely, a nonnegative potential sum permits the canonical compensation
\(\Theta_\varpi(s)=-\inf_{t\le s}\psi_1^\varpi(t)\).
For fixed \(\vartheta\) the two inequalities become a fixed-point problem
for \(\varpi\): with \(\epsilon_1=-1\) and \(\epsilon_2=1\), dividing
\(\psi_i^\varpi\ge\epsilon_i\vartheta\) by \(\bar L_i\) and rearranging gives
\(\varpi\ge\mathcal A_{i,d,\vartheta}\varpi\). The operator below also takes
a supremum over future seller costs, which keeps the cumulative mass
nonincreasing. On bounded Borel functions \(\varpi\), define
\begin{equation}\label{eq:two-operator}
\begin{aligned}
(\mathcal T_{d,\vartheta}\varpi)(s)
 &=\sup_{s\le t\le\bar b}
 \max\{0,\mathcal A_{1,d,\vartheta}\varpi(t),
           \mathcal A_{2,d,\vartheta}\varpi(t)\},\\
\mathcal A_{i,d,\vartheta}\varpi(t)
 &=1-\frac{d\cdot t}{\bar L_i(t)}
 +\frac{\epsilon_i\vartheta(t)}{\bar L_i(t)}
 +\frac{\int_{(t,\bar b]}(v-t)\varpi(v)\,(-\mathrm dH_i(v))}
        {\bar L_i(t)}.
\end{aligned}
\end{equation}

\begin{theoremE}[Pricing for ordered buyer pairs]\label{thm:two-pricing}
For every ordered buyer pair supported on \([0,\bar b]\) and every \(d>0\):
\begin{enumerate}
\renewcommand{\labelenumi}{(\roman{enumi})}
\item For each \(\vartheta\in\mathcal V\), the operator
\(\mathcal T_{d,\vartheta}\) has a unique bounded fixed point
\(\varpi_{d,\vartheta}\), the pointwise least feasible cumulative mass.
Iteration from zero converges uniformly with contraction factor
\(\bar b/(1+\bar b)\).
\item The value \(\varpi_{d,\vartheta}(0)\) is jointly convex in
\((d,\vartheta)\), and the minimum
\begin{equation}\label{eq:two-functional}
\mathcal J_d^{(2)}(H_1,H_2)
 =\min_{\vartheta\in\mathcal V}\varpi_{d,\vartheta}(0)
\end{equation}
is attained.
\item For \(0<\beta<1\), there is a price probability with bounded support satisfying
\begin{equation}\label{eq:two-price-target}
\sum_{i=1}^2\int g_i(s_i,z)\,\pi(\mathrm dz)
\ge\beta\sum_{i=1}^2\bar L_i(s_i)-\beta d(s_1+s_2)
\quad(0\le s_1\le s_2)
\end{equation}
if and only if \(\beta\mathcal J_d^{(2)}(H_1,H_2)\le1\).
At \(d=(1-\beta)/\beta\), this is a welfare guarantee of \(\beta\)
against every ordered seller distribution in the normalized model.
\item If \(H_i^\delta\) is the survival function of
\(\delta\lfloor Z_i/\delta\rfloor\), then
\[
\mathcal J_d^{(2)}(H_1,H_2)
\le\mathcal J_d^{(2)}(H_1^\delta,H_2^\delta)+\delta.
\]
The supremum over bounded ordered buyer bodies consequently equals
the supremum over finite-support ordered buyer bodies.
\end{enumerate}
\end{theoremE}

The tail is realized by density \(\beta d\) on
\((\bar b,\bar b+\varpi_{d,\vartheta}(\bar b)/d)\), an empty
interval when the tail mass is zero. This measure need not be finitely
atomic, and it is not the only such price: the same tail mass may
instead be placed as a single atom at
\(\bar b+\varpi_{d,\vartheta}(\bar b)/d\).
Appendix~\ref{app:two-pricing} proves the theorem
and gives the finite-node recurrence.
This extends the buyer-dependent pricing step of
\cite[Theorem~3.1]{lrw23}; the new ingredient is the prefix compensation.

\subsection{The structured family}\label{sec:two-candidate}

All quantities carrying the subscript $\mathrm f$ belong to the family
defined here. The family describes limiting two-unit instances, constructed
in Appendix~\ref{app:2fam}: both sellers have the same mass at zero, the
second buyer's body is a single value, the second seller has two atoms, and
the first buyer is described by one concave curve $\mathcal P$. We use $t_{\mathrm f}$ and $p_{\mathrm f}$ for its scalar parameters,
$\xi$ for its curve coordinate, and $\mathcal P$ for its reciprocal state.
They are distinct from the single-unit coordinates in the main text.
For $0<t_{\mathrm f}<1$, $t_{\mathrm f}\le p_{\mathrm f}\le1$, and $\xi_0>0$, define
\begin{gather}
 m_{\mathrm f}=\frac{2t_{\mathrm f}}{1+t_{\mathrm f}},\qquad
 N_{\mathrm f}=\frac2{1+t_{\mathrm f}},\qquad
 c_{\mathrm f}=\frac{p_{\mathrm f}-t_{\mathrm f}}{2},\qquad
 \delta_{\mathrm f}=\frac{p_{\mathrm f}+t_{\mathrm f}}{2},\notag\\
 e_{\mathrm f}=\frac{1+t_{\mathrm f}}{2},\qquad
 v_{\mathrm f}=1-e_{\mathrm f},\qquad
 a_{\mathrm f}=\frac{c_{\mathrm f}}{t_{\mathrm f}p_{\mathrm f}},\qquad
 h_{\mathrm f}=\frac{v_{\mathrm f}}{\xi_0}.
 \label{eq:2fam-parameters}
\end{gather}
We require $\xi_0\ge c_{\mathrm f}$.
The closed curve class $\mathfrak F$ consists of positive,
concave Lipschitz functions on $[c_{\mathrm f},\xi_0]$ satisfying
\begin{equation}
 \mathcal P(c_{\mathrm f})=p_{\mathrm f},\qquad
 \mathcal P(\xi_0)=1,\qquad
 \mathcal P(\xi)\le U_{\mathrm f}(\xi)
 :=\min\{\xi+\delta_{\mathrm f},h_{\mathrm f}\xi+e_{\mathrm f}\}.
 \label{eq:2fam-class}
\end{equation}
A parameter triple $(t_{\mathrm f},p_{\mathrm f},\xi_0)$ is feasible when
this class is nonempty.
For a nondegenerate interval, concavity and the two endpoint obstacles
imply $h_{\mathrm f}\le\mathcal P'\le1$ almost everywhere.
The case $p_{\mathrm f}=1$ is included as a degenerate interval by continuity.
We distinguish the subclass whose endpoint slopes are $1$ and
$h_{\mathrm f}$; the closed class also admits endpoint atoms and corners.
Its functionals are
\begin{align}
 T_{\mathrm f}&=\int_{c_{\mathrm f}}^{\xi_0}\mathcal P^{-2}\,\mathrm d\xi,
 &Q_{\mathrm f}&=\int_{c_{\mathrm f}}^{\xi_0}
          \xi(\mathcal P'/\mathcal P)^2\,\mathrm d\xi,\notag\\
 M_{\mathrm f}&=N_{\mathrm f}(a_{\mathrm f}+T_{\mathrm f}-\xi_0),
 &G_{\mathrm f}&=2+2m_{\mathrm f}a_{\mathrm f}
       -N_{\mathrm f}\log p_{\mathrm f}-N_{\mathrm f}Q_{\mathrm f},\notag\\
 R_{\mathrm f}&=\frac{M_{\mathrm f}+2}{M_{\mathrm f}+G_{\mathrm f}}.
 &&\label{eq:2fam-functionals}
\end{align}
Both sellers have mass $m_{\mathrm f}$ at zero, so
$t_{\mathrm f}$ fixes this common mass, and the second buyer's body is the
single value $a_{\mathrm f}$. The curve $\mathcal P$ is the reciprocal of the
first buyer's normalized tail integral as a function of the coordinate $\xi$,
and its right derivative is that buyer's survival function. The first
buyer's body starts at $a_{\mathrm f}$, where $\mathcal P=p_{\mathrm f}$, ends
at the coordinate value $\xi_0$, and has length $T_{\mathrm f}$ in value units.
The constant $N_{\mathrm f}$ scales the first seller's CDF on the body, and
$c_{\mathrm f},\delta_{\mathrm f},e_{\mathrm f},v_{\mathrm f},h_{\mathrm f}$
locate the coordinate interval and the two affine obstacles in
\eqref{eq:2fam-class}; the choice $h_{\mathrm f}=v_{\mathrm f}/\xi_0$ puts the
second obstacle through $(\xi_0,1)$. In these instances $M_{\mathrm f}$ and
$G_{\mathrm f}$ are the limiting seller welfare and efficient gains, and
$R_{\mathrm f}$ is the welfare ratio (Lemma~\ref{lem:2fam-realization}).

\begin{proposition}[The minimum within the two-unit curve family]
\label{prop:2fam-optimum}
Let $\mathfrak F$ and $R_{\mathrm f}$ be defined by
\eqref{eq:2fam-parameters}--\eqref{eq:2fam-functionals}.
The infimum
$ r_{\mathrm{fam}}=\inf_{\mathfrak F}R_{\mathrm f}$
is attained at a unique curve and parameter triple $(t_{\mathrm f},p_{\mathrm f},\xi_0)$.
Its curve consists of a positive initial affine contact, one strictly
concave interval satisfying
$\mathcal P''+C_{\mathrm f}\mathcal P/\xi^2=0$ for a constant $C_{\mathrm f}>0$,
and a positive final affine contact. The same infimum holds in the
subclass with endpoint slopes $1,h_{\mathrm f}$, and
\[
\frac{18227}{25000}<r_{\mathrm{fam}}<\frac{729081}{10^6}.
\]
The minimizing body laws are approached by strictly positive finite-mean
ordered instances with a common escaping buyer atom.
\end{proposition}

The family value is approximately $0.729080$.
Appendix~\ref{app:2fam} proves the proposition and characterizes its value
by the unique self-consistent physical stationary solution.

\begin{figure}[!ht]
\centering
\includegraphics[width=.88\linewidth]{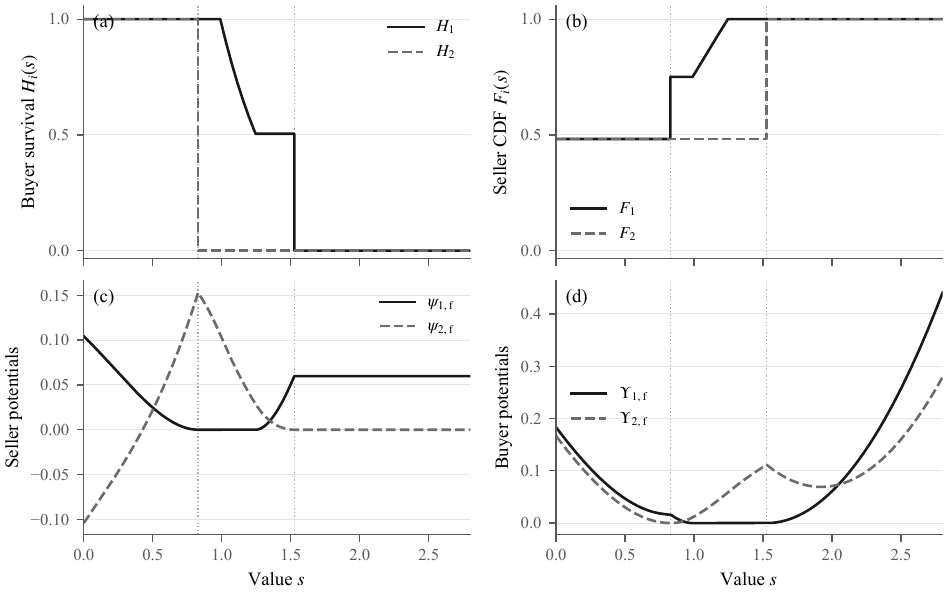}
\caption{The two-unit family minimizer. The top panels show the buyer survival functions and seller CDFs. The lower panels show the seller and buyer measure potentials for the reference price in Proposition~\ref{prop:2fam-separate-responses}. Solid and dashed curves refer to the first and second units. Their ordered sums are nonnegative. Vertical lines mark the second buyer's body atom and the common terminal body value.}
\label{fig:two-structure}
\end{figure}

\subsection{The remaining comparison}\label{sec:two-open}

\begin{conjecture}\label{conj:two-exact}
The optimal two-unit common-price welfare guarantee equals
\(r_2^*=r_{\mathrm{fam}}\).
\end{conjecture}

Proposition~\ref{prop:two-tail} and Theorem~\ref{thm:two-pricing}
give the following sufficient condition for the matching lower bound:
\begin{equation}\label{eq:two-open-budget}
\sup_{H_1\ge H_2\ \mathrm{bounded\ buyer\ bodies}}
\mathcal J_{d_{\mathrm{fam}}}^{(2)}(H_1,H_2)
\le\frac1{r_{\mathrm{fam}}},
\qquad d_{\mathrm{fam}}=\frac{1-r_{\mathrm{fam}}}{r_{\mathrm{fam}}}.
\end{equation}
Finite buyer support suffices in this supremum by part~(iv) of
Theorem~\ref{thm:two-pricing}. If this budget bound holds, averaging \(\Gamma_0\le2\) under the
price in Theorem~\ref{thm:two-pricing} gives
\(2+\mathbb E(Y_1+Y_2)\ge
r_{\mathrm{fam}}[2+\sum_i\mathbb E\max\{Z_i,Y_i\}]\)
for every admissible pair.
At the family minimizer, Proposition~\ref{prop:2fam-separate-responses}
checks this comparison from each side: in the limiting model of
Appendix~\ref{app:2fam-one-sided}, one price distribution guarantees the
ratio \(r_{\mathrm{fam}}\) against every ordered seller pair when the buyers
are held at the minimizer, and against every ordered buyer pair when the
sellers are held there. The budget bound \eqref{eq:two-open-budget} asks for
the first property at every bounded ordered buyer pair, with a price that
may depend on the buyers. The identity \(1+d_{\mathrm{fam}}=1/r_{\mathrm{fam}}\)
is exactly the budget condition under which
Appendix~\ref{app:two-recurrence} extends the finite-node constraints to
every ordered seller pair.

\paragraph{Numerical evidence.}
A floating-point search over 20,000 finite-support ordered buyer pairs,
broad random samples together with random perturbations of the discretized
family minimizer, found no pair above the budget. The largest value,
\(1.3715902\), came from a perturbed minimizer, against
\(1/r_{\mathrm{fam}}\approx1.3715909\). This is search evidence for the
conjecture and does not establish the supremum in
\eqref{eq:two-open-budget}; the sampling design is in
Appendix~\ref{app:two-recurrence}.

%% file: two_unit_pricing_proofs.tex
\subsection{Proof of the two-unit tail normalization}\label{app:two-tail}

\begin{proof}[Proof of Proposition~\ref{prop:two-tail}]
We clip the instance at a cutoff chosen by its maximum price gain,
balance the removed buyer tails, and realize arbitrary admissible
bodies by finite tails.

\textit{Choosing the cutoff.}
We write \(\Gamma_{2,\max}=\sup_z\Gamma_2(z)\) for the original maximum price gain.
For an instance with \(R_2=(M_2+\Gamma_{2,\max})/O_2<1\), we have \(\Gamma_{2,\max}>0\).
Indeed, if every rational price had zero gain, the union of
\(\{S_i<z<B_i\}\) over rational \(z\) would give \(O_2=M_2\).
Also \(\Gamma_{2,\max}<O_2-M_2\le\mathbb E(B_1+B_2)\).
The function \(c\mapsto\sum_i\mathbb E(B_i-c)_+\) is continuous,
starts above \(\Gamma_{2,\max}\), and tends to zero.
We therefore choose \(c_2>0\) so that
\begin{equation}\label{eq:two-tail-cut}
e_i=\mathbb E(B_i-c_2)_+,\qquad e_1+e_2=\Gamma_{2,\max}.
\end{equation}
Let \(Z_i=B_i\wedge c_2\) and \(Y_i=S_i\wedge c_2\).
We take the new buyer to equal \(Z\) with probability \(1-\eta\)
and \((\Gamma_{2,\max}/(2\eta),\Gamma_{2,\max}/(2\eta))\) with probability \(\eta\).
We choose \(\eta\) small enough that the tail exceeds
\(c_2+\varepsilon\), and take the new seller to be
\(Y+(\varepsilon,\varepsilon)\), independently of the buyer.

\textit{Controlling every price.}
For \(z\le c_2\), define
\[
\begin{aligned}
f_i(z)&=\Pr(S_i\le z-\varepsilon),\\
t_i(z)&=\mathbb E[(S_i+\varepsilon)
                    \mathbf1_{\{S_i\le z-\varepsilon\}}],\\
\mathscr K_i(z)&=\mathbb E[(Z_i-S_i-\varepsilon)
              \mathbf1_{\{S_i+\varepsilon\le z\le Z_i\}}],\\
\mathscr L_i(z)&=\mathbb E[(B_i-S_i)
              \mathbf1_{\{S_i\le z-\varepsilon,\ z\le B_i\}}].
\end{aligned}
\]
The displacement excludes sellers clipped to \(c_2\).
Independence and \(Z_i=B_i\wedge c_2\) give
\[
\mathscr L_i(z)=\mathscr K_i(z)+e_if_i(z)
             +\varepsilon f_i(z)\Pr(B_i\ge z).
\]
Writing \(\Gamma_{\eta,\varepsilon}\) for the new gain, we obtain
\begin{equation}\label{eq:two-tail-balance}
\begin{aligned}
\sum_i\mathscr L_i(z)-\Gamma_{\eta,\varepsilon}(z)
={}&\eta\sum_i\mathscr K_i(z)
  +\frac{(e_1-e_2)(f_1(z)-f_2(z))}{2}\\
 &+\varepsilon\sum_i f_i(z)\Pr(B_i\ge z)
  +\eta\sum_i t_i(z)\ge0.
\end{aligned}
\end{equation}
Every term is nonnegative, since both orders align the differences
in the product. The original price includes all trades counted by
\(\mathscr L_i(z)\), so
\(\Gamma_{\eta,\varepsilon}(z)\le\Gamma_2(z)\le \Gamma_{2,\max}\).
For \(z>c_2\), only the common tail can buy, and its gain is at most
\(\Gamma_{2,\max}\). Every price strictly between \(c_2+\varepsilon\) and the
tail value attains \(\Gamma_{2,\max}-\eta M_\varepsilon\), where
\(M_\varepsilon=\sum_i\mathbb E(Y_i+\varepsilon)\).
Consequently
\begin{equation}\label{eq:two-tail-gain-bounds}
\Gamma_{2,\max}-\eta M_\varepsilon
\le\sup_z\Gamma_{\eta,\varepsilon}(z)\le \Gamma_{2,\max}.
\end{equation}

\textit{Comparing the removed welfare.}
Define the removed welfare quantities
\[
\Delta_S=\sum_i\mathbb E(S_i-c_2)_+,\qquad
\Delta_O=\sum_i\mathbb E[S_i-\max\{B_i,c_2\}]_+.
\]
They satisfy
\(0\le\Delta_O\le\Delta_S\) and
\(M_\varepsilon=M_2-\Delta_S+2\varepsilon\).
The pointwise identity
\[
\max\{b,s\}-\max\{b\wedge c,s\wedge c\}-(b-c)_+
 =[s-\max\{b,c\}]_+
\]
gives \(\Gamma_{2,\max}+\sum_i\mathbb E\max\{Z_i,Y_i\}=O_2-\Delta_O\).
The new efficient welfare is
\(\Gamma_{2,\max}+(1-\eta)\sum_i\mathbb E\max\{Z_i,Y_i+\varepsilon\}\).
By \eqref{eq:two-tail-gain-bounds}, its ratio tends to
\((M_2-\Delta_S+\Gamma_{2,\max})/(O_2-\Delta_O)\).
The comparison with \(R_2\) follows from
\[
(M_2+\Gamma_{2,\max})(O_2-\Delta_O)
 -(M_2-\Delta_S+\Gamma_{2,\max})O_2
=O_2(\Delta_S-\Delta_O)+(O_2-M_2-\Gamma_{2,\max})\Delta_O\ge0.
\]
At each fixed price, letting \(\eta,\varepsilon\downarrow0\)
in the gain formula gives \eqref{eq:two-body} with \(\Gamma_{2,\max}/2\)
in place of \(1\). Scaling all values by \(2/\Gamma_{2,\max}\) therefore
gives admissible bodies with no larger ratio.

\textit{Realizing arbitrary admissible bodies.}
We fix admissible \((Z,Y)\) bounded by \(\bar b\).
We take the buyer to equal \(Z\) with probability \(1-\eta\)
and \((1/\eta,1/\eta)\) with probability \(\eta\), and use
the seller \(Y+(\varepsilon,\varepsilon)\).
We choose \(1/\eta>\bar b+\varepsilon\).
For \(z\le1/\eta\), the gain is
\[
\begin{aligned}
 &(1-\eta)\sum_i\mathbb E[(Z_i-Y_i-\varepsilon)
            \mathbf1_{\{Y_i+\varepsilon\le z\le Z_i\}}]\\
 &\quad+\sum_i\Pr(Y_i\le z-\varepsilon)
 -\eta\sum_i\mathbb E[(Y_i+\varepsilon)
                      \mathbf1_{\{Y_i\le z-\varepsilon\}}]
 \le\Gamma_0(z)\le2.
\end{aligned}
\]
Prices above \(1/\eta\) gain zero; prices between
\(\bar b+\varepsilon\) and \(1/\eta\) have gain tending to \(2\).
The initial and efficient welfare tend to
\(\mathbb E(Y_1+Y_2)\) and
\(2+\sum_i\mathbb E\max\{Z_i,Y_i\}\).
We add a common positive displacement to all values to make them
strictly positive. It translates prices and adds the same amount
to both welfare quantities. Letting this displacement vanish gives
\[
r_2^*\le\inf_{(Z,Y)\ \mathrm{admissible}}\mathcal R_2(Z,Y)\le r_2^*.
\]
\end{proof}

\subsection{Proof of the two-unit pricing theorem}\label{app:two-pricing}

\begin{proof}[Proof of Theorem~\ref{thm:two-pricing}]
We express gains through cumulative price mass and use the positive
kernel to find the least feasible price. Compactness gives an optimal
compensation, and leastness makes the value independent of \(\bar b\).
Realizing its tail, averaging over sellers, and transferring rounded
buyers then prove the two remaining parts.

\textit{Cumulative mass and the prefix constraint.}
A unit of ideal tail mass contributes
\[
\bar L_i(s)-\int_{(s,\bar b]}(t-s)\,(-\mathrm dH_i(t))=1
\]
to \eqref{eq:two-potential}.
A unit of mass at price \(z>s\) contributes
\[
\bar L_i(s)-\int_{s<t<z}(t-s)\,(-\mathrm dH_i(t))
=1+\mathbb E[(Z_i-s)\mathbf1_{\{Z_i\ge z\}}]=g_i(s,z).
\]
Integration gives \eqref{eq:two-potential}.
For nonnegative nonincreasing \(\varpi\), the potential constraints
are equivalent to \(\varpi\ge\mathcal T_{d,\vartheta}\varpi\).
If the potential sum is nonnegative on every ordered pair, the
canonical prefix \(\Theta_\varpi\) satisfies both constraints.
It is bounded and right-continuous because the potentials are
bounded and right-continuous.

\textit{The least price for fixed compensation.}
The positive kernel in \eqref{eq:two-operator} has mass
\[
\frac{\int_{(t,\bar b]}(v-t)\,(-\mathrm dH_i(v))}{\bar L_i(t)}
=\frac{\bar L_i(t)-1}{\bar L_i(t)}
\le\frac{\bar b}{1+\bar b}<1.
\]
Maxima and suprema do not increase uniform differences, so the operator
is a contraction on bounded Borel functions.
We iterate from \(\varpi^{(0)}=0\). Successive differences are
bounded by a geometric series, giving a uniform limit that is a
fixed point. The contraction inequality makes the fixed point unique.
The output is nonnegative and nonincreasing. The integral terms are
continuous in \(t\), and both obstacles are right-continuous.
Their future supremum is right-continuous as well: a strict gap
from its right limit would require an isolated obstacle peak.
These properties pass to the uniform limit.

The positive kernel makes the operator order-preserving.
For any feasible nonnegative cumulative mass \(\widetilde\varpi\),
induction gives
\[
\varpi^{(0)}\le\widetilde\varpi,\qquad
\varpi^{(n+1)}\le
\mathcal T_{d,\vartheta}\widetilde\varpi\le\widetilde\varpi.
\]
The limit is pointwise least.

\textit{Convexity and attainment.}
The potential inequalities are affine in
\((\varpi,d,\vartheta)\).
A convex combination of two feasible triples is feasible;
leastness gives the convexity in part~(ii).
The constant \(\varpi=1+\bar b\), with \(\vartheta=0\),
is feasible. We take a minimizing sequence with
\(\varpi_n(0)\le \bar\varpi<\infty\).
The gain representation bounds
\[
0\le\varpi_n(s)\le \bar\varpi,\qquad
0\le G_i^{\varpi_n}(s)\le(1+\bar b)\bar\varpi.
\]
Replacing each compensation by its canonical prefix retains feasibility
and bounds its absolute value by
\(d\cdot\bar b+(1+\bar b)(1+\bar\varpi)\).
We select a common subsequence converging at every rational point
and both endpoints. The limits at these points inherit the monotonicity
of \(\varpi_n\) and of the compensations \(\vartheta_n\), and we extend
them to monotone functions on \([0,\bar b]\) by one-sided limits.
At a continuity point \(s\) of both extensions, monotonicity squeezes
\(\varpi_n(s)\) and \(\vartheta_n(s)\) between their values at rational
points on either side of \(s\), so both sequences converge at \(s\).
A further diagonal selection at their at most countably many
discontinuities gives pointwise convergence everywhere to
\(\varpi_{\rm raw},\vartheta_{\rm raw}\).
Bounded convergence passes the integral inequalities to the limit.

We take right limits at \(s<\bar b\), keeping the values at \(\bar b\),
to obtain \(\varpi^+,\vartheta^+\).
Passing the raw inequalities along \(t\downarrow s\) leaves
\(\varpi_{\rm raw}\) in the positive kernel.
Replacing it by \(\varpi^+\le\varpi_{\rm raw}\) only decreases
the right side, so the regularized pair is feasible.
Leastness and the definition of the infimum give
\[
\mathcal J_d^{(2)}
\le\varpi_{d,\vartheta^+}(0)
\le\varpi^+(0)\le\varpi_{\rm raw}(0)=\mathcal J_d^{(2)}.
\]

\textit{The choice of \(\bar b\).}
We fix \(\bar b'>\bar b\) and show that the minimum in
\eqref{eq:two-functional} is unchanged when the operator is defined on
\([0,\bar b']\).
Because the buyers put no mass above \(\bar b\), the kernel integrals at
\(t\le\bar b\) are the same on both intervals.
A feasible pair on \([0,\bar b']\) therefore restricts to a feasible pair
on \([0,\bar b]\) with the same value at \(0\), since a future supremum
over a shorter range is no larger.
In the other direction, we extend a feasible pair on \([0,\bar b]\) by
\(\varpi(t)=\varpi(\bar b)\) and \(\vartheta(t)=\vartheta(\bar b)\) for
\(\bar b<t\le\bar b'\).
On \([\bar b,\bar b']\) we have \(\bar L_i=1\), and the kernel integral
vanishes. Since \(d>0\) and the pair is feasible at \(\bar b\), we obtain,
for \(\bar b<t\le\bar b'\),
\[
\mathcal A_{i,d,\vartheta}\varpi(t)
=1-d\cdot t+\epsilon_i\vartheta(\bar b)
\le1-d\cdot\bar b+\epsilon_i\vartheta(\bar b)
=\mathcal A_{i,d,\vartheta}\varpi(\bar b)\le\varpi(\bar b).
\]
The extended mass is nonnegative, nonincreasing, and right-continuous,
and the extended compensation is bounded, nondecreasing, and
right-continuous. The obstacles at \(t\le\bar b\) are unchanged, so the
extended pair is feasible on \([0,\bar b']\).
The two values are equal by leastness.

\textit{A bounded price interval and the reverse implication.}
We choose an optimal compensation and write
\(\varpi=\varpi_{d,\vartheta}\).
We put \(-\beta\,\mathrm d\varpi\) on \((0,\bar b]\)
and density \(\beta d\) on
\((\bar b,\bar b+\varpi(\bar b)/d)\).
Its mass is \(\beta\varpi(0)\le1\), and the remainder is put
at zero. For \(s\le\bar b\), the tail contributes
\(\beta\varpi(\bar b)\), exactly as the ideal tail does.
For \(s=\bar b+w\), the realized potential, with gain divided
by \(\beta\), differs from its value at \(\bar b\) by
\[
d\cdot w-\min\{d\cdot w,\varpi(\bar b)\}
=d\bigl(w-\varpi(\bar b)/d\bigr)_+\ge0.
\]
A single atom of the same mass at \(\bar b+\varpi(\bar b)/d\) gives the
difference \(d\cdot w-\varpi(\bar b)\mathbf1_{\{d\cdot w\ge\varpi(\bar b)\}}\ge0\)
instead.
Clipping seller values at \(\bar b\) preserves their order and weakly
decreases their potential sum. On the clipped domain,
\[
\psi_1^\varpi(s_1)+\psi_2^\varpi(s_2)
\ge-\vartheta(s_1)+\vartheta(s_2)\ge0
\qquad(s_1\le s_2).
\]
Multiplication by \(\beta\) proves \eqref{eq:two-price-target},
including \(\varpi(\bar b)=0\).

Conversely, we take a probability satisfying
\eqref{eq:two-price-target}, discard its zero-price mass,
and divide by \(\beta\).
For sellers in \([0,\bar b]\), prices above \(\bar b\)
can be collected into the ideal tail without changing gain.
The resulting cumulative mass \(\widetilde\varpi\) has a
nonnegative potential sum and a feasible canonical compensation.
Leastness gives
\[
\beta\mathcal J_d^{(2)}
\le\beta\varpi_{d,\Theta_{\widetilde\varpi}}(0)
\le\beta\widetilde\varpi(0)\le1.
\]

\textit{The welfare guarantee.}
Let \(d=(1-\beta)/\beta\), so that \(\beta(1+d)=1\), and let a price
\(\pi\) satisfy \eqref{eq:two-price-target}.
For bounded seller bodies \(Y_1\le Y_2\), independent of the buyers and
of the price, we put \((s_1,s_2)=(Y_1,Y_2)\) in
\eqref{eq:two-price-target} and take expectations.
By independence, the left side becomes the expected gain
\(\int\Gamma_0\,\mathrm d\pi\) from \eqref{eq:two-body}; since
\(s+\bar L_i(s)=1+\mathbb E\max\{Z_i,s\}\), the right side becomes
\(\beta\sum_i\mathbb E(1+\max\{Z_i,Y_i\})-\mathbb E(Y_1+Y_2)\).
Adding the initial welfare \(\mathbb E(Y_1+Y_2)\) to both sides bounds the
expected welfare at the price below by \(\beta\) times the normalized
efficient welfare:
\[
\mathbb E(Y_1+Y_2)+\int\Gamma_0\,\mathrm d\pi
\ge\beta\Bigl[2+\sum_{i=1}^2\mathbb E\max\{Z_i,Y_i\}\Bigr].
\]

\textit{Rounding the buyer bodies.}
We transfer a price for the rounded buyers and add tail mass to cover
the rounding loss. Since \(Z_i^\delta\le Z_i<Z_i^\delta+\delta\),
\[
0\le\bar L_i(s)-\bar L_i^\delta(s)\le\delta.
\]
For \(s<z\), the function
\((Z_i-s)\mathbf1_{\{z\le Z_i\}}\) is nondecreasing in \(Z_i\),
including its jump at \(z\).
Thus \(g_i(s,z)\ge g_i^\delta(s,z)\).
We take an optimal rounded price in its ideal-tail representation
and add tail mass \(\delta\).
Each potential increases by at least
\(-(\bar L_i-\bar L_i^\delta)+\delta\ge0\).
The same compensation is feasible at the additional mass \(\delta\).
Because \(\mathcal J_d^{(2)}\) does not depend on \(\bar b\), both values
are computed on \([0,\bar b]\), which also contains the rounded bodies.
The map \(x\mapsto\delta\lfloor x/\delta\rfloor\) is nondecreasing with
\(0\le\delta\lfloor x/\delta\rfloor\le x\) for \(x\ge0\), so the rounded
bodies are again ordered and take finitely many values.
Taking suprema and then letting \(\delta\downarrow0\) gives
\[
\sup_{\text{bounded ordered bodies}}\mathcal J_d^{(2)}
=\sup_{\text{finite-support ordered bodies}}\mathcal J_d^{(2)}.
\]
\end{proof}

\subsection{Finite-node evaluation}\label{app:two-recurrence}

Suppose the buyers are supported on nodes
\(0=v_0<\cdots<v_n=\bar b\).
Write \(p_{i\ell}=\Pr(Z_i=v_\ell)\),
\(\bar L_{ij}=\bar L_i(v_j)\), and
\(\vartheta_j=\vartheta(v_j)\).
For prescribed nondecreasing node compensations, the least solution
of the node inequalities is obtained backwards:
\begin{equation}\label{eq:two-backward}
\begin{aligned}
\varpi_n&=\max\{0,1-d\cdot\bar b-\vartheta_n,1-d\cdot\bar b+\vartheta_n\},\\
\varpi_j&=\max\left\{\varpi_{j+1},\
 \max_{i=1,2}\left[
  1-\frac{d\cdot v_j}{\bar L_{ij}}+\frac{\epsilon_i\vartheta_j}{\bar L_{ij}}
  +\sum_{\ell>j}
     \frac{(v_\ell-v_j)p_{i\ell}}{\bar L_{ij}}\varpi_\ell
 \right]\right\}.
\end{aligned}
\end{equation}
The value is piecewise affine and convex in
\((d,\vartheta_0,\ldots,\vartheta_n)\).
The price at \(v_j\) has normalized mass
\(\varpi_{j-1}-\varpi_j\).

For a pricing certificate, the node solution must also cover the
intervals between nodes. At the relevant budget
\(\varpi_0\le1+d\), this follows from the gain kernel.
Between consecutive nodes,
\[
(\psi_i^\varpi)'(s)
=d+H_i(s)-\int_{(s,\bar b]}\Pr(Z_i\ge z)\,
                              \widehat\pi_0(\mathrm dz)
\ge d-H_i(s)(\varpi_0-1)\ge0.
\]
The first inequality holds because \(\Pr(Z_i\ge z)\le H_i(s)\) for
\(z>s\) and \(\widehat\pi_0((s,\bar b])\le\varpi(s)\le\varpi_0\), and the
second because \(0\le H_i\le1\) and \(\varpi_0\le1+d\).
At a node, the potential drops by the price atom's mass times
\(\bar L_i(s)\). Its minimum on each interval \([v_j,v_{j+1})\) is
therefore its value at \(v_j\), and ordered node constraints
cover all ordered real seller values. At arbitrary budgets,
\eqref{eq:two-backward} is only a node optimization;
the full interval problem is \eqref{eq:two-operator}.

\paragraph{Sampling design for Section~\ref{sec:two-open}.}
At the numerical value \(d_{\mathrm{fam}}\approx0.371591\), we evaluated
20,000 finite-support ordered buyer pairs: 16,000 broad random samples
with 2--32 common support nodes and value scales between \(0.01\) and
\(100\), and 4,000 random perturbations of discretized reference buyers
with up to 65 atoms in each marginal. The broad samples had maximum
\(1.3645392\). The largest values agreed under cumulative-mass and
price-mass formulations of the inner linear program.

%% file: two_unit_family.tex
\section{Optimization within a two-unit family}\label{app:2fam}

This appendix optimizes a family with a deterministic second buyer body,
two endpoint masses for the second seller, and equal initial seller masses.
The family has a unique minimizing curve. Its value also gives a limiting
upper bound for ordered two-unit trade. A separate finite instance gives
a strict upper bound without using the optimization argument.

\subsection{Trading interpretation of the curve family}

We use the family and functionals defined before Proposition~\ref{prop:2fam-optimum}. The proof first identifies the fixed-endpoint energy minimizer, then excludes outer boundaries and isolates its stationary branch.

The trading interpretation uses the increasing map and body endpoint
\begin{equation}
 s_{\mathrm f}(\xi)=a_{\mathrm f}
       +\int_{c_{\mathrm f}}^\xi\mathcal P(r)^{-2}\,\mathrm dr,
 \qquad b_{\mathrm f}=s_{\mathrm f}(\xi_0).
 \label{eq:2fam-physical-time}
\end{equation}
The first buyer body has survival $1$ below $a_{\mathrm f}$ and
\begin{equation}
 H_{1,\mathrm f}(s_{\mathrm f}(\xi))=\mathcal P'_+(\xi),\qquad
 F_{1,\mathrm f}(s_{\mathrm f}(\xi))
       =N_{\mathrm f}\{\mathcal P(\xi)-\xi\mathcal P'_+(\xi)\}
 \quad(c_{\mathrm f}\le\xi<\xi_0).
 \label{eq:2fam-laws}
\end{equation}
The first seller has mass $m_{\mathrm f}$ at zero, has constant CDF
$m_{\mathrm f}$ between zero and $a_{\mathrm f}$, follows
\eqref{eq:2fam-laws} on the body interval, and is completed to mass one
at $b_{\mathrm f}$. If $a_{\mathrm f}=0$, coincident initial masses are
combined. The second buyer body is $Z_2=a_{\mathrm f}$, and the second
seller has law
$m_{\mathrm f}\delta_0+(1-m_{\mathrm f})\delta_{b_{\mathrm f}}$.
Endpoint atoms complete the first buyer survival to a probability law.
Coupling the bodies on each side by common quantiles gives the body
vector laws. A sequence of instances converges to $\mathcal P$ when its
buyer and seller vector laws converge weakly to them.

\begin{lemma}[Realization by limiting instances]\label{lem:2fam-realization}
Every member of $\mathfrak F$ is the limit of strictly positive,
finite-mean ordered two-unit instances with independent buyer and seller
vectors, whose optimal common-price welfare ratios converge to
$R_{\mathrm f}$. Their seller welfare and efficient gains converge to
$M_{\mathrm f}$ and $G_{\mathrm f}$, and their best price gains converge
to $2$.
\end{lemma}
\begin{proof}
We reconstruct the ordered body laws and separate coincident buyer and
seller atoms by a vanishing seller shift.

\textit{Body laws.} Concavity makes $H_{1,\mathrm f}$ nonincreasing and
$F_{1,\mathrm f}$ nondecreasing, and both are right-continuous because
$\mathcal P'_+$ is. The obstacles imply
$m_{\mathrm f}\le F_{1,\mathrm f}\le1$ and
$0\le H_{1,\mathrm f}\le1$. The common-quantile couplings give
$Z_1\ge Z_2$ and $Y_1\le Y_2$, and we draw the sides independently.
The limiting first-unit tail integral and integrated seller CDF are
\[
 \bar L_{1,\mathrm f}=\mathcal P^{-1},\qquad
 A_{1,\mathrm f}=N_{\mathrm f}\xi/\mathcal P,
 \qquad
 F_{1,\mathrm f}\bar L_{1,\mathrm f}
       +H_{1,\mathrm f}A_{1,\mathrm f}=N_{\mathrm f}.
\]

\textit{Positive approximations.} We give both buyer bodies weight
$1-\eta_{\mathrm f}$ and add a common buyer atom of mass
$\eta_{\mathrm f}$ at $1/\eta_{\mathrm f}$, so each escaping tail has
first moment one. We shift both sellers up by $\epsilon_{\mathrm f}>0$.
When $a_{\mathrm f}=0$, we also translate all values by a common
vanishing positive amount. These operations preserve the orders.

At a price $s$, the shifted first-unit gain is bounded by the unshifted
right-limit gain at $s-\epsilon_{\mathrm f}$: both the buyer stop-loss
and its acceptance probability decrease as the price moves up.
It is therefore at most $N_{\mathrm f}$ when
$s-\epsilon_{\mathrm f}\in[a_{\mathrm f},b_{\mathrm f})$.
For $a_{\mathrm f}<s<a_{\mathrm f}+\epsilon_{\mathrm f}$ the bound is
$m_{\mathrm f}(a_{\mathrm f}+1/p_{\mathrm f})\le N_{\mathrm f}$.
The second unit contributes at most $m_{\mathrm f}$ in these regions,
and $N_{\mathrm f}+m_{\mathrm f}=2$.
For prices $s\le a_{\mathrm f}$, write
$\bar z_{\mathrm f}=a_{\mathrm f}+1/p_{\mathrm f}-1$ for the first body mean.
The total gain is at most
$m_{\mathrm f}(2+\bar z_{\mathrm f}+a_{\mathrm f})=2$.
For $s\ge b_{\mathrm f}+\epsilon_{\mathrm f}$, only the escaping
buyer tails can trade; their combined first moment is two. Prices
above the shifted sellers and below the
remote buyer atom have gains tending to two. Integrating the body
states gives \eqref{eq:2fam-functionals}, so
\[
 \lim_{\eta_{\mathrm f},\epsilon_{\mathrm f}\downarrow0}
 \frac{\text{seller welfare}+\text{best price gains}}
      {\text{seller welfare}+\text{efficient gains}}
       =\frac{M_{\mathrm f}+2}{M_{\mathrm f}+G_{\mathrm f}}.
\]
\end{proof}

For a trial ratio $0<\beta_{\mathrm f}<1$, define
$d_{\mathrm f}=(1-\beta_{\mathrm f})/\beta_{\mathrm f}$ and
\begin{align}
 E_{\mathrm f}(\mathcal P)
   &=\int_{c_{\mathrm f}}^{\xi_0}
       \{\xi(\log\mathcal P)'^2+d_{\mathrm f}\mathcal P^{-2}\}
       \,\mathrm d\xi,\notag\\
 K_{\mathrm f}(\mathcal P)
   &=\frac{\beta_{\mathrm f}G_{\mathrm f}
         -(1-\beta_{\mathrm f})M_{\mathrm f}-2}
          {\beta_{\mathrm f}N_{\mathrm f}}\notag\\
   &=1-\frac{t_{\mathrm f}}{p_{\mathrm f}}-\log p_{\mathrm f}
       -d_{\mathrm f}a_{\mathrm f}
       -d_{\mathrm f}(1+t_{\mathrm f})
       +d_{\mathrm f}\xi_0-E_{\mathrm f}(\mathcal P).
 \label{eq:2fam-comparison}
\end{align}
Thus $R_{\mathrm f}<\beta_{\mathrm f}$ exactly when $K_{\mathrm f}>0$.
The normalization by $N_{\mathrm f}$ depends on $t_{\mathrm f}$; all outer
derivatives below are derivatives of $K_{\mathrm f}$ along outer
variations, in which one of $t_{\mathrm f},p_{\mathrm f},h_{\mathrm f}$
moves and the other two stay fixed. A derivative is one-sided where
only one direction of the variation is feasible.

\subsection{Compactness and the fixed-endpoint minimizer}

\begin{lemma}[Compactness below three quarters]\label{lem:2fam-compact}
Let $0<\beta_{\mathrm f}\le3/4$. Every curve with
$R_{\mathrm f}<\beta_{\mathrm f}$ satisfies
$1/32<m_{\mathrm f}<3/4$, $\xi_0<600$, and $p_{\mathrm f}<1$.
If $K_{\mathrm f}$ takes a positive value, its positive supremum is
attained in this parameter region.
\end{lemma}
\begin{proof}
We first exclude the parameter boundaries using body means, then use
compactness of concave curves on a common interval.

\textit{Parameter bounds.} The identity
$\bar z_{\mathrm f}+a_{\mathrm f}=2(1-m_{\mathrm f})/m_{\mathrm f}$
gives $G_{\mathrm f}\le2/m_{\mathrm f}$ and $R_{\mathrm f}\ge m_{\mathrm f}$.
Also $m_{\mathrm f}a_{\mathrm f}\le1-m_{\mathrm f}$,
$p_{\mathrm f}\ge m_{\mathrm f}/N_{\mathrm f}$, and $Q_{\mathrm f}\ge0$, so
\[
 G_{\mathrm f}\le4+2\log(2/m_{\mathrm f}),\qquad
 M_{\mathrm f}\ge(1-m_{\mathrm f})b_{\mathrm f}
       \ge\frac{(1-m_{\mathrm f})^2}{m_{\mathrm f}}.
\]
The last inequality uses
$b_{\mathrm f}\ge\bar z_{\mathrm f}\ge
(\bar z_{\mathrm f}+a_{\mathrm f})/2$.
For $m_{\mathrm f}\le1/32$ these estimates give
\begin{equation}
 \beta_{\mathrm f}G_{\mathrm f}-(1-\beta_{\mathrm f})M_{\mathrm f}-2
 \le\frac34\{4+2\log(2/m_{\mathrm f})\}
       -\frac{(1-m_{\mathrm f})^2}{4m_{\mathrm f}}-2<0.
 \label{eq:2fam-small-mass}
\end{equation}
The scalar expression increases on this interval, because its derivative is
$(1/4-3m_{\mathrm f}/2-m_{\mathrm f}^2/4)/m_{\mathrm f}^2\ge(831/4096)/m_{\mathrm f}^2$,
and interval evaluation shows that it is less than $-2694/10000$ at $1/32$.
When $K_{\mathrm f}>0$, the stronger estimate $G_{\mathrm f}\le2/m_{\mathrm f}$
gives $M_{\mathrm f}<184$. Since
$A_{1,\mathrm f}(b_{\mathrm f})=N_{\mathrm f}\xi_0\le b_{\mathrm f}$,
we have $M_{\mathrm f}>(5/16)\xi_0$ and hence $\xi_0<600$.
Also $p_{\mathrm f}\ge t_{\mathrm f}\ge1/63$ and $\xi_0\ge v_{\mathrm f}>0$.

\textit{A common compact domain.} We extend each curve leftwards by
$\xi+\delta_{\mathrm f}$ and rightwards by $h_{\mathrm f}\xi+e_{\mathrm f}$.
The extensions are concave and have slopes in $[0,1]$ on $[0,600]$.
We choose a uniformly convergent subsequence with convergent parameters.
The obstacles and endpoint conditions pass to the limit.
Derivatives of concave curves converge almost everywhere under uniform
convergence. The common positive lower bound on $\mathcal P$ then gives
dominated convergence in $T_{\mathrm f}$ and $Q_{\mathrm f}$.
Thus the comparison is continuous on this compact extension.
At $p_{\mathrm f}=1$, positive slopes force $\xi_0=c_{\mathrm f}=v_{\mathrm f}$,
whose limiting ratio is
\[
 R_{\mathrm f}=\frac{1+3t_{\mathrm f}^2}{(1+t_{\mathrm f})^2}
     =\frac34+\frac{(3t_{\mathrm f}-1)^2}{4(1+t_{\mathrm f})^2}\ge\frac34.
\]
\end{proof}

\begin{lemma}[Three-arc energy minimizer]\label{lem:2fam-shape}
At every feasible fixed $(t_{\mathrm f},p_{\mathrm f},\xi_0)$, the energy
$E_{\mathrm f}$ has a unique minimizer in \eqref{eq:2fam-class}.
It is entirely affine, or consists of an initial affine contact, one
strictly concave free interval, and a final affine contact.
Either contact may have zero length. On the free interval there is
a constant $C_{\mathrm f}>0$ such that
\begin{equation}
 \mathcal P''+\frac{C_{\mathrm f}}{\xi^2}\mathcal P=0,
 \qquad
 d_{\mathrm f}
  =\xi\mathcal P'^2-\mathcal P\mathcal P'
          +\frac{C_{\mathrm f}\mathcal P^2}{\xi}.
 \label{eq:2fam-euler}
\end{equation}
If $c_{\mathrm f}=0$, every non-affine minimizer has a positive initial
contact interval.
\end{lemma}
\begin{proof}
We minimize in the convex logarithmic obstacle class and then show that
its minimizer is concave in the original reciprocal state.

\textit{Existence in the obstacle class.} With
$z_{\mathrm f}=\log\mathcal P$, the energy is
$\int(\xi z_{\mathrm f}'^2+d_{\mathrm f}e^{-2z_{\mathrm f}})$
and the obstacle is $z_{\mathrm f}\le\log U_{\mathrm f}$.
Truncation below $\log p_{\mathrm f}$ decreases both terms.
For $c_{\mathrm f}>0$, bounded energy gives a bounded sequence in $H^1$;
the direct method and strict convexity give a unique minimizer.
For $c_{\mathrm f}=0$, we take weak limits on each compact interval
away from zero. The bounds $p_{\mathrm f}\le\mathcal P\le p_{\mathrm f}+\xi$
preserve its trace at zero, and lower semicontinuity gives a minimizer
of the weighted energy. Strict convexity still gives uniqueness.

\textit{Contact and curvature.} The variational inequality is
\begin{equation}
 \mu_{\mathrm f}=(\xi z_{\mathrm f}')'
                   +d_{\mathrm f}e^{-2z_{\mathrm f}}\ge0,
 \qquad\operatorname{supp}\mu_{\mathrm f}
        \subseteq\{\mathcal P=U_{\mathrm f}\}.
 \label{eq:2fam-multiplier}
\end{equation}
Away from zero, $\xi z_{\mathrm f}'$ has bounded variation.
At a smooth contact, the obstacle requires the left derivative of
$\mathcal P$ to be at least its affine slope and the right derivative
to be at most that slope. The measure in \eqref{eq:2fam-multiplier}
requires a nonnegative jump. Both derivatives must therefore equal
the obstacle slope. The same signs exclude contact at a genuine
downward corner of the obstacle.

On a free interval, we put $\omega_{\mathrm f}=\xi\mathcal P'/\mathcal P$
to integrate the Euler equation. It gives
\[
 \omega_{\mathrm f}'=-d_{\mathrm f}/\mathcal P^2,\qquad
 C_{\mathrm f}=\omega_{\mathrm f}-\omega_{\mathrm f}^2
                 +d_{\mathrm f}\xi/\mathcal P^2,
 \qquad C_{\mathrm f}'=0.
\]
The right endpoint slope is positive, while the left value of
$\omega_{\mathrm f}$ is less than one: this follows from the positive
intercept at a contact, or from the slope bound at $c_{\mathrm f}>0$.
Since $\omega_{\mathrm f}$ decreases, we obtain $0<\omega_{\mathrm f}<1$,
$C_{\mathrm f}>0$, and $\mathcal P''<0$ on the free interval.
A strictly concave interval cannot leave and return to the same affine
obstacle, because it lies above its chord. The first obstacle is active
before their intersection and the second afterwards, so at most one
free interval separates the initial and final contacts.
If a free interval started at zero, bounded positive $\mathcal P$
and the equation for $\omega_{\mathrm f}'$ would force $\omega_{\mathrm f}(0)=0$
and then $\omega_{\mathrm f}<0$, contradicting $\mathcal P\ge p_{\mathrm f}$.

\textit{Global comparison.} For the resulting minimizer $\mathcal P_*$,
write $\upsilon=\log(\mathcal P/\mathcal P_*)$.
Integration by parts in \eqref{eq:2fam-multiplier} gives
\begin{equation}
 E_{\mathrm f}(\mathcal P)-E_{\mathrm f}(\mathcal P_*)
  =\int_{c_{\mathrm f}}^{\xi_0}
    \{\xi\upsilon'^2+d_{\mathrm f}\mathcal P_*^{-2}
       (e^{-2\upsilon}-1+2\upsilon)\}\,\mathrm d\xi
       -2\int\upsilon\,\mathrm d\mu_{\mathrm f}\ge0.
 \label{eq:2fam-gap}
\end{equation}
\end{proof}

\subsection{Boundary exclusion and the stationary system}

The two entirely affine cases have the following rational functionals:
\begin{align}
 M_{\mathrm f,L}&=
 \frac{p_{\mathrm f}^2t_{\mathrm f}+p_{\mathrm f}t_{\mathrm f}^2-4p_{\mathrm f}t_{\mathrm f}+p_{\mathrm f}+t_{\mathrm f}}
      {p_{\mathrm f}t_{\mathrm f}(1+t_{\mathrm f})},
 &G_{\mathrm f,L}&=
 \frac{-p_{\mathrm f}^2+p_{\mathrm f}t_{\mathrm f}+5p_{\mathrm f}-t_{\mathrm f}}{p_{\mathrm f}(1+t_{\mathrm f})},
 \notag\\
 M_{\mathrm f,R}&=
 \frac{(p_{\mathrm f}-t_{\mathrm f})(t_{\mathrm f}-1)(p_{\mathrm f}t_{\mathrm f}-2p_{\mathrm f}+1)}
 {p_{\mathrm f}t_{\mathrm f}(1+t_{\mathrm f})(2p_{\mathrm f}-t_{\mathrm f}-1)},
 &G_{\mathrm f,R}&=
 \frac{p_{\mathrm f}t_{\mathrm f}+3p_{\mathrm f}-t_{\mathrm f}+1}{p_{\mathrm f}(1+t_{\mathrm f})}.
 \label{eq:2fam-affine-boundaries}
\end{align}
Their respective domains are $t_{\mathrm f}\le p_{\mathrm f}<1$ and
$(1+t_{\mathrm f})/2<p_{\mathrm f}<1$. On both domains,
$R_{\mathrm f}\ge73/100$ follows from exact polynomial positivity
after multiplying by positive denominators. Positive Bernstein expansions on rational subrectangles prove these inequalities;
the polynomial coefficients are provided in the supplementary material.

\begin{lemma}[Interior outer parameters]\label{lem:2fam-boundary}
Suppose $\beta_{\mathrm f}\in[18227/25000,729081/10^6]$.
At a positive global maximum of $K_{\mathrm f}$, or at a global
minimum of $R_{\mathrm f}$ whose value is $\beta_{\mathrm f}$, one has
$t_{\mathrm f}<p_{\mathrm f}<1$ and both affine contacts have positive length.
The derivatives of $K_{\mathrm f}$ with respect to
$(t_{\mathrm f},p_{\mathrm f},h_{\mathrm f})$ vanish there.
\end{lemma}
\begin{proof}
We reduce both cases to maximization of the comparison, then exclude
each boundary by a feasible outer variation. Since
\[
 K_{\mathrm f}
 =\frac{M_{\mathrm f}+G_{\mathrm f}}{N_{\mathrm f}}
    \left(1-\frac{R_{\mathrm f}}{\beta_{\mathrm f}}\right),
\]
a ratio minimizer of value $\beta_{\mathrm f}$ maximizes
$K_{\mathrm f}$ at zero. In either case, we replace the curve by its
fixed-endpoint energy minimizer, which maximizes $K_{\mathrm f}$
at those endpoints by \eqref{eq:2fam-comparison}.
The bounds following \eqref{eq:2fam-affine-boundaries} force a nonempty free interval.
We use variations affine on contact intervals and blend them inside a
compact part of the strictly concave free interval; its strict curvature
and obstacle gap preserve feasibility for small changes.

\textit{Vanishing initial contact.} If $c_{\mathrm f}>0$, put
$H_{\mathrm f,L}=\mathcal P'(c_{\mathrm f}+)\le1$.
At fixed $(t_{\mathrm f},\xi_0)$ the endpoint variation is
\begin{equation}
 \partial_{p_{\mathrm f}}K_{\mathrm f}
 =-\frac{c_{\mathrm f}(H_{\mathrm f,L}-2)^2}{2p_{\mathrm f}^2}<0.
 \label{eq:2fam-left-variation}
\end{equation}
For a negative change $\mathrm dp_{\mathrm f}$ and $H_{\mathrm f,L}<1$, we extend
the free curve slightly leftwards and change its initial value by
$(1-H_{\mathrm f,L}/2)\mathrm dp_{\mathrm f}+O(\mathrm dp_{\mathrm f}^2)$.
The strict slope inequality preserves the initial obstacle.
When $H_{\mathrm f,L}=1$, we insert a short tangent segment and shift
the adjacent curve by $\mathrm dp_{\mathrm f}/2$. Each change blends to zero inside
the free interval and increases the comparison. The case
$c_{\mathrm f}=0$ already has positive initial contact by
Lemma~\ref{lem:2fam-shape}.

\textit{Vanishing final contact.} Put
$H_{\mathrm f,R}=\mathcal P'(\xi_0-)\ge h_{\mathrm f}>0$.
Increasing $\xi_0$ gives
\begin{equation}
 \frac{\mathrm dK_{\mathrm f}}{\mathrm d\xi_0}=\xi_0H_{\mathrm f,R}^2>0.
 \label{eq:2fam-right-variation}
\end{equation}
We lower the old endpoint by $H_{\mathrm f,R}\,\mathrm d\xi_0$ and append its
tangent segment. The new obstacle slope
$v_{\mathrm f}/(\xi_0+\mathrm d\xi_0)$ is smaller than $h_{\mathrm f}$, so
the appended segment is feasible. The same shift lies below the new
obstacle near the old endpoint; farther away, we blend it through the
strict obstacle gap. The unchanged tangent slope preserves concavity.

\textit{Parameter faces.} At $p_{\mathrm f}=t_{\mathrm f}$, the initial contact
has positive length. If $\mathcal P_\ell$ is its terminal value, then
\[
 2\partial_{p_{\mathrm f}}K_{\mathrm f}
  =(t_{\mathrm f}+d_{\mathrm f})
       (t_{\mathrm f}^{-2}-\mathcal P_\ell^{-2})>0.
\]
We realize the increase by a shift $\mathrm dp_{\mathrm f}/2$ on the initial contact
and an interior blend. Lemma~\ref{lem:2fam-compact} excludes
$p_{\mathrm f}=1$ and the noncompact boundaries. At all remaining parameters,
the same contact variations and interior blends work in both directions,
so
\[
 \partial_{t_{\mathrm f}}K_{\mathrm f}
 =\partial_{p_{\mathrm f}}K_{\mathrm f}
 =\partial_{h_{\mathrm f}}K_{\mathrm f}=0.
\]
\end{proof}

Write $\xi_\ell<\xi_r$ for the contact times,
$\mathcal P_\ell=\xi_\ell+\delta_{\mathrm f}$, and
$\mathcal P_r=h_{\mathrm f}\xi_r+e_{\mathrm f}$.
Define the two dimensionless contact slopes
$q_{\ell,\mathrm f}=\xi_\ell/\mathcal P_\ell$ and
$q_{r,\mathrm f}=h_{\mathrm f}\xi_r/\mathcal P_r$.
Integration of \eqref{eq:2fam-euler} gives
\begin{align}
 T_{\mathrm f}&=p_{\mathrm f}^{-1}-\mathcal P_\ell^{-1}
      +(q_{\ell,\mathrm f}-q_{r,\mathrm f})/d_{\mathrm f}
      +(\mathcal P_r^{-1}-1)/h_{\mathrm f},\notag\\
 Q_{\mathrm f}&=-\log p_{\mathrm f}-C_{\mathrm f}\log(\xi_r/\xi_\ell)
                    -\delta_{\mathrm f}/p_{\mathrm f}+e_{\mathrm f},\notag\\
 G_{\mathrm f}&=3-m_{\mathrm f}+m_{\mathrm f}a_{\mathrm f}
                       +N_{\mathrm f}C_{\mathrm f}\log(\xi_r/\xi_\ell).
 \label{eq:2fam-quadratures}
\end{align}
The three outer equations are
\begin{equation}
\begin{aligned}
 0=2\partial_{p_{\mathrm f}}K_{\mathrm f}
   &=\frac{t_{\mathrm f}+d_{\mathrm f}}{p_{\mathrm f}^2}
        -\frac{\delta_{\mathrm f}+d_{\mathrm f}}{\mathcal P_\ell^2},\\
 0=\partial_{h_{\mathrm f}}K_{\mathrm f}
   &=-\xi_0^2h_{\mathrm f}
     +\frac{2(h_{\mathrm f}e_{\mathrm f}+d_{\mathrm f})}{h_{\mathrm f}^2}
       \left\{\mathcal P_r^{-1}-1
          +\frac{e_{\mathrm f}}2(1-\mathcal P_r^{-2})\right\},\\
 0=\partial_{t_{\mathrm f}}K_{\mathrm f}
   &=-\frac1{2p_{\mathrm f}}+\frac{d_{\mathrm f}}{2t_{\mathrm f}^2}
      -d_{\mathrm f}
      -\frac{\delta_{\mathrm f}+d_{\mathrm f}}{2\mathcal P_\ell^2}
      -\frac{\xi_0h_{\mathrm f}}2
      +\frac{h_{\mathrm f}e_{\mathrm f}+d_{\mathrm f}}{2h_{\mathrm f}}
         (\mathcal P_r^{-2}-1).
\end{aligned}\label{eq:2fam-stationarity}
\end{equation}

\subsection{The stationary branch and the family minimum}

We eliminate the contact parameters using
\eqref{eq:2fam-stationarity}. From here on, $\mathcal P_\ell$, $\xi_\ell$, and
$C_{\mathrm f}$ denote the functions of $(t_{\mathrm f},p_{\mathrm f})$ defined in
the first line below; at a stationary point they agree with the contact
data above, by \eqref{eq:2fam-euler} and the first equation of
\eqref{eq:2fam-stationarity}. Define
\begin{gather}
 \mathcal P_\ell=p_{\mathrm f}
       \sqrt{\frac{\delta_{\mathrm f}+d_{\mathrm f}}
                    {t_{\mathrm f}+d_{\mathrm f}}},\qquad
 \xi_\ell=\mathcal P_\ell-\delta_{\mathrm f},\qquad
 C_{\mathrm f}
       =\frac{(\delta_{\mathrm f}+d_{\mathrm f})\xi_\ell}
                    {\mathcal P_\ell^2},\notag\\
 v_{\mathrm f}^2(q_{r,\mathrm f}^2-q_{r,\mathrm f}+C_{\mathrm f})
       =C_{\mathrm f}q_{r,\mathrm f}^2,\qquad 0<q_{r,\mathrm f}<v_{\mathrm f},
 \notag\\
 \gamma_{\mathrm f}=-d_{\mathrm f}/t_{\mathrm f}^2+2d_{\mathrm f}
          +1/p_{\mathrm f}+(t_{\mathrm f}+d_{\mathrm f})/p_{\mathrm f}^2,\qquad
 \mathcal A_{\mathrm f}(t_{\mathrm f},p_{\mathrm f})
       =-\gamma_{\mathrm f}
        +\frac{v_{\mathrm f}(v_{\mathrm f}-q_{r,\mathrm f})}
               {e_{\mathrm f}(v_{\mathrm f}+q_{r,\mathrm f})}=0.
 \label{eq:2fam-algebraic}
\end{gather}
For every $d_{\mathrm f}>0$ we call the set
$0<t_{\mathrm f}<p_{\mathrm f}<1$ the algebraic branch. There
$\mathcal P_\ell>p_{\mathrm f}>\delta_{\mathrm f}$, so $\xi_\ell>0$ and
$C_{\mathrm f}>0$.
The equation for $q_{r,\mathrm f}$ has one root in $(0,v_{\mathrm f})$,
because $v_{\mathrm f}^2q_{r,\mathrm f}(1-q_{r,\mathrm f})/
(v_{\mathrm f}^2-q_{r,\mathrm f}^2)$ increases from zero to infinity.
With $\mathcal D_{\mathrm f}(u)=u^2-u+C_{\mathrm f}$, define
\begin{gather}
 \mathcal I_{\mathrm f}=\int_{q_{r,\mathrm f}}^{q_{\ell,\mathrm f}}
          \frac{\mathrm du}{\mathcal D_{\mathrm f}(u)},\notag\\
 \mathcal F_{\mathrm f}(t_{\mathrm f},p_{\mathrm f})
 =\log\frac{1-q_{r,\mathrm f}}{1-q_{\ell,\mathrm f}}
    +\frac12\log\frac{\mathcal D_{\mathrm f}(q_{\ell,\mathrm f})}
                         {\mathcal D_{\mathrm f}(q_{r,\mathrm f})}
    +\frac{\mathcal I_{\mathrm f}}2-\log\frac{e_{\mathrm f}}{\delta_{\mathrm f}}=0.
 \label{eq:2fam-connection}
\end{gather}
The recovered parameters are
\begin{equation}
 h_{\mathrm f}
   =\frac{d_{\mathrm f}q_{r,\mathrm f}(1-q_{r,\mathrm f})}
               {e_{\mathrm f}\mathcal D_{\mathrm f}(q_{r,\mathrm f})},\qquad
 \mathcal P_r=\frac{e_{\mathrm f}}{1-q_{r,\mathrm f}},\qquad
 \xi_r=\frac{q_{r,\mathrm f}\mathcal P_r}{h_{\mathrm f}},\qquad
 \xi_0=\frac{v_{\mathrm f}}{h_{\mathrm f}}.
 \label{eq:2fam-recovery}
\end{equation}
A physical stationary point means that these formulas produce a curve
in \eqref{eq:2fam-class}, with
$c_{\mathrm f}<\xi_\ell<\xi_r<\xi_0$,
$0<q_{r,\mathrm f}<q_{\ell,\mathrm f}<1$, and positive
$\mathcal D_{\mathrm f}$ on the integration interval.

\begin{lemma}[Stationary domain and monotonicity]\label{lem:2fam-domain}
For $37/100\le d_{\mathrm f}\le3/8$, every physical stationary
point satisfies $1/5<t_{\mathrm f}<421/1000$.
On the algebraic branch, $\partial_{p_{\mathrm f}}\mathcal A_{\mathrm f}>0$.
\end{lemma}
\begin{proof}
Write $\chi_{\mathrm f}(C_{\mathrm f})=
v_{\mathrm f}(v_{\mathrm f}-q_{r,\mathrm f})/
[e_{\mathrm f}(v_{\mathrm f}+q_{r,\mathrm f})]$.
Differentiation gives
\begin{align*}
 \partial_{p_{\mathrm f}}C_{\mathrm f}
 &=\frac{d_{\mathrm f}+t_{\mathrm f}}{p_{\mathrm f}^3}
   \left(t_{\mathrm f}+\frac{p_{\mathrm f}}2-\mathcal P_\ell
        +\frac{\mathcal P_\ell p_{\mathrm f}}
                   {4(d_{\mathrm f}+\delta_{\mathrm f})}\right)
 <\frac{t_{\mathrm f}(d_{\mathrm f}+t_{\mathrm f})}{p_{\mathrm f}^3},\\
 \chi_{\mathrm f}'
 &=-\frac{2(v_{\mathrm f}-q_{r,\mathrm f})^2}
 {e_{\mathrm f}\{v_{\mathrm f}^2(1-2q_{r,\mathrm f})+q_{r,\mathrm f}^2\}},
 \qquad -2/e_{\mathrm f}<\chi_{\mathrm f}'<0.
\end{align*}
If $\partial_{p_{\mathrm f}}C_{\mathrm f}<0$, its contribution to
$\partial_{p_{\mathrm f}}\mathcal A_{\mathrm f}
=p_{\mathrm f}^{-2}+2(d_{\mathrm f}+t_{\mathrm f})p_{\mathrm f}^{-3}
+\chi_{\mathrm f}'\partial_{p_{\mathrm f}}C_{\mathrm f}$ is positive.
Otherwise these two bounds and $e_{\mathrm f}>t_{\mathrm f}$
give the same strict sign.

We exclude $t_{\mathrm f}\le1/5$ using $\gamma_{\mathrm f}>0$.
Expansion at $p_{\mathrm f}=7t_{\mathrm f}/4+u$ has positive coefficients in
$u\ge0$, so $p_{\mathrm f}<7t_{\mathrm f}/4$.
It follows that $\delta_{\mathrm f}/p_{\mathrm f}>11/14$ and
$\mathcal P_\ell/p_{\mathrm f}\le\sqrt{43/38}<16/15$, giving
$q_{\ell,\mathrm f}<59/224<3/10$.
But $d_{\mathrm f}/(d_{\mathrm f}+\delta_{\mathrm f})
\ge74/129>9/16$ and $v_{\mathrm f}\ge2/5$.
The order $q_{r,\mathrm f}<q_{\ell,\mathrm f}$ requires
$q_{\ell,\mathrm f}^2>
d_{\mathrm f}v_{\mathrm f}^2/
(d_{\mathrm f}+\delta_{\mathrm f})$, a contradiction.
For the upper bound, $\gamma_{\mathrm f}<v_{\mathrm f}/e_{\mathrm f}$
and $p_{\mathrm f}<1$ imply
\[
 -\frac3{8t_{\mathrm f}^2}+\frac{111}{100}+1+t_{\mathrm f}
            -\frac{1-t_{\mathrm f}}{1+t_{\mathrm f}}<0.
\]
The expression is increasing and already positive at $421/1000$, so
\[
 \frac15<t_{\mathrm f}<\frac{421}{1000}.
\]
\end{proof}

\begin{lemma}[Validated stationary branch]\label{lem:2fam-cover}
For every $\beta_{\mathrm f}\in[18227/25000,729081/10^6]$, the system
\eqref{eq:2fam-algebraic}--\eqref{eq:2fam-connection} has exactly one
physical stationary point. At $\beta_{\mathrm f}=18227/25000$ it satisfies
$\beta_{\mathrm f}G_{\mathrm f}-(1-\beta_{\mathrm f})M_{\mathrm f}-2<0$.
\end{lemma}
\begin{proof}
We first isolate the reduced root and then reconstruct its curve.

\textit{The reduced root.} We use Lemma~\ref{lem:2fam-domain} to reduce
the covering problem to one dimension. Strict monotonicity of
$\mathcal A_{\mathrm f}$ in $p_{\mathrm f}$
encloses its unique root over each rational $t_{\mathrm f}$ interval.
At $p_{\mathrm f}=t_{\mathrm f}$, its sign is negative; signs at bracket endpoints
exclude the remaining price domain. Each discarded interval either has
no algebraic root, violates a physical phase condition, or has a
nonzero enclosure of $\mathcal F_{\mathrm f}$.
Every interval that passes the first two tests has an enclosure of
$C_{\mathrm f}$ above $1/4$, so $\mathcal D_{\mathrm f}$ has no real root
there and the enclosures of $\mathcal F_{\mathrm f}$ evaluate
$\mathcal I_{\mathrm f}$ through its arctan antiderivative.
The retained intervals are contiguous and have
\[
 \partial_{t_{\mathrm f}}\mathcal F_{\mathrm f}
 -\partial_{p_{\mathrm f}}\mathcal F_{\mathrm f}
   \frac{\partial_{t_{\mathrm f}}\mathcal A_{\mathrm f}}
        {\partial_{p_{\mathrm f}}\mathcal A_{\mathrm f}}>0.
\]
Opposite signs at the two outer endpoints give existence and uniqueness
along this branch. A finite rational covering gives these sign enclosures
throughout the trial interval. The covering and its interval bounds are
provided in the supplementary material.

\textit{Reconstructing the curve.} Interval evaluation of
\eqref{eq:2fam-algebraic} and \eqref{eq:2fam-recovery} over the
retained root enclosure gives
\[
 c_{\mathrm f}<\xi_\ell<\xi_r<\xi_0,\qquad
 0<q_{r,\mathrm f}<q_{\ell,\mathrm f}<1,\qquad
 0<h_{\mathrm f}<1,\qquad C_{\mathrm f}>\tfrac14;
\]
the corresponding interval bounds are supplied with the covering.
In particular, $\mathcal D_{\mathrm f}(u)>0$ for every real $u$.
For $q_{r,\mathrm f}\le u\le q_{\ell,\mathrm f}$, we define
\[
 \xi(u)=\xi_\ell\exp\!\left(
   \int_u^{q_{\ell,\mathrm f}}\frac{\mathrm dv}{\mathcal D_{\mathrm f}(v)}
                         \right),\qquad
 \mathcal P(\xi(u))
   =\sqrt{\frac{d_{\mathrm f}\xi(u)}{\mathcal D_{\mathrm f}(u)}}.
\]
Differentiation gives
\[
 \frac{\mathrm d\xi}{\mathrm du}=-\frac{\xi}{\mathcal D_{\mathrm f}(u)},\qquad
 \mathcal P'=\frac{u\mathcal P}{\xi},\qquad
 \mathcal P''=-\frac{C_{\mathrm f}\mathcal P}{\xi^2}<0.
\]
The identity
$\mathcal D_{\mathrm f}(q_{\ell,\mathrm f})
 =d_{\mathrm f}\xi_\ell/\mathcal P_\ell^2$ gives
$\mathcal P(\xi_\ell)=\mathcal P_\ell$ and
$\mathcal P'(\xi_\ell)=1$. Since
$\delta_{\mathrm f}=\mathcal P_\ell(1-q_{\ell,\mathrm f})$,
the connection equation gives
\[
 \mathcal P(\xi(q_{r,\mathrm f}))=\mathcal P_r,\qquad
 \xi(q_{r,\mathrm f})
 =\frac{\mathcal P_r^2\mathcal D_{\mathrm f}(q_{r,\mathrm f})}
        {d_{\mathrm f}}=\xi_r,\qquad
 \mathcal P'(\xi_r)=h_{\mathrm f}.
\]
We extend the curve by $\xi+\delta_{\mathrm f}$ on
$[c_{\mathrm f},\xi_\ell]$ and by
$h_{\mathrm f}\xi+e_{\mathrm f}$ on $[\xi_r,\xi_0]$.
The resulting curve is positive, continuously differentiable, concave,
and Lipschitz, with endpoint values $p_{\mathrm f}$ and $1$.
Both obstacles are tangent lines to this concave curve, so it lies
below both. Substitution into \eqref{eq:2fam-stationarity} gives
\[
 \partial_{p_{\mathrm f}}K_{\mathrm f}=0,\qquad
 \partial_{h_{\mathrm f}}K_{\mathrm f}
 =\frac{d_{\mathrm f}
    \{v_{\mathrm f}^2\mathcal D_{\mathrm f}(q_{r,\mathrm f})
       -C_{\mathrm f}q_{r,\mathrm f}^2\}}
       {e_{\mathrm f}h_{\mathrm f}^2
        \mathcal D_{\mathrm f}(q_{r,\mathrm f})}=0,\qquad
 2\partial_{t_{\mathrm f}}K_{\mathrm f}
       =\mathcal A_{\mathrm f}(t_{\mathrm f},p_{\mathrm f})=0.
\]

\textit{The lower trial value.} At the lower endpoint the covering gives
\[
 \beta_{\mathrm f}G_{\mathrm f}-(1-\beta_{\mathrm f})M_{\mathrm f}-2
     \in[-2.116\cdot10^{-6},-1.157\cdot10^{-6}]<0.
\]
\end{proof}

\begin{lemma}[A rational comparison curve]\label{lem:2fam-witness}
The class \(\mathfrak F\) contains a rational concave polygon
with \(R_{\mathrm f}<729081/10^6\).
\end{lemma}
\begin{proof}
We use the rational polygon specified in Appendix~\ref{app:two-finite}.
Exact segment integration with its rational vertices gives
\[
R_{\mathrm f}<729081/10^6.
\]
\end{proof}

\begin{proof}[Proof of Proposition~\ref{prop:2fam-optimum}]
We first show that stationary curves solve the reduced system, then
place the family minimum in the validated trial interval, and use
stationary-branch uniqueness to identify it globally.

\textit{Stationary curves.} We call a three-arc class curve interior
stationary when both contacts have positive length and
\eqref{eq:2fam-stationarity} holds, and we show that it gives a
physical stationary point. By the proof of Lemma~\ref{lem:2fam-shape},
the curve meets its contacts with slopes $1$ and $h_{\mathrm f}$, and on
the free interval $\omega_{\mathrm f}=\xi\mathcal P'/\mathcal P$ satisfies
$\omega_{\mathrm f}'=-d_{\mathrm f}/\mathcal P^2$ and decreases from
$q_{\ell,\mathrm f}$ to $q_{r,\mathrm f}$ inside $(0,1)$.
Multiplying the first integral in \eqref{eq:2fam-euler} by
$\xi/\mathcal P^2$ gives
$\mathcal D_{\mathrm f}(\omega_{\mathrm f})=d_{\mathrm f}\xi/\mathcal P^2>0$,
so
\[
 \frac{\mathrm d\xi}{\xi}
   =-\frac{\mathrm d\omega_{\mathrm f}}{\mathcal D_{\mathrm f}(\omega_{\mathrm f})}.
\]
Integration from $\xi_\ell$ to $\xi_r$ gives
$\log(\xi_r/\xi_\ell)=\mathcal I_{\mathrm f}$. At both contacts
$\mathcal P^2=d_{\mathrm f}\xi/\mathcal D_{\mathrm f}(\omega_{\mathrm f})$,
and $\delta_{\mathrm f}=\mathcal P_\ell(1-q_{\ell,\mathrm f})$,
$e_{\mathrm f}=\mathcal P_r(1-q_{r,\mathrm f})$; taking the logarithm of
$\mathcal P_r^2/\mathcal P_\ell^2$ gives \eqref{eq:2fam-connection}.
The first equation of \eqref{eq:2fam-stationarity} and the first
integral at $\xi_\ell$ give the first line of \eqref{eq:2fam-algebraic}.
The first integral at $\xi_r$ gives the formula for $h_{\mathrm f}$ in
\eqref{eq:2fam-recovery}, whose other three formulas restate the
definitions of $\mathcal P_r$, $q_{r,\mathrm f}$, and $h_{\mathrm f}$.
The substitution in the proof of Lemma~\ref{lem:2fam-cover} then turns
the other two equations of \eqref{eq:2fam-stationarity} into the
equation for $q_{r,\mathrm f}$ and $\mathcal A_{\mathrm f}=0$. That
equation reads
$C_{\mathrm f}(v_{\mathrm f}^2-q_{r,\mathrm f}^2)
 =v_{\mathrm f}^2q_{r,\mathrm f}(1-q_{r,\mathrm f})>0$,
so $q_{r,\mathrm f}<v_{\mathrm f}$ and $q_{r,\mathrm f}$ is the root in
$(0,v_{\mathrm f})$. The contacts give
$c_{\mathrm f}<\xi_\ell<\xi_r<\xi_0$, the decrease of
$\omega_{\mathrm f}$ gives $0<q_{r,\mathrm f}<q_{\ell,\mathrm f}<1$, and
$\mathcal D_{\mathrm f}(\omega_{\mathrm f})>0$ gives positive
$\mathcal D_{\mathrm f}$ on $[q_{r,\mathrm f},q_{\ell,\mathrm f}]$.
The reconstruction in the proof of Lemma~\ref{lem:2fam-cover} integrates
the same two identities and returns this curve.

\textit{The lower trial value.} Set
$\beta_{\mathrm f}=18227/25000$.
If $K_{\mathrm f}$ had a positive value, Lemma~\ref{lem:2fam-compact}
would give an attained positive maximum.
Lemmas~\ref{lem:2fam-shape} and~\ref{lem:2fam-boundary} would make
it an interior stationary three-arc curve, hence a physical stationary
point by the first step, contradicting Lemma~\ref{lem:2fam-cover}.
Equality $R_{\mathrm f}=\beta_{\mathrm f}$
is also impossible: it would maximize $K_{\mathrm f}$ at zero and the
same boundary and stationary argument would apply.

\textit{Attainment and uniqueness.} Lemma~\ref{lem:2fam-witness}
provides a curve below $\beta_{\mathrm f}=729081/10^6$.
We choose a minimizing sequence below this curve's ratio.
Its ratios are below $3/4$, so Lemma~\ref{lem:2fam-compact} places the
parameters in $1/32<m_{\mathrm f}<3/4$, $\xi_0<600$, $p_{\mathrm f}<1$.
The extraction in the proof of that lemma applies to this sequence: a
subsequence converges to a curve in $\mathfrak F$ along which
$R_{\mathrm f}$ converges, and its ratio $r_{\mathrm{fam}}<3/4$ excludes
$p_{\mathrm f}=1$, where $R_{\mathrm f}\ge3/4$. So the minimum
$r_{\mathrm{fam}}$ is attained.
At the trial value $\beta_{\mathrm f}=r_{\mathrm{fam}}$, every minimizing
curve maximizes $K_{\mathrm f}$ at zero. The shape and boundary lemmas
make it an interior stationary three-arc curve, hence a physical
stationary point by the first step. Lemma~\ref{lem:2fam-cover} makes
this point unique, and \eqref{eq:2fam-recovery} fixes the parameter
triple. By \eqref{eq:2fam-comparison}, a maximizer of $K_{\mathrm f}$
minimizes $E_{\mathrm f}$ at its own parameters, so the uniqueness in
Lemma~\ref{lem:2fam-shape} fixes the curve.

\textit{Self-consistency.} Suppose a physical stationary solution in
the trial interval has $K_{\mathrm f}=0$.
If a different curve had $R_{\mathrm f}<\beta_{\mathrm f}$, the attained
positive maximum of $K_{\mathrm f}$ would give another physical
stationary point, by Lemmas~\ref{lem:2fam-shape}
and~\ref{lem:2fam-boundary} and the first step.
Lemma~\ref{lem:2fam-cover} excludes this possibility.
Thus every self-consistent solution is the family minimum.
Its positive contacts put it in the endpoint-slope subclass, and
\[
 \inf_{\mathfrak F}R_{\mathrm f}
 =\inf_{\substack{\mathfrak F\text{ with endpoint slopes }1,h_{\mathrm f}}}
       R_{\mathrm f}
 =r_{\mathrm{fam}}\in
       \left(\frac{18227}{25000},\frac{729081}{10^6}\right).
\]
\end{proof}

\subsection{A price certificate with either reference side fixed}
\label{app:2fam-one-sided}

We now fix the minimizing curve of Proposition~\ref{prop:2fam-optimum}, put
$\beta_{\mathrm f}=r_{\mathrm{fam}}$, and denote its body laws by
$(Z^*,Y^*)$. This subsection uses the limiting model with escaping buyer
first moment one per unit. For ordered finite-mean body vectors $Z$ and
seller vectors $Y$, drawn independently, define
\begin{align}
 \mathcal O_{\mathrm f}(Z,Y)
   &=\sum_{i=1}^2\mathbb E[Y_i+1+(Z_i-Y_i)_+],\notag\\
 \mathcal W_{\mathrm f,\pi}(Z,Y)
   &=\sum_{i=1}^2\mathbb EY_i
     +\int_0^\infty\sum_{i=1}^2
       \mathbb E\!\left[
          \mathbf1_{Y_i<s}\{1+(Z_i-Y_i)\mathbf1_{Z_i\ge s}\}
       \right]\,\pi(\mathrm ds).
 \label{eq:2fam-limiting-welfare}
\end{align}
The strict seller convention agrees with the vanishing positive seller
shift; the atomless pricing measure also makes either seller convention
equivalent after averaging.

Let $b_{\ell,\mathrm f}=s_{\mathrm f}(\xi_\ell)$ and
$b_{r,\mathrm f}=s_{\mathrm f}(\xi_r)$.
Define the low-price and tail masses and the low-price coefficients by
\begin{align}
 \pi_{a,\mathrm f}&=\beta_{\mathrm f}
      \{1-a_{\mathrm f}(d_{\mathrm f}+t_{\mathrm f})\},
 &r_{a,\mathrm f}&=\beta_{\mathrm f}(d_{\mathrm f}+t_{\mathrm f}),
 \notag\\
 \pi_{T,\mathrm f}&=1-\pi_{a,\mathrm f}
  -\beta_{\mathrm f}C_{\mathrm f}
    \{2-c_{\mathrm f}/\xi_\ell-\xi_r/\xi_0+\log(\xi_r/\xi_\ell)\},
 \notag\\
 \lambda_{0,\mathrm f}&=2\pi_{a,\mathrm f}/a_{\mathrm f}-r_{a,\mathrm f},
 &\lambda_{1,\mathrm f}&=2(r_{a,\mathrm f}-\pi_{a,\mathrm f}/a_{\mathrm f})/a_{\mathrm f}.
 \label{eq:2fam-price-coefficients}
\end{align}
Set $\overline p_{\mathrm f}=b_{\mathrm f}+\pi_{T,\mathrm f}/(1-\beta_{\mathrm f})$.
The price density is zero outside $(0,\overline p_{\mathrm f})$ and is
\begin{equation}
 f_{\pi,\mathrm f}(s)=
 \begin{cases}
 \lambda_{0,\mathrm f}+\lambda_{1,\mathrm f}s,&0<s<a_{\mathrm f},\\
 \beta_{\mathrm f}C_{\mathrm f}\mathcal P^2/\xi_\ell,
       &a_{\mathrm f}<s<b_{\ell,\mathrm f},\\
 \beta_{\mathrm f}C_{\mathrm f}\mathcal P^2/\xi,
       &b_{\ell,\mathrm f}<s<b_{r,\mathrm f},\\
 \beta_{\mathrm f}C_{\mathrm f}\xi_r\mathcal P^2/\xi^2,
       &b_{r,\mathrm f}<s<b_{\mathrm f},\\
 (1-\beta_{\mathrm f}),&b_{\mathrm f}<s<\overline p_{\mathrm f},
 \end{cases}
 \label{eq:2fam-price-density}
\end{equation}
where $\xi=s_{\mathrm f}^{-1}(s)$ on the body interval.
Its coefficients are positive and its integral is one.
The resulting probability is $\pi_{\mathrm f}(\mathrm ds)=f_{\pi,\mathrm f}(s)\,\mathrm ds$.

For the reference states, let $F_i$ be the seller CDF, $H_i$ the buyer
body survival, $A_i(s)=\mathbb E(s-Y_i^*)_+$, and
$\bar L_i(s)=1+\mathbb E(Z_i^*-s)_+$.
Set
\begin{gather}
 \alpha_i(s)=\int_0^sF_i(r)f_{\pi,\mathrm f}(r)\,\mathrm dr,
 \qquad \Lambda_i(s)=\int_s^\infty H_i(r)f_{\pi,\mathrm f}(r)\,\mathrm dr,\notag\\
 g_i^S=f_{\pi,\mathrm f}\bar L_i+\Lambda_i
                  -\beta_{\mathrm f}H_i-(1-\beta_{\mathrm f}),\qquad
 g_i^B=f_{\pi,\mathrm f}A_i+\alpha_i-\beta_{\mathrm f}F_i.
 \label{eq:2fam-state-gradients}
\end{gather}
The seller and buyer measure potentials, respectively, are
\begin{equation}
 \psi_{1,\mathrm f}(s)=-\int_{b_{\ell,\mathrm f}}^s g_1^S,\quad
 \psi_{2,\mathrm f}(s)=-\int_{b_{\mathrm f}}^s g_2^S,\qquad
 \Upsilon_{1,\mathrm f}(z)=\int_{b_{\ell,\mathrm f}}^z g_1^B,\quad
 \Upsilon_{2,\mathrm f}(z)=\int_{a_{\mathrm f}}^z g_2^B.
 \label{eq:2fam-potentials}
\end{equation}

\begin{proposition}[Separate global responses at the reference]
\label{prop:2fam-separate-responses}
Let $(Z^*,Y^*)$ be the minimizing body pair from Proposition~\ref{prop:2fam-optimum}, and set $\beta_{\mathrm f}=r_{\mathrm{fam}}$.
For the normalized welfare quantities in \eqref{eq:2fam-limiting-welfare}, the
probability $\pi_{\mathrm f}$ defined by \eqref{eq:2fam-price-density} satisfies
\begin{align*}
 \mathcal W_{\mathrm f,\pi_{\mathrm f}}(Z^*,Y)
     &\ge\beta_{\mathrm f}\mathcal O_{\mathrm f}(Z^*,Y)
       &&\text{for every finite-mean }0\le Y_1\le Y_2,\\
 \mathcal W_{\mathrm f,\pi_{\mathrm f}}(Z,Y^*)
     &\ge\beta_{\mathrm f}\mathcal O_{\mathrm f}(Z,Y^*)
       &&\text{for every finite-mean }Z_1\ge Z_2\ge0.
\end{align*}
For every $z\ge0$,
\begin{equation}
 \Upsilon_{2,\mathrm f}(z)+\inf_{w\ge z}\Upsilon_{1,\mathrm f}(w)\ge0,
 \qquad
 \Upsilon_{2,\mathrm f}(z)+\inf_{w\ge z}\Upsilon_{1,\mathrm f}(w)=0
       \quad\Longleftrightarrow\quad z=a_{\mathrm f}.
 \label{eq:2fam-second-gap}
\end{equation}
Each welfare inequality fixes the entire opposite reference side and
keeps the escaping buyer first moments equal to one.
\end{proposition}
\begin{proof}
We prove nonnegativity of the two ordered sums of measure potentials;
integrating them then gives the separate comparisons.

\textit{Body and low values.} The contact and outer stationary equations
give $g_1^S=0$ on $[a_{\mathrm f},b_{r,\mathrm f}]$,
$g_1^B=0$ on $[b_{\ell,\mathrm f},b_{\mathrm f}]$, and
$f_{\pi,\mathrm f}(b_{\mathrm f}-)=(1-\beta_{\mathrm f})$.
On the initial contact,
$g_1^B=(\beta_{\mathrm f}N_{\mathrm f}C_{\mathrm f}/\xi_\ell)
(\mathcal P^2-\mathcal P_\ell^2)<0$; on the final contact,
$(g_1^S)'=f_{\pi,\mathrm f}'\bar L_1-2f_{\pi,\mathrm f}H_1<0$.
Thus $\psi_{1,\mathrm f}\ge0$ and $\Upsilon_{1,\mathrm f}\ge0$ on the body interval, with
their respective zero intervals just identified.
The second-unit formulas are
\[
 g_2^S=f_{\pi,\mathrm f}-(1-\beta_{\mathrm f})>0,\qquad
 g_2^B=m_{\mathrm f}
 \left\{\pi_{a,\mathrm f}+\int_{a_{\mathrm f}}^sf_{\pi,\mathrm f}
                          +sf_{\pi,\mathrm f}(s)-\beta_{\mathrm f}\right\}.
\]
The positive second-unit kernel gives
$\Upsilon_{2,\mathrm f}(s)>0$ for $a_{\mathrm f}<s\le b_{\mathrm f}$, while
$\psi_{2,\mathrm f}(s)>0$ for $a_{\mathrm f}\le s<b_{\mathrm f}$.
Below $a_{\mathrm f}$, direct integration gives
\begin{align*}
 \Upsilon_{2,\mathrm f}(z)&=m_{\mathrm f}(a_{\mathrm f}-z)^2
  \{\lambda_{0,\mathrm f}+\tfrac12\lambda_{1,\mathrm f}(z+2a_{\mathrm f})\}>0,\\
 \psi_{1,\mathrm f}(s)&=(a_{\mathrm f}-s)^2
  \{r_{a,\mathrm f}-\lambda_{1,\mathrm f}/(2p_{\mathrm f})
                          -\lambda_{1,\mathrm f}(a_{\mathrm f}-s)/2\}.
\end{align*}
Here $\Upsilon_{1,\mathrm f}(z)=\Upsilon_{2,\mathrm f}(z)+\Upsilon_{1,\mathrm f}(a_{\mathrm f})$.
The strict coefficient inequalities make $\psi_{1,\mathrm f}$ nonincreasing on this low
interval and $-g_1^S-g_2^S>0$ there.
The normalization at zero follows from the mass variation:
\[
 \frac{\psi_{1,\mathrm f}(0)+\psi_{2,\mathrm f}(0)}{\beta_{\mathrm f}}
      =K_{\mathrm f}-(1+t_{\mathrm f})\partial_{t_{\mathrm f}}K_{\mathrm f}=0.
\]
Indeed, varying the common initial mass contributes
$m_{\mathrm f}'[\psi_{1,\mathrm f}(0)+\psi_{2,\mathrm f}(0)]$; the other reference supports
lie on zero potentials and moving endpoints have zero first-order
contribution. Differentiating the normalized comparison gives this normalization identity. Consequently, if $y_1\le y_2<a_{\mathrm f}$,
$\psi_{1,\mathrm f}(y_1)+\psi_{2,\mathrm f}(y_2)\ge\psi_{1,\mathrm f}(y_2)+\psi_{2,\mathrm f}(y_2)\ge0$.

\textit{The finite price tail.} Put
$A_{1b}=N_{\mathrm f}\xi_0$, $A_{2b}=m_{\mathrm f}b_{\mathrm f}$,
\[
 \sigma_0=m_{\mathrm f}\{b_{\mathrm f}(1-\pi_{T,\mathrm f}-\beta_{\mathrm f})
                  +a_{\mathrm f}(\beta_{\mathrm f}-\pi_{a,\mathrm f})\},
 \qquad \sigma_1=m_{\mathrm f}(1-\pi_{T,\mathrm f})
                   +(1-\beta_{\mathrm f})A_{2b}-\beta_{\mathrm f}.
\]
For $0\le w\le\overline p_{\mathrm f}-b_{\mathrm f}$ the buyer potentials are
$\Upsilon_{1,\mathrm f}(b_{\mathrm f}+w)=(1-\beta_{\mathrm f})w^2$ and
$\Upsilon_{2,\mathrm f}(b_{\mathrm f}+w)=\sigma_0+\sigma_1w+(1-\beta_{\mathrm f})w^2$.
Their sum is positive because
$\sigma_0-\sigma_1^2/(8(1-\beta_{\mathrm f}))>0$.
Above $\overline p_{\mathrm f}$ their slopes are opposite, with first
slope $\pi_{T,\mathrm f}-(1-\beta_{\mathrm f})A_{1b}>0$; their sum is
constant and positive. The slope balance is
$-\beta_{\mathrm f}N_{\mathrm f}K_{\mathrm f}=0$.
The seller potentials satisfy, for every $w\ge0$,
\[
 \psi_{i,\mathrm f}(b_{\mathrm f}+w)
   =\psi_{i,\mathrm f}(b_{\mathrm f})+(1-\beta_{\mathrm f})w
                       -\min\{(1-\beta_{\mathrm f})w,\pi_{T,\mathrm f}\}\ge0.
\]
The low and body signs and the increasing first buyer potential above
$b_{\mathrm f}$ prove \eqref{eq:2fam-second-gap} and both ordered
potential inequalities. The polynomial identities and interval bounds for these coefficients,
including the positive second-unit kernel, are supplied in the
supplementary material.

\textit{Integration against the changed side.} At the reference, the
seller potential sum has expectation zero by its zero-support
conditions and cancellation at the common initial atom; the buyer
potential sum also has expectation zero. With either opposite side
fixed, the change in $\mathcal W_{\mathrm f,\pi_{\mathrm f}}
-\beta_{\mathrm f}\mathcal O_{\mathrm f}$ is exactly the expectation
of its corresponding ordered potential sum. The potentials grow at
most linearly, so bounded-body comparison extends to finite first
moments. Hence each such difference is nonnegative.
\end{proof}

%% file: two_unit_computation.tex
\clearpage
\section{The two-unit instance}\label{app:two-finite}

We approximate the three-piece curve of Proposition~\ref{prop:2fam-optimum}
by a rational concave polygon with \(128\) interior segments.
Its initial parameters, rounded to six decimal places, are
\[
\begin{aligned}
p_{\mathrm w}&\approx0.671601,&
\xi_{\ell,\mathrm w}&\approx0.258131,\\
\xi_{r,\mathrm w}&\approx0.429836,&
\xi_{0,\mathrm w}&\approx0.675567.
\end{aligned}
\]
For consecutive polygon vertices \((X_j,P_j)\), set
\(h_j=(P_{j+1}-P_j)/(X_{j+1}-X_j)\).
The physical-value increment is
\((X_{j+1}-X_j)/(P_jP_{j+1})\);
the first buyer survival is \(h_j\), and the first seller CDF is
\((2-m_{\mathrm w})(P_j-X_jh_j)\).
The finite marginal laws are obtained by rounding values and cumulative
probabilities to fixed rational grids, specified with the full vertices
and weights in the supplementary material. The table rounds body
parameters to six decimal places.

\begin{center}
\small
\begin{tabular}{@{}ll@{}}
\toprule
Quantity & Value\\
\midrule
Second buyer body & \(a_{\mathrm w}\approx0.828691\), mass \(1-\eta_{\mathrm w}\)\\
Common buyer tail & value \(10^{10}\), mass \(\eta_{\mathrm w}=10^{-10}\)\\
Second seller, lower atom & \(10^{-12}\), mass \(m_{\mathrm w}\approx0.482351\)\\
Second seller, upper atom & \(\approx1.525881\), mass \(1-m_{\mathrm w}\)\\
Marginal support sizes & \(131,2\) for the buyer; \(131,2\) for the seller\\
\bottomrule
\end{tabular}
\end{center}

Common-quantile coupling within each side gives \(131\) ordered vector
types on each side. We draw the buyer and seller vectors independently.
The seller displacement of \(10^{-12}\) separates coincident body
atoms. Every value is strictly positive.

\begin{proof}[Proof of Theorem~\ref{thm:two-separation}]
We sum welfare exactly at every valuation event of the rational instance.
Each type pair contributes its marginal gain on the closed price
interval \([S_i,B_i]\). The inclusive event value contains all trades
available in an adjacent open interval; welfare is constant between
events and equals initial welfare outside the support.
Thus the \(264\) event values cover every real price.
Their maximum divided by efficient welfare is a rational number
strictly below \(0.7290804\).

For the symmetric instance, we use the supplementary rational law
for the seller vector and its coordinate reversal for the independent
buyer vector. Exact evaluation of its \(30\) valuation events gives
a ratio below \(0.83693\).
Together with Theorem~\ref{thm:exact}, the two finite instances give
\[
r_2^*<0.7290804<0.738024<\beta_*,
\qquad r_{2,\mathrm{sym}}^*<0.83693.
\]
\end{proof}

%% file: paper.bbl
\newcommand{\etalchar}[1]{$^{#1}$}
\begin{thebibliography}{CBGdK{\etalchar{+}}17}

\bibitem[ABD03]{abd03}
Giovanni Alberti, Guy Bouchitt{\'e}, and Gianni {Dal Maso}.
\newblock The calibration method for the {Mumford--Shah} functional and
  free-discontinuity problems.
\newblock {\em Calculus of Variations and Partial Differential Equations},
  16(3):299--333, 2003.
\newblock \url{https://cvgmt.sns.it/paper/1428/}.

\bibitem[BCGZ18]{bcgz18}
Moshe Babaioff, Yang Cai, Yannai~A. Gonczarowski, and Mingfei Zhao.
\newblock The best of both worlds: Asymptotically efficient mechanisms with a
  guarantee on the expected gains-from-trade.
\newblock In {\em Proceedings of the ACM Conference on Economics and
  Computation ({EC})}, page 373, 2018.
\newblock
  \href{https://doi.org/10.1145/3219166.3219203}{doi:\nolinkurl{10.1145/3219166.3219203}}.
  Extended version: \href{https://arxiv.org/abs/1802.08023}{arXiv:1802.08023}.

\bibitem[BCWZ17]{bcwz17}
Johannes Brustle, Yang Cai, Fa~Wu, and Mingfei Zhao.
\newblock Approximating gains from trade in two-sided markets via simple
  mechanisms.
\newblock In {\em Proceedings of the ACM Conference on Economics and
  Computation ({EC})}, pages 589--590, 2017.
\newblock \href{https://arxiv.org/abs/1706.04637}{arXiv:1706.04637}.

\bibitem[BD21]{bd21}
Liad Blumrosen and Shahar Dobzinski.
\newblock ({Almost}) efficient mechanisms for bilateral trading.
\newblock {\em Games and Economic Behavior}, 130:369--383, 2021.
\newblock
  \href{https://doi.org/10.1016/j.geb.2021.08.011}{doi:\nolinkurl{10.1016/j.geb.2021.08.011}}.

\bibitem[BF16]{bf16}
Guy Bouchitt{\'e} and Ilaria Fragal{\`a}.
\newblock A duality theory for non-convex problems in the calculus of
  variations.
\newblock \href{https://arxiv.org/abs/1607.02878}{arXiv:1607.02878}, 2016.

\bibitem[BGG20]{bgg20}
Moshe Babaioff, Kira Goldner, and Yannai~A. Gonczarowski.
\newblock {Bulow--Klemperer}-style results for welfare maximization in
  two-sided markets.
\newblock In {\em Proceedings of the ACM-SIAM Symposium on Discrete Algorithms
  ({SODA})}, pages 2452--2471, 2020.
\newblock \href{https://arxiv.org/abs/1903.06696}{arXiv:1903.06696}.

\bibitem[BM16]{bm16}
Liad Blumrosen and Yehonatan Mizrahi.
\newblock Approximating gains-from-trade in bilateral trading.
\newblock In {\em Web and Internet Economics ({WINE})}, pages 400--413, 2016.
\newblock
  \href{https://doi.org/10.1007/978-3-662-54110-4_28}{doi:\nolinkurl{10.1007/978-3-662-54110-4_28}}.

\bibitem[CBdKLT16]{cbklt16}
Riccardo Colini-Baldeschi, Bart de~Keijzer, Stefano Leonardi, and Stefano
  Turchetta.
\newblock Approximately efficient double auctions with strong budget balance.
\newblock In {\em Proceedings of the ACM-SIAM Symposium on Discrete Algorithms
  ({SODA})}, pages 1424--1443, 2016.

\bibitem[CBGdK{\etalchar{+}}17]{cbg17}
Riccardo Colini-Baldeschi, Paul~W. Goldberg, Bart de~Keijzer, Stefano Leonardi,
  and Stefano Turchetta.
\newblock Fixed price approximability of the optimal gain from trade.
\newblock In {\em Web and Internet Economics ({WINE})}, pages 146--160, 2017.
\newblock Extended version:
  \href{https://arxiv.org/abs/1710.08394}{arXiv:1710.08394}.

\bibitem[{\v{C}}P16]{cp16}
Jernej {\v{C}}opi{\v{c}} and Clara Ponsat{\'i}.
\newblock Optimal robust bilateral trade: risk neutrality.
\newblock {\em Journal of Economic Theory}, 163:276--287, 2016.
\newblock
  \href{https://doi.org/10.1016/j.jet.2015.11.008}{doi:\nolinkurl{10.1016/j.jet.2015.11.008}}.

\bibitem[CW23]{cw23}
Yang Cai and Jinzhao Wu.
\newblock On the optimal fixed-price mechanism in bilateral trade.
\newblock In {\em Proceedings of the ACM Symposium on Theory of Computing
  ({STOC})}, pages 737--750, 2023.
\newblock
  \href{https://doi.org/10.1145/3564246.3585171}{doi:\nolinkurl{10.1145/3564246.3585171}}.
  Extended version:
  \href{https://arxiv.org/abs/2301.05167v2}{arXiv:2301.05167v2}.

\bibitem[DFL{\etalchar{+}}21]{dfflr21}
Paul D{\"u}tting, Federico Fusco, Philip Lazos, Stefano Leonardi, and Rebecca
  Reiffenh{\"a}user.
\newblock Efficient two-sided markets with limited information.
\newblock In {\em Proceedings of the ACM Symposium on Theory of Computing
  ({STOC})}, pages 1452--1465, 2021.
\newblock
  \href{https://doi.org/10.1145/3406325.3451076}{doi:\nolinkurl{10.1145/3406325.3451076}}.

\bibitem[DMS{\etalchar{+}}25]{dmsww25}
Yuan Deng, Jieming Mao, Balasubramanian Sivan, Kangning Wang, and Jinzhao Wu.
\newblock Approximately efficient bilateral trade with samples.
\newblock \href{https://arxiv.org/abs/2502.13122}{arXiv:2502.13122}, 2025.

\bibitem[DMSW22]{dmsw22}
Yuan Deng, Jieming Mao, Balasubramanian Sivan, and Kangning Wang.
\newblock Approximately efficient bilateral trade.
\newblock In {\em Proceedings of the ACM Symposium on Theory of Computing
  ({STOC})}, pages 718--721, 2022.
\newblock
  \href{https://doi.org/10.1145/3519935.3520054}{doi:\nolinkurl{10.1145/3519935.3520054}}.

\bibitem[DS24]{ds24}
Shahar Dobzinski and Ariel Shaulker.
\newblock Bilateral trade with correlated values.
\newblock In {\em Proceedings of the ACM Symposium on Theory of Computing
  ({STOC})}, pages 237--246, 2024.
\newblock
  \href{https://doi.org/10.1145/3618260.3649659}{doi:\nolinkurl{10.1145/3618260.3649659}}.

\bibitem[DS26]{ds26}
Shahar Dobzinski and Ariel Shaulker.
\newblock Welfare maximization in bilateral trade: improved approximation
  guarantees beyond the fixed price barrier.
\newblock \href{https://arxiv.org/abs/2606.04890v1}{arXiv:2606.04890v1}, 2026.

\bibitem[Fei22]{fei22}
Yumou Fei.
\newblock Improved approximation to first-best gains-from-trade.
\newblock In {\em Web and Internet Economics ({WINE})}, volume 13778 of {\em
  Lecture Notes in Computer Science}, pages 204--218, 2022.
\newblock
  \href{https://doi.org/10.1007/978-3-031-22832-2_12}{doi:\nolinkurl{10.1007/978-3-031-22832-2_12}}.
  Extended version: \href{https://arxiv.org/abs/2205.00140}{arXiv:2205.00140}.

\bibitem[GdK26]{gdk26}
Giordano Giambartolomei and Bart de~Keijzer.
\newblock On the approximation ratio of optimal fixed-price mechanisms for
  single and multi-unit bilateral trade.
\newblock {\em Proceedings of the AAAI Conference on Artificial Intelligence},
  40(20):16946--16953, 2026.
\newblock
  \href{https://doi.org/10.1609/aaai.v40i20.38741}{doi:\nolinkurl{10.1609/aaai.v40i20.38741}}.

\bibitem[GGdK{\etalchar{+}}19]{ggkls19}
Matthias Gerstgrasser, Paul~W. Goldberg, Bart de~Keijzer, Philip Lazos, and
  Alexander Skopalik.
\newblock Multi-unit bilateral trade.
\newblock {\em Proceedings of the AAAI Conference on Artificial Intelligence},
  33(01):1973--1980, 2019.
\newblock
  \href{https://doi.org/10.1609/aaai.v33i01.33011973}{doi:\nolinkurl{10.1609/aaai.v33i01.33011973}}.

\bibitem[HR87]{hr87}
Kathleen~M. Hagerty and William~P. Rogerson.
\newblock Robust trading mechanisms.
\newblock {\em Journal of Economic Theory}, 42(1):94--107, 1987.
\newblock
  \href{https://doi.org/10.1016/0022-0531(87)90104-9}{doi:\nolinkurl{10.1016/0022-0531(87)90104-9}}.

\bibitem[JGC26]{jgc26}
Tao Jiang, Minbo Gao, and Shaowei Cai.
\newblock The exact approximation ratio of the optimal fixed-price mechanism in
  bilateral trade.
\newblock \href{https://arxiv.org/abs/2609.27878v1}{arXiv:2609.27878v1}, 2026.

\bibitem[Jo26]{jo26}
Sunghyeon Jo.
\newblock Tighter bounds for the random-offerer mechanism in bilateral trade.
\newblock \href{https://arxiv.org/abs/2607.13959}{arXiv:2607.13959}, 2026.

\bibitem[KPV22]{kpv22}
Zi~Yang Kang, Francisco Pernice, and Jan Vondr{\'a}k.
\newblock Fixed-price approximations in bilateral trade.
\newblock In {\em Proceedings of the ACM-SIAM Symposium on Discrete Algorithms
  ({SODA})}, pages 2964--2985, 2022.
\newblock \href{https://arxiv.org/abs/2107.14327}{arXiv:2107.14327}.

\bibitem[KS00]{ks00}
David Kinderlehrer and Guido Stampacchia.
\newblock {\em An Introduction to Variational Inequalities and Their
  Applications}.
\newblock SIAM, 2000.
\newblock Reprint of the 1980 edition.
  \href{https://doi.org/10.1137/1.9780898719451}{doi:\nolinkurl{10.1137/1.9780898719451}}.

\bibitem[KV19]{kv19}
Zi~Yang Kang and Jan Vondr{\'a}k.
\newblock Fixed-price approximations to optimal efficiency in bilateral trade.
\newblock SSRN working paper,
  \href{https://ssrn.com/abstract=3460336}{ssrn.com/abstract=3460336}, 2019.
\newblock Earlier version: Strategy-proof approximations of optimal efficiency
  in bilateral trade, 2018.

\bibitem[LQRW26]{lqrw26}
Zhengyang Liu, Ying Qin, Zeyu Ren, and Zihe Wang.
\newblock Second-best bilateral trade is $1/2$ efficient.
\newblock \href{https://arxiv.org/abs/2606.03849}{arXiv:2606.03849}, 2026.

\bibitem[LRW23]{lrw23}
Zhengyang Liu, Zeyu Ren, and Zihe Wang.
\newblock Improved approximation ratios of fixed-price mechanisms in bilateral
  trades.
\newblock In {\em Proceedings of the ACM Symposium on Theory of Computing
  ({STOC})}, pages 751--760, 2023.
\newblock
  \href{https://doi.org/10.1145/3564246.3585160}{doi:\nolinkurl{10.1145/3564246.3585160}}.
  Extended version:
  \href{https://arxiv.org/abs/2303.15711v1}{arXiv:2303.15711v1}.

\bibitem[McA92]{mcafee92}
R.~Preston McAfee.
\newblock A dominant strategy double auction.
\newblock {\em Journal of Economic Theory}, 56(2):434--450, 1992.
\newblock
  \href{https://doi.org/10.1016/0022-0531(92)90091-U}{doi:\nolinkurl{10.1016/0022-0531(92)90091-U}}.

\bibitem[MS83]{ms83}
Roger~B. Myerson and Mark~A. Satterthwaite.
\newblock Efficient mechanisms for bilateral trading.
\newblock {\em Journal of Economic Theory}, 29(2):265--281, 1983.
\newblock
  \href{https://doi.org/10.1016/0022-0531(83)90048-0}{doi:\nolinkurl{10.1016/0022-0531(83)90048-0}}.

\bibitem[Roc70]{rock70}
R.~Tyrrell Rockafellar.
\newblock {\em Convex Analysis}.
\newblock Princeton University Press, 1970.
\newblock
  \href{https://doi.org/10.1515/9781400873173}{doi:\nolinkurl{10.1515/9781400873173}}.

\bibitem[Sio58]{sion58}
Maurice Sion.
\newblock On general minimax theorems.
\newblock {\em Pacific Journal of Mathematics}, 8(1):171--176, 1958.
\newblock \url{https://msp.org/pjm/1958/8-1/pjm-v8-n1-p14-p.pdf}.

\bibitem[Zha26]{zhang26}
Wanchang Zhang.
\newblock Generalized double auctions: maximizing welfare guarantees in
  bilateral trade.
\newblock {\em American Economic Review}, 2026.
\newblock Forthcoming.

\end{thebibliography}
